\documentclass{toc}

\tocdetails{
  title = {Shrinkage under Random Projections,\\
 and Cubic Formula Lower Bounds for $\AC^0$},
  runningtitle = {Shrinkage under Random Projections},
  plaintexttitle = {Shrinkage under Random Projections, and Cubic Formula Lower Bounds for AC0},
  number_of_pages = {48}, 
  number_of_bibitems = {60},
  number_of_figures = {0},
  conference_version = {ITCS21},
  author = {Yuval Filmus, Or Meir, and Avishay Tal},
  authorlist = {Yuval Filmus, Or Meir, Avishay Tal},
  acmclassification = {F.1.1},
  amsclassification = {68Q06},
  keywords = {circuit complexity, lower bounds, shrinkage, formulas, formula complexity, KRW conjecture}
}

\tocdetails{%
volume = 19,      %
number = 7,       %
year = 2023,
received =  {February 19, 2021},
revised = {December 28, 2021},
published = {December 5, 2023},
doi = {10.4086/toc.2023.v019a007},
}   %

\usepackage{packages}
\usepackage{aumacros}

\begin{document}

\begin{frontmatter}[classification=text]

\title{Shrinkage under Random Projections, \\
  and Cubic Formula Lower Bounds\\ for $\ACz$ \titlefootnote{An extended
abstract of this article 
appeared in the \href{https://doi.org/10.4230/LIPIcs.ITCS.2021.89}%
{Proceeding of the 12th Innovations in Theoretical
Computer Science Conference (ITCS'21)}.}}

\author[filmus]{Yuval Filmus\thanks{Taub Fellow --- supported by
			the Taub Foundations. The research was funded by ISF grant 1337/16.}}
\author[meir]{Or Meir\thanks{Partially supported by the Israel Science Foundation (grant No. 1445/16).}}
\author[tal]{Avishay Tal}

\begin{abstract}
H{\aa}stad showed that any De~Morgan formula (composed of AND, OR
and NOT gates) shrinks by a factor of~$\tilde{O}(p^{2})$ under a random
restriction that leaves each variable alive independently with probability~$p$
{[}SICOMP, 1998{]}. Using this result, he gave an $\widetilde{\Omega}(n^{3})$
formula size lower bound for the Andreev function, which, up to lower
order improvements, remains the state-of-the-art lower bound for any
explicit function.

In this %
paper,  %
we extend the shrinkage result of H{\aa}stad to hold
under a far wider family of random restrictions and their generalization
--- random projections. Based on our shrinkage results, we obtain
an $\widetilde{\Omega}(n^{3})$ formula size lower bound for an explicit
function computable in~$\ACz$. This improves upon the best known
formula size lower bounds for~$\ACz$, that were only quadratic prior
to our work. In addition, we prove that the KRW conjecture {[}Karchmer
et al., Computational Complexity 5(3/4), 1995{]} holds for inner functions
for which the unweighted quantum adversary bound is tight. In particular,
this holds for inner functions with a tight Khrapchenko bound.

Our random projections are tailor-made to the function's structure
so that the function maintains structure even under projection ---
using such projections is necessary, as standard random restrictions
simplify $\ACz$ circuits. In contrast, we show that any De~Morgan
formula shrinks by a quadratic factor under our random projections,
allowing us to prove the cubic lower bound.

Our proof techniques build on %
H{\aa}stad's proof   %
for the simpler
case of balanced formulas. This allows for a significantly simpler
proof at the cost of slightly worse parameters. As such, when specialized
to the case of $p$-random restrictions, our proof can be used as
an exposition of H{\aa}stad's result. 
\end{abstract}

%

\end{frontmatter}

\section{Introduction}

\subsection{Background}

Is there an efficiently solvable computational task
that does not admit an efficient parallel algorithm?  %
A formalization of this question asks,  %
is $\P\not\subseteq\NC^{1}$? The answer
is still unknown. The question can be rephrased as follows (ignoring issues of uniformity): 
is there a
Boolean  %
function in~$\P$ that does not have a (De~Morgan) formula of
polynomial size?   %
(We define De~Morgan formulas formally in \expref{Section}{sec:prel}; informally, these are formulas composed of AND, OR and NOT gates.)

The history of formula lower bounds for functions in~$\P$ goes back
to the 1960s, with the seminal result of Subbotovskaya~\cite{S61}
that introduced the technique of random restrictions. Subbotovskaya
showed that the Parity function on $n$ variables requires formulas
of size at least $\Omega(n^{1.5})$. Khrapchenko~\cite{K72}, using
a different proof technique, showed that in fact the Parity function
on $n$ variables requires formulas of size $\Theta(n^{2})$. Later,
Andreev~\cite{A87} came up with a new explicit function (now known
as the Andreev function) for which he was able to obtain an $\Omega(n^{2.5})$
size lower bound. This lower bound was subsequently improved
in  %
\cite{IN93,PZ93,H98,T14}
to $n^{3-o(1)}$.

The line of work initiated by Subbotovskaya and Andreev relies on
the \emph{shrinkage} of formulas under $p$-random restrictions. A
$p$-random restriction is a randomly chosen partial assignment to
the inputs of a function. Set a parameter $p\in(0,1)$. We fix each
variable independently with probability $1-p$ to a uniformly random
bit, and we keep the variable alive with probability $p$. Under such
a restriction, formulas shrink (in expectation) by a factor more significant
than $p$. Subbotovskaya showed that formulas shrink to
at most $p^{1.5}$ times their original size, whereas subsequent
work~\cite{PZ93,IN93}  %
improved the bound to $p^{1.55}$ and
then to  %
$p^{1.63}$.
Finally, H{\aa}stad~\cite{H98} showed that the shrinkage
exponent of formulas is $2$, or in other words, that 
formulas shrink by a factor of $p^{2-o(1)}$ under $p$-random restrictions.
Tal~\cite{T14} improved the shrinkage factor to $O(p^{2})$ ---
obtaining a tight result, as exhibited by the Parity function.

In a nutshell, shrinkage results are useful %
for  %
proving lower bounds
as long as the explicit function being analyzed maintains structure
under such restrictions and does not trivialize. For example, the
Parity function does not become constant as long as at least one variable
remains alive. Thus any formula $F$ that computes Parity must be
of at least quadratic size, or else the formula $F$ under restriction,
keeping each variable alive with probability $100/n$, would likely
become a constant function, whereas Parity would not. Andreev's idea
is similar, though he manages to construct a function such that under
a random restriction keeping only $\Theta(\log n)$ of the variables,
the formula size should be at least $\tilde{\Omega}(n)$ (in expectation).
This ultimately gives the nearly cubic lower bound.

\subsubsection{The KRW Conjecture}  %

Despite much effort, proving $\P\not\subseteq\NC^{1}$, and even just
breaking the \textsf{cubic barrier} in formula lower bounds, have
remained a challenge for more than two decades. An approach to solve
the $\P$ versus $\NC^{1}$ problem was suggested by Karchmer, Raz
and Wigderson \cite{KRW95}. They conjectured that when composing
two Boolean functions, $f$ and $g$, the formula size of the resulting
function, $f\d g$, is (roughly) the product of the formula sizes
of $f$ and $g$.\footnote{More precisely, the original KRW conjecture \cite{KRW95} concerns
depth complexity rather than formula complexity. The variant of the
conjecture for formula complexity, which is discussed above, was posed
in \cite{GMWW17}. } We will refer to this conjecture as the \textsf{``KRW conjecture''}.
Under the KRW conjecture (and even under weaker variants of it), \cite{KRW95}
constructed a function in $\P$ with no polynomial-size formulas.
It remains a major open challenge to settle the KRW conjecture.

A few special cases of the KRW conjecture are known to be true. The
conjecture holds when either $f$ or $g$ is the AND or the OR function.
H{\aa}stad's result~\cite{H98} and its improvement~\cite{T14}
show that the conjecture holds when the inner function~$g$ is the
Parity function and the outer function~$f$ is any function. This
gives an alternative explanation to the $n^{3-o(1)}$ lower bound
for the Andreev function. Indeed, the Andreev function is at least
as hard as the composition of a maximally hard function $f$ on $\log n$
bits and $g=\parity_{n/\log n}$, where the formula size of $f$ is
$\tilde{\Omega}(n)$ and the formula size of $\parity_{n/\log n}$
is $\Theta(n^{2}/\log^{2}n)$. Since the KRW conjecture holds for
this special case, the formula size of the Andreev function is at
least $\tilde{\Omega}(n^{3})$. In other words, the state-of-the-art
formula size lower bounds for explicit functions follow from a special
case of the KRW conjecture --- the case in which $g$ is the Parity
function. Moreover, this special case follows from the shrinkage offormulas under $p$-random restrictions.

\subsubsection{Bottom-up versus top-down techniques}

Whereas random restrictions are a ``bottom-up'' proof technique
\cite{HJP95}, a different line of work suggested a ``top-down"
approach using the language of communication complexity. The connection
between formula size and communication complexity was introduced in
the seminal %
paper by %
Karchmer and Wigderson \cite{KW90}. They defined
for any Boolean function $f$ a two-party communication problem $\KW_{f}$:
Alice gets an input $x$ such that $f(x)=1$, and Bob gets an input
$y$ such that $f(y)=0$. Their goal is to identify a coordinate $i$
on which $x_{i}\neq y_{i}$, while minimizing their communication.
It turns out that there is a one-to-one correspondence between any
protocol tree solving $\KW_{f}$ and any formula computing the function
$f$. Since protocols naturally traverse the tree from root to leaf,
proving lower bounds on their size or depth is done usually in a top-down
fashion. This framework has proven to be very useful in proving formula
lower bounds in the monotone setting (see, \eg, \cite{KW90,GH92,RW92,KRW95,RM99,GP18,PR17})
and in studying the KRW conjecture (see, \eg, \cite{KRW95,EIRS01,HW93,GMWW17,DM18,KM18,M20,RMNPR20,MS21}).
Moreover,
in a recent paper~\cite{DM18}, Dinur and Meir were able to %
reprove H{\aa}stad's cubic lower bound using the framework of Karchmer
and Wigderson. As Dinur and Meir's proof showed that top-down techniques
can replicate H{\aa}stad's cubic lower bound, a natural question
(which motivated this project) arose: 
\begin{quote}
\emph{Are top-down techniques superior to bottom-up techniques?} 
\end{quote}
Towards that, we focused on a candidate problem: prove a cubic lower
bound for an explicit function in $\ACz$.\footnote{Recall that $\ACz$ is the class of functions computed by constant
depth polynomial size circuits composed of AND and OR gates of unbounded
fan-in, with variables or their negation at the leaves.} Based on the 
work of Dinur and Meir~\cite{DM18}, we suspected that
such a lower bound could be achieved using top-down techniques. We
were also \emph{certain} that the problem cannot be solved using the
random restriction technique. Indeed, in order to prove a lower bound
on a function~$f$ using random restrictions, one should argue that
$f$~remains hard under a random restriction, however, it is well-known
that functions in~$\ACz$ trivialize under $p$-random restrictions
\cite{A83,FSS84,Y85,H86}. Based on this intuition, surely random
restrictions cannot show that a function in $\ACz$ requires cubic
size. Our intuition turned out to be false.

\subsection{Our results}

In this %
article,  %
we construct an explicit function in $\ACz$ which requiresformulas of size $n^{3-o(1)}$. Surprisingly, our proof
is conducted via the bottom-up technique of random projections, which
is a generalization of random restrictions (more details below).
\begin{theorem}
\label{thm:main}There exists a family of Boolean functions $h_{n}\colon\B^{n}\to\B$
for $n\in\N$ such that 
\begin{enumerate}
\item $h_{n}$ can be computed by uniform depth-$4$ unbounded fan-in formulas
of size $O(n^{3})$. 
\item The formula size of $h_{n}$ is at least $n^{3-o(1)}$.
\end{enumerate}
\end{theorem}

Prior to our work, the best formula size lower bounds
for  %
an explicit
function in $\ACz$ were only quadratic \cite{N66,CKK12,JuknaBook,BM12}.

Our hard function is a variant of the Andreev function. More specifically,
recall that the Andreev function is based on the composition $f\d g$,
where $f$ is a maximally hard function and $g$ is the Parity function.
Since Parity is not in $\ACz$, we cannot take $g$ to be the Parity
function in our construction. Instead, our hard function is obtained
by replacing the Parity function with the Surjectivity function of Beame and Machmouchi~\cite{BM12}.

As in the case of the Andreev function, we establish the hardness
of our function by proving an appropriate special case of the KRW
conjecture. To this end, we introduce a generalization of the unweighted
adversary method~\cite{A02}, called the \emph{soft-adversary bound}, which we
denote~$\Am(g)$ (see~\expref{Section}{sec:adversary}). We prove the KRW
conjecture for the special case in which the outer function~$f$
is any function, and $g$~is a function whose formula complexity
is bounded tightly by the soft-adversary bound. We then obtain \expref{Theorem}{thm:main}
by applying this version of the KRW conjecture to the case where $g$~is
the Surjectivity function. We note that our KRW result holds in particular
for functions~$g$ with a large Khrapchenko bound, which in turn
implies the known lower bounds in the cases where $g$ is the Parity
function~\cite{H98} and the Majority function~\cite{GTN19}.

Our proof of the special case of the KRW conjecture follows the methodology
of H{\aa}stad~\cite{H98},   %
who proved the special case in which
$g$ is Parity on $m$ variables. H{\aa}stad proved that formulas
shrink by a factor of (roughly) $p^{2}$ under $p$-random restrictions.
Choosing $p=1/m$ shrinks a formula for $f\d g$ by a factor of roughly
$m^{2}$, which coincides with the formula complexity of $g$. On
the other hand, on average each copy of $g$ simplifies to a single
input variable, and so $f\d g$ simplifies to $f$. This shows that
$L(f\d g)\gtrsim L(f)\cdot L(g)$.

Our main technical contribution is a new shrinkage theorem that works
in a far wider range of scenarios than just $p$-random restrictions:
A \emph{random projection }is\emph{ }a generalization of a random
restriction in which each of the input variables~$x_{1},\ldots,x_{n}$
may either be fixed to a constant or be replaced with a literal from
the set $y_{1},\ldots,y_{m},\overline{y_{1}},\ldots,\overline{y_{m}}$
where $y_{1},\ldots,y_{m}$ are new variables. Given a function $g$
with soft-adversary bound $\Am(g)$, we construct a random projection
which, on the one hand, shrinks formulas by a factor of
$\Am(g)$, and on the other hand, simplifies $f\d g$ to $f$. We
thus show that $L(f\d g)\gtrsim L(f)\cdot\Am(g)$, and in particular,
if $\Am(g)\approx L(g)$, then $L(f\d g)\gtrsim L(f)\cdot L(g)$,
just as in H{\aa}stad's proof. Our random projections are tailored
specifically to the structure of the function $f\d g$, ensuring that
$f\d g$ simplifies to $f$ under projection. This enables us to overcome
the aforementioned difficulty. In contrast, $p$-random restrictions
that do not respect the structure of $f\d g$ would likely result
in a restricted function that is much simpler than $f$ and in fact
would be a constant function with high probability.

Our shrinkage theorem applies more generally to two types of random
projections, which we call \emph{fixing projections} and \emph{hiding
projections}. Roughly speaking, fixing projections are random projections in which substituting
a constant for one of the variables $y_{1},\ldots,y_{m}$ results
in a projection that is much more probable. Hiding projections are
random projections in which substituting a constant for a variable~$y_{j}$
hides which of the input variables $x_{1},\ldots,x_{n}$ were mapped
to $\{y_{j},\overline{y}_{j}\}$ before the substitution. We note
that our shrinkage theorem for fixing projections captures H{\aa}stad's
result for $p$-random restrictions as a special case.

The proof of our shrinkage theorem is based on H{\aa}stad's proof~\cite{H98},
but also simplifies it. In particular, we take the simpler argument
that H{\aa}stad uses for the special case of completely balanced
trees, and adapt it to the general case. As such, our proof avoids
a complicated case analysis, at the cost of slightly worse bounds.
Using our bounds, it is nevertheless easy to obtain the $n^{3-o(1)}$
lower bound for the Andreev function. Therefore, one can see the specialization
of our shrinkage result to $p$-random restrictions as an exposition
of H{\aa}stad's cubic lower bound.

\subsubsection{An example: our techniques when specialized to $f\protect\d\protect\majority_{m}$}   %

To illustrate our choice of random projections, we present its instantiation
to the special case of $f\d g$, where $f\colon\B^{k}\to\B$ is non-constant
and $g=\majority_{m}$ for some odd integer $m$. In this case, the
input variables to $f\d g$ are composed of $k$ disjoint blocks,
$B_{1},\ldots,B_{k}$, each containing $m$ variables. We use the
random projection that for each block $B_{i}=\{x_{m(i-1)+1},\ldots,x_{mi}\}$,
picks one variable in the block $B_{i}$ uniformly at random, projects
this variable to the new variable $y_{i}$, and fixes the rest of
the variables in the block in a balanced way so that the number of
zeros and ones in the block is equal (\ie, we have exactly $(m-1)/2$
zeros and $(m-1)/2$ ones). It is not hard to see that under this
choice, $f\d g$ simplifies to $f$. On the other hand, we show that
this choice of random projections shrinks the formula complexity by
a factor of $\approx1/m^{2}$. Combining the two together, we get
that $L(f\d\majority_{m})\gtrsim L(f)\cdot m^{2}$. Note that in this
distribution of random projections, the different coordinates are
not independent of one another, and this feature allows us to maintain
structure.

\subsection{Related work}

Our technique of using tailor-made random projections was inspired
by the celebrated result of Rossman, Servedio, and Tan
\cite{HRST17}   %
that proved an average-case depth hierarchy. In fact, the idea to
use tailor-made random restrictions goes back to H{\aa}stad's thesis
\cite[Chapter~6.2]{H87}. Similar to our case, in
\cite{H87,HRST17},  %
$p$-random restrictions are too crude to separate depth $d$ from
depth $d+1$ circuits. Given a circuit $C$ of depth $d+1$, the main
challenge is to construct a distribution of random restrictions or
projections (tailored to the circuit $C$) that on the one hand maintains
structure for $C$, but on the other hand simplify any depth $d$
circuit $C'$.

\subsection{Paper outline}   %

The paper starts with brief preliminaries in~\expref{Section}{sec:prel}. Then,
in \expref{Section}{sec:Main-results}, we define the notions of fixing and hiding
projections, state the corresponding shrinkage theorems, and sketch
how they can be used to prove a cubic formula lower bound for $\ACz$.
We prove the shrinkage theorems for fixing and hiding projections in
\expref{Sections}{sec:shrinkage-fixing} 
\expref{and}{sec:shrinkage-hiding}, %
 respectively.
In \expref{Section}{sec:join} we provide a brief interlude on the join operation
for projections. The quantum adversary bound, the Khrapchenko bound,
and their relation to hiding projections are discussed in \expref{Section}{sec:hiding-from-complexity}.
Finally, \expref{Section}{sec:applications} contains a proof of \expref{Theorem}{thm:main},
as a corollary of a more general result which is a special case of
the KRW~conjecture. In the same section we also rederive the cubic
lower bound on Andreev's function, and the cubic lower bound on the
Majority-based variant considered in~\cite{GTN19}.

\section{\label{sec:prel}Preliminaries}

Throughout the paper, we use bold letters to denote random variables.
For any $n\in\N$, we denote by $\left[n\right]$ the set $\left\{ 1,\ldots,n\right\} $.
Given a bit~$\sigma\in\B$, we denote its negation by $\overline{\sigma}$.
We assume familiarity with the basic definitions of communication
complexity (see, \eg, \cite{KN_book}). All logarithms in this paper
are base~$2$. 
\begin{definition}
A \emph{(De~Morgan) formula (with bounded fan-in)} is a binary tree,
whose leaves are labeled with literals from the set $\left\{ x_{1},\overline{x}_{1},\ldots,x_{n},\overline{x}_{n}\right\} $,
and whose internal vertices are labeled as AND ($\wedge$) or OR ($\vee$)
gates. The \emph{size} of a formula~$\phi$, denoted $\s(\phi)$,
is the number of leaves in the tree. The \emph{depth} of the formula
is the depth of the tree. A \emph{formula with unbounded fan-in} is
defined similarly, but every internal vertex in the tree can have
any number of children. Unless stated explicitly otherwise, whenever
we say ``formula'' we refer to a formula with bounded fan-in. 
\end{definition}

\begin{definition}
A formula $\phi$ computes a Boolean function~$f\colon\B^{n}\to\B$
in the natural way. The \emph{formula complexity} of a Boolean function~$f\colon\B^{n}\to\B$,
denoted $L(f)$, is the size of the smallest formula that computes~$f$.
The \emph{depth complexity} of~$f$, denoted~$D(f)$, is the smallest
depth of a formula that computes $f$. For convenience, we define
the size and depth of the constant functions to be zero. 
\end{definition}

A basic property of formula complexity is that it is subadditive: 
\begin{fact}
\label{fact:subadditivity}For every two functions $f_{1},f_{2}\colon\B^{n}\to\B$
we have $L(f_{1}\wedge f_{2})\le L(f_{1})+L(f_{2})$ and $L(f_{1}\vee f_{2})\le L(f_{1})+L(f_{2})$. 
\end{fact}

The following theorem shows that every small formula can be ``balanced''
to obtain a shallow formula. 
\begin{theorem}[Formula balancing, \cite{BB94}, following \cite{S71,B74}]
\label{thm:size-depth-tradeoff}For every $s \in \N$ and $\alpha>0$, the following
holds: For every formula~$\phi$ of size~$s$, there exists an equivalent
formula~$\phi'$ of depth at most~$O(2^{\frac{1}{\alpha}}\cdot\log s)$
and size at most~$s^{1+\alpha}$ (where the constant inside the big-O
notation does not depend on the choice of~$\alpha$).
\end{theorem}

\begin{notation}
With a slight abuse of notation, we will often identify a formula~$\phi$
with the function it computes. In particular, the notation $L(\phi)$
denotes \emph{the formula complexity of the function} computed by~$\phi$,
and \emph{not the size of~$\phi$} (which is denoted by $\s(\phi)$). 
\end{notation}

\begin{notation}
Given a Boolean variable $z$, we denote by $z^{0}$ and $z^{1}$
the literals $z$ and $\overline{z}$, respectively. In other words,
$z^{b}=z\oplus b$. 
\end{notation}

\begin{notation}
Given a literal $\ell$, we define $\var(\ell)$ to be the underlying
variable, that is, $\var(z)=\var(\overline{z})=z$. 
\end{notation}

\begin{notation}
\label{notation:partition-to-rectangles}Let $\Pi$ be a deterministic communication
protocol that takes inputs from $\cA\times\cB$, and recall that the
leaves of the protocol induce a partition of $\cA\times\cB$ to combinatorial
rectangles. For every leaf~$\ell$ of~$\Pi$, we denote by $\cA_{\ell}\times\cB_{\ell}$
the combinatorial rectangle that is associated with~$\ell$. 
\end{notation}

We use the framework of Karchmer--Wigerson relations~\cite{KW90},
which relates the complexity of~$f$ to the complexity of a related
communication problem~$\KW_{f}$. 
\begin{definition}[\cite{KW90}]
Let $f\colon\B^{n}\to\B$ be a Boolean function. The \emph{Karchmer--Wigderson
relation} of~$f$, denoted $\KW_{f}$, is the following communication
problem: The inputs of Alice and Bob are strings $a\in f^{-1}(1)$
and $b\in f^{-1}(0)$, respectively, and their goal is to find a coordinate
$i\in\left[n\right]$ such that $a_{i}\ne b_{i}$. Note that such
a coordinate must exist since $f^{-1}(1)\cap f^{-1}(0)=\emptyset$
and hence $a\ne b$. 
\end{definition}

\begin{theorem}[\cite{KW90}, see also \cite{R90}]
\label{thm:KW-connection}Let $f\colon\B^{n}\to\B$. The communication
complexity of~$\KW_{f}$ is equal to~$D(f)$, and the minimal number
of leaves in a protocol that solves~$\KW_{f}$ is $L(f)$. 
\end{theorem}

We use the following two standard inequalities.
\begin{fact}[the AM--GM inequality]
If $x,y\ge 0$ then $\sqrt{x\cdot y}\le\frac{x+y}{2}$.
\end{fact}

\begin{fact}[special case of Cauchy-Schwarz inequality]
\label{fact:special-cauchy-schwarz} For every $t$~non-negative real
numbers $x_{1},\ldots,x_{t}$ we have $\sqrt{x_{1}}+\ldots+\sqrt{x_{t}}\le\sqrt{t}\cdot\sqrt{x_{1}+\ldots+x_{t}}$.   \qedhere
\end{fact}

\section{\label{sec:Main-results}Main results}

In this section we define the notions of fixing and hiding projections
and state their associated shrinkage theorems
(see \expref{Sections}{subsec:Fixing-projections}
\expref{and}{subsec:Hiding-projections}, %
respectively). We then sketch the proof of our cubic formula lower
bounds for~$\ACz$ assuming those theorems (see \expref{Section}{subsec:application-sketch}).
We start by defining projections and the relevant notation.
\begin{definition}
Let $x_{1},\ldots,x_{n}$ and $y_{1},\ldots,y_{m}$ be Boolean variables.
A \emph{projection}~$\pi$ \emph{from $x_{1},\ldots,x_{n}$ to $y_{1},\ldots,y_{m}$}
is a function from the set $\left\{ x_{1},\ldots,x_{n}\right\} $
to the set $\left\{ 0,1,y_{1},\overline{y_{1}},\ldots,y_{m},\overline{y_{m}}\right\} $.
Given such a projection~$\pi$ and a Boolean function~$f\colon\B^{n}\to\B$
over the variables $x_{1},\ldots,x_{n}$, we denote by $f|_{\pi}\colon\B^{m}\to\B$
the function obtained from~$f$ by substituting $\pi(x_{i})$ for
$x_{i}$ for each $i\in\left[n\right]$ in the natural way. Unless
stated explicitly otherwise, all projections in this section are from
\emph{$x_{1},\ldots,x_{n}$} to $y_{1},\ldots,y_{m}$, and all functions
from~$\B^{n}$ to $\B$ are over the variables $x_{1},\ldots,x_{n}$.
A \emph{random projection} is a distribution over projections. 
\end{definition}

\begin{notation}
Let $\pi$ be a projection. For every $j\in[m]$ and bit $\sigma\in\B$,
we denote by $\pi_{y_{j}\gets\sigma}$ the projection that is obtained
from $\pi$ by substituting $\sigma$ for~$y_{j}$. 
\end{notation}

\begin{notation}
With a slight abuse of notation, if a projection~$\pi$ maps all
the variables $x_{1},\ldots,x_{n}$ to constants in~$\B$, we will
sometimes treat it as a binary string in~$\B^{n}$. 
\end{notation}

\subsection{\label{subsec:Fixing-projections}Fixing projections}

Intuitively, a $q$-fixing projection is a random projection~$\p$
in which for every variable~$x_{i}$, the probability that $\p$
maps a variable~$x_{i}$ to a literal~$y_{j}^{\sigma}$ is much
smaller than the probability that $\p$~\emph{fixes~}$y_{j}$ to
a constant, \emph{regardless of the values that~$\pi$ assigns to
the other variables}. This property is essentially the minimal property
that is required in order to carry out the argument of H{\aa}stad
\cite{H98}. Formally, we define $q$-fixing projections as follows. 
\begin{definition}
\label{def:fixing-projections}Let $0\le q_{0},q_{1}\le1$. We say that
a random projection $\p$ is a \emph{$(q_{0},q_{1})$-fixing projection}
if for every projection~$\pi$, every bit $\sigma\in\B$, and every
variable~$x_{i}$, we have 
\begin{equation}
\Pr\!\left[\p(x_{i})\notin\B\text{ and }\p_{\var(\p(x_{i}))\gets\sigma}=\pi\right]\le q_{\sigma}\cdot\Pr\!\left[\p=\pi\right].\label{eq:fixing-projection}
\end{equation}
For shorthand, we say that $\p$ is a $q$-fixing projection, for
$q=\sqrt{q_{0}q_{1}}$. 
\end{definition}

If needed, one can consider without loss of generality only variables
$x_{i}$ such that $\pi(x_{i})\in\{0,1\}$, as otherwise \expref{Equation}{eq:fixing-projection}
holds trivially, with the left-hand side equaling zero.
\begin{example}
In order to get intuition for the definition of fixing projections,
let us examine how this definition applies to random restrictions.
In our terms, a restriction is a projection from $x_{1},\ldots,x_{n}$
to $x_{1},\ldots,x_{n}$ that maps every variable $x_{i}$ either
to itself or to $\B$. Suppose that $\rs$ is any distribution over
restrictions, and that $\rho$ is some fixed restriction. In this
case, the condition of being $q$-fixing can be rewritten as follows:
\[
\Pr\!\left[\rs(x_{i})=x_{i}\text{ and }\rs_{x_{i}\gets\sigma}=\rho\right]\le q_{\sigma}\cdot\Pr\!\left[\rs=\rho\right].
\]
Denote by $\rs',\rho'$ the restrictions obtained from $\rs,\rho$
by truncating $x_{i}$ (\ie, $\rs'=\rs|_{\left\{ x_{1},\ldots,x_{n}\right\} -\left\{ x_{i}\right\} }$).
Using this notation, we can rewrite the foregoing equation as 
\[
\Pr\!\left[\rs(x_{i})=x_{i}\text{ and }\rs'=\rho'\text{ and }\rs_{x_{i}\gets\sigma}(x_{i})=\rho(x_{i})\right]\le q_{\sigma}\cdot\Pr\!\left[\rs(x_{i})=\rho(x_{i})\text{ and }\rs'=\rho'\right].
\]
Now, observe that it is always the case $\rs_{x_{i}\gets\sigma}(x_{i})=\sigma$,
and therefore the probability on the left-hand side is non-zero only
if $\rho(x_{i})=\sigma$. Hence, we can restrict ourselves to the
latter case, and the foregoing equation can be rewritten again as
\[
\Pr\!\left[\rs(x_{i})=x_{i}\text{ and }\rs'=\rho'\right]\le q_{\sigma}\cdot\Pr\!\left[\rs(x_{i})=\sigma\text{ and }\rs'=\rho'\right].
\]
Finally, if we divide both sides by $\Pr\!\left[\rs'=\rho'\right]$,
we obtain the following intuitive condition: 
\[
\Pr\!\left[\rs(x_{i})=x_{i}\mid\rs'=\rho'\right]\le q_{\sigma}\cdot\Pr\!\left[\rs(x_{i})=\sigma\mid\rs'=\rho'\right].
\]
This condition informally says the following: $\rs$~is a fixing
projection if the probability of leaving $x_{i}$ unfixed is at most
$q_{\sigma}$ times the probability of fixing it to $\sigma$, and
this holds regardless of what the restriction assigns to the other
variables.

In particular, it is now easy to see that the classic random restrictions
are fixing projections. Recall that a $p$\emph{-random restriction}
fixes each variable independently with probability $1-p$ to a random
bit. Due to the independence of the different variables, the foregoing
condition simplifies to 
\[
\Pr\!\left[\rs(x_{i})=x_{i}\right]\le q_{\sigma}\cdot\Pr\!\left[\rs(x_{i})=\sigma\right],
\]
and it is easy to see that this condition is satisfied for $q_{0}=q_{1}=\frac{2p}{1-p}$. 
\end{example}

\noindent We prove the following shrinkage theorem for $q$-fixing
projections, which is analogous to the shrinkage theorem of~\cite{H98}
for random restrictions in the case of \emph{balanced} formulas. 
\begin{theorem}[Shrinkage under fixing projections]
\label{thm:shrinkage-fixing}Let $\phi$ be a formula of size~$s$ and
depth~$d$, and let $\p$ be a $q$-fixing projection. Then 
\[
\E\left[L(\phi|_{\p})\right]=O\!\left(q^{2}\cdot d^{2}\cdot s+q\cdot\sqrt{s}\right).
\]
\end{theorem}

\noindent We prove \expref{Theorem}{thm:shrinkage-fixing} in \expref{Section}{sec:shrinkage-fixing}.
It is instructive to compare our shrinkage theorem to the shrinkage
theorem by H{\aa}stad~\cite{H98} for \emph{unbalanced} formulas:
\begin{theorem}[\cite{H98}]
\label{thm:shrinkage-hastad}Let $\phi$ be a formula of size~$s$,
and let $\rs$ be a $p$-random restriction. Then
\[
\E\left[L(\phi|_{\rs})\right]=O\!\left(p^{2}\cdot\left(1+\log^{3/2}\left(\min\left\{ \frac{1}{p},s\right\} \right)\right)\cdot s+p\cdot\sqrt{s}\right).
\]
\end{theorem}

Our \expref{Theorem}{thm:shrinkage-fixing} has somewhat worse parameters compared
to~\expref{Theorem}{thm:shrinkage-hastad}: specifically, the factor of $d^{2}$
does not appear in \expref{Theorem}{thm:shrinkage-hastad}. The reason is that the
proof of~\cite{H98} uses a fairly complicated case analysis in order
to avoid losing that factor, and we chose to skip this analysis in
order to obtain a simpler proof. We did not check if the factor of~$d^{2}$
in our result can be avoided by using a similar case analysis. By
applying formula balancing (\expref{Theorem}{thm:size-depth-tradeoff}) to our shrinkage
theorem, we can obtain the following result, which is independent
of the depth of the formula. 
\begin{corollary}
\label{cor:shrinkage-without-depth-fixing}Let $f\colon\B^{n}\to\B$ be
a function with formula complexity~$s$, and let $\p$ be a $q$-fixing
projection. Then 
\[
\E\left[L(f|_{\p})\right]=q^{2}\cdot s^{1+O\bigl(\frac{1}{\sqrt{\log s}}\bigr)}+q\cdot s^{1/2+O\bigl(\frac{1}{\sqrt{\log s}}\bigr)}.
\]
\end{corollary}

\begin{proof}
By assumption, there exists a formula~$\phi$ of size~$s$ that
computes~$f$. We balance the formula~$\phi$ by applying 
\expref{Theorem}{thm:size-depth-tradeoff}
with $\alpha=\frac{1}{\sqrt{\log s}}$, and obtain a new formula~$\phi'$
that computes~$f$ and has size~$s^{1+\frac{1}{\sqrt{\log s}}}$
and depth $O(2^{\sqrt{\log s}}\cdot\log s)=s^{O(\frac{1}{\sqrt{\log s}})}$.
The required result now follows by applying \expref{Theorem}{thm:shrinkage-fixing}
to~$\phi'$.
\end{proof}

\subsection{\label{subsec:Hiding-projections}Hiding projections}

Intuitively, a hiding projection is a random projection $\p$ in which,
when given $\p_{y_{j}\gets\sigma}$, it is hard to tell in which locations
the variable $y_{j}$ appears in $\p$. Formally, we define $q$-hiding
projections as follows. 
\begin{definition}
\label{def:hiding-projections}Let $0\leq q_{0},q_{1}\leq1$. We say that
a random projection $\p$ is a \emph{$(q_{0},q_{1})$-hiding projection}
if for every projection $\pi$, every bit $\sigma\in\B$, and all
variables $x_{i},y_{j}$, we have 
\[
\Pr\!\left[\p(x_{i})\in\left\{ y_{j},\overline{y_{j}}\right\} \;\middle|\;\p_{y_{j}\gets\sigma}=\pi\right]\leq q_{\sigma},
\]
whenever the event conditioned on has positive probability. For shorthand,
we say that $\p$ is a $q$-hiding projection, for $q=\sqrt{q_{0}q_{1}}$. 
\end{definition}

\noindent To illustrate the definition, consider the following natural
random restriction: given $n$~variables $x_{1},\ldots,x_{n}$, the
restriction chooses a set of $m$~variables uniformly at random,
and fixes all the other variables to random bits. This restriction
is not captured by the notion of $p$-random restrictions or by fixing
projections, but as we demonstrate next, it can be implemented by
hiding projections. We start with the simple case of~$m=1$, and
then consider the general case. 
\begin{example}
\label{ex:random-edge}In order to implement the case of~$m=1$, consider
the random projection $\p$ from $x_{1},\ldots,x_{n}$ to $y$ that
is defined as follows: the projection~$\p$ chooses an index~$\i\in\left[n\right]$
and a bit $\btau\in\B$ uniformly at random, sets $\p(x_{\i})=y^{\btau}$,
and sets $\p(x_{i'})$ to a random bit for all $i'\in\left[n\right]\setminus\left\{ \i\right\} $.
It is clear that $\p$ is essentially equivalent to the random restriction
described above for~$m=1$. We claim that $\p$ is a $\frac{1}{n}$-hiding
projection. To see this, observe that for every bit~$\sigma\in\B$,
the projection $\p_{y\gets\sigma}$ is a uniformly distributed string
in~$\B^{n}$, and moreover, this is true conditioned on any possible
value of~$\i$. In particular, the random variable~$\i$ is independent
of~$\p_{y\gets\sigma}$. Therefore, for every projection~$\pi\in\B^{n}$
and index~$i\in\left[n\right]$ we have 
\[
\Pr\!\left[\p(x_{i})\in\left\{ y,\overline{y}\right\} \;\middle|\;\p_{y\gets\sigma}=\pi\right]=\Pr\!\left[\i=i\;\middle|\;\p_{y\gets\sigma}=\pi\right]=\Pr\!\left[\i=i\right]=\frac{1}{n},
\]
so $\p$ satisfies the definition of a $(q_{0},q_{1})$-hiding projection
with $q_{0}=q_{1}=\frac{1}{n}$. Intuitively, given $\p_{y\gets\sigma}$,
one cannot guess the original location of~$y$ in~$\p$ better than
a random guess.
\end{example}

\begin{example}
\label{ex:random-m-variables-alive}We turn to consider the case of a
general~$m\in\N$. Let $\p$ be the random projection from $x_{1},\ldots,x_{n}$
to $y_{1},\ldots,y_{m}$ that is defined as follows: the projection~$\p$
chooses $m$~distinct indices $\i_{1},\ldots,\i_{m}\in\left[n\right]$
and $m$~bits $\btau_{1},\ldots,\btau_{m}\in\B$ uniformly at random,
sets $\p(x_{\i_{j}})=y_{j}^{\btau_{j}}$ for every $j\in\left[m\right]$,
and sets all the other variables~$x_{i}$ to random bits. It is clear
that $\p$ is essentially equivalent to the random projection described
above. We show that $\p$ is a $\frac{1}{n-m+1}$-hiding projection.
To this end, we should show that for every~$i\in\left[n\right]$,
$j\in\left[m\right]$, and $\sigma\in\B$ we have 
\[
\Pr\!\left[\p(x_{i})\in\left\{ y_{j},\overline{y_{j}}\right\} \;\middle|\;\p_{y_{j}\gets\sigma}=\pi\right]\le\frac{1}{n-m+1}.
\]
For simplicity of notation, we focus on the case where~$j=m$. Observe
that $\p_{y_{m}\gets\sigma}$ reveals the values of $\i_{1},\ldots,\i_{m-1}$,
and the projection of $\p_{y_{m}\gets\sigma}$ on the remaining $n-m+1$
indices is uniform over $\{0,1\}^{n-m+1}$. Moreover, the latter assertion
remains true conditioned on any possible value of~$\i_{m}$ (which
must be different from the known values of $\i_{1},\ldots,\i_{m-1}$).
Therefore given $\i_{1},\ldots,\i_{m-1}$, the random variable $\i_{m}$
(ranging over all indices other than $\i_{1},\ldots,\i_{m-1}$) is
independent of the projection of $\p_{y_{m}\gets\sigma}$ on the $n-m+1$
indices other than $\i_{1},\ldots,\i_{m-1}$. It follows that for
every projection~$\pi$ in the support of~$\p_{y_{m}\gets\sigma}$
and every index~$i\in\left[n\right]$, if $\p(x_{i_{1}})=y_{1}^{\btau_{1}},\ldots,\p(x_{i_{m-1}})=y_{m-1}^{\btau_{m-1}}$
then: 
\begin{align*}
\Pr\!\left[\p(x_{i})\in\left\{ y_{m},\overline{y_{m}}\right\} \;\middle|\;\p_{y_{m}\gets\sigma}=\pi\right] & =\Pr\!\left[\i_{m}=i\;\middle|\;\p_{y_{m}\gets\sigma}=\pi\right]\\
 & =\Pr\!\left[\i_{m}=i\;\middle|\;\i_{1}=i_{1},\ldots,\i_{m-1}=i_{m-1}\right]\le\frac{1}{n-m+1},
\end{align*}
as required.
\end{example}

In \expref{Section}{sec:shrinkage-hiding}, we prove the following shrinkage theorem
for hiding projections. 
\begin{theorem}[Shrinkage under hiding projections]
\label{thm:shrinkage-hiding}Let $\phi$ be a formula of size $s$ and
depth $d$, and let $\p$ be a $q$-hiding projection. Then 
\[
\E\left[L(\phi|_{\p})\right]=O\!\left(m^{4}\cdot q^{2}\cdot d^{2}\cdot s+m^{2}\cdot q\cdot\sqrt{s}\right).
\]
\end{theorem}

Applying formula balancing, we can obtain an analog of 
\expref{Corollary}{cor:shrinkage-without-depth-fixing},
with an identical proof. 
\begin{corollary}
\label{cor:shrinkage-without-depth-hiding} Let $f\colon\B^{n}\to\B$
be a function with formula complexity~$s$, and let $\p$ be a $q$-hiding
projection. Then 
\[
\E\left[L(f|_{\p})\right]=m^{4}\cdot q^{2}\cdot s^{1+O\bigl(\frac{1}{\sqrt{\log s}}\bigr)}+m^{2}\cdot q\cdot s^{\frac{1}{2}+O\bigl(\frac{1}{\sqrt{\log s}}\bigr)}.
\]
\end{corollary}

\begin{remark}
Note that \expref{Theorem}{thm:shrinkage-hiding} 
loses a factor of~$m^{4}$ compared
to \expref{Theorem}{thm:shrinkage-fixing}. While 
this loss might be large in general,
it will be of no importance in our applications. The factor $m^{4}$
can be improved to $m^{2}$ in our actual applications, as discussed
in %
\expref{Section}{subsec:hiding-to-fixing}.  %

The following example shows that the loss of a factor of at least
$m^{2}$ is necessary for \expref{Theorem}{thm:shrinkage-hiding}. 
Let $f=\parity_{n}$
be the Parity function over~$x_{1},\ldots,x_{n}$, and note that
the formula complexity of~$f$ is~$s=\Theta(n^{2})$. Let $\p$
be the random projection from \expref{Example}{ex:random-m-variables-alive}, and
recall that it is a $q$-hiding projection for~$q=\frac{1}{n-m+1}$.
Then, $f|_{\p}$ is the Parity function over $m$~bits, and therefore
its formula complexity is~$\Theta(m^{2})$. It follows that when
$m\leq n/2$, 
\[
\E\left[L(f|_{\p})\right]=\Theta(m^{2})=\Theta(m^{2}\cdot q^{2}\cdot s).
\]
\end{remark}

\subsection{\label{subsec:application-sketch}Sketch: Cubic formula lower bound
for \texorpdfstring{$\protect\ACz$}{AC0}}

We now sketch how our shrinkage theorems can be used to prove the
cubic formula lower bound for $\ACz$. We start by sketching how our
shrinkage theorems can be used to reprove a quadratic formula lower
bound for the surjectivity function~$\surj$ (originally due to~\cite{BM12}).
Then, we sketch a lower bound for compositions of the form~$f\d\surj$,
where $f$~is an arbitrary function. Finally, we combine the latter
bound with an idea of Andreev~\cite{A87} to obtain a cubic lower
bound for an explicit function in~$\ACz$. We note that the first
step is the key idea of the proof, whereas the two latter steps are
relatively standard. A formal proof following this sketch can be found
in \expref{Section}{sec:applications}.

\subsubsection{Lower bound for the surjectivity function}

The surjectivity function~$\surj_{n}\colon\B^{n}\to\B$ \cite{BM12} takes
as input a list of $r$~numbers in~$\left[S\right]$ for some~$S\in\N$,
encoded in binary, and outputs~$1$ if and only if all the numbers
in~$\left[S\right]$ appear in the list. For our purposes, we 
write %
$S'=S-2$ and set $r=\frac{3}{2}\cdot S'+2$. We sketch a proof that
$L(\surj_{n})\approx\Omega(S^{2})$. Since the input length of~$\surj_{n}$,
when measured in bits, is $n=O(S\log S)$, this implies that $L(\surj_{n})\approx\Omega(\frac{n^{2}}{\log^{2}n})$.

To this end, we construct a random projection~$\p$ for $\surj_{n}$
that maps the input variables $x_{1},\ldots,x_{n}$ to a single variable~$y$
as follows: We choose two distinct numbers~$\i,\j\in\left[S\right]$ at random,
and then partition the remaining $S-2$ numbers at random to two sets
$\A$ and $\bB$ of size $\frac{1}{2}(S-2)=\frac{1}{2}S'$ each. Then,
we create a list of $r=\frac{3}{2}\cdot S'+2$ numbers in~$\left[S\right]$
as follows:
\begin{enumerate}
\item Initially, all the places in the list are vacant.

\item We put each of the $\frac{1}{2}S'+1$~numbers in~$\A\cup\left\{ \j\right\} $
at a single random vacant place in the list.
\item We put each of the $\frac{1}{2}S'$~numbers in~$\bB$ in two random vacant
places in the list (for a total of $S'$ places).
\item Note that so far we filled $\frac{3}{2}S'+1=r-1$ vacant places in the list.
The remaining vacant place in the list is left empty.
\end{enumerate}
We set~$\p$ such that it fixes all the inputs variables $x_{1},\ldots,x_{n}$
to constants according to the list we chose, except for the variables
corresponding to the empty place in the list. Then, the $\lceil \log S \rceil$~binary input variables
that correspond to the empty place are mapped to $\left\{ 0,1,y,\overline{y}\right\} $
such that if we substitute $y=1$ we get the number~$\i$ and if
we substitute $y=0$ we get the number~$\j$. Observe that if $y=1$,
then the numbers in~$\A\cup\left\{ \i,\j\right\} $ appear in the
list exactly once and the numbers in~$\bB$ appear exactly twice.
On the other hand, if $y=0$, then the numbers in~$\A$ appear in
the list exactly once, the numbers in~$\bB\cup\left\{ \j\right\} $
appear in the list exactly twice, and $\i$ does not appear at all.

Observe that $\surj_{n}|_{\p}=y$. Indeed, if $y=1$ then the foregoing
list of numbers contains all the numbers in~$\left[S\right]$ and
thus $\surj_{n}$ takes the value~$1$, and if $y=0$ then the list
of numbers does not contain~$\i$ and thus $\surj_{n}$ takes the
value~$0$. We claim that $\p$ is an $O(\frac{1}{S})$-hiding projection.
Intuitively, if we substitute a value~$\sigma\in\B$ in~$y$, then
given the resulting projection~$\p_{y\gets\sigma}$ it is difficult
to tell which input variables were mapped to~$y$ in the original
projection~$\p$. The reason is that in order to tell that, one has
to guess the empty place in the original list. To this end, one has
to guess~$\i$ among the $\approx\frac{1}{2}S'$ numbers that appear
exactly once in the list (if $\sigma=1$) or to guess~$\j$ among
the $\approx\frac{1}{2}S'$ numbers that appear exactly twice in the
list (if~$\sigma=0$).

A bit more formally, observe that $\p$ maps an input variable~$x_{h}$
to~$\left\{ y,\overline{y}\right\} $ only if $x_{h}$~corresponds
to the empty place in the list of~$\p$, and that $\p_{y\gets0}$
and $\p_{y\gets1}$ represent lists of $r$~numbers in~$\left[S\right]$.
Hence, to %
find an upper bound on  %
the probability that 
$\p(x_{h})\in\left\{ y,\overline{y}\right\} $
conditioned on the choice of~$\p_{y\gets\sigma}$ (for $\sigma\in\B$),
we %
need an upper bound on  %
the probability that a place in the list represented
by $\p_{y\gets\sigma}$ was the empty place in the list of~$\p$.
We consider two cases:
\begin{itemize}
\item If $\sigma=1$, then the empty place in the list of~$\p$ contains~$\i$
in the list of~$\p_{y\gets1}$. From the perspective of someone who
only sees $\p_{y\gets1}$, the number~$\i$ could be any number among
the $\frac{1}{2}S'+2$ numbers that appear exactly once in the list
of~$\p_{y\gets1}$. Therefore, conditioned on~$\p_{y\gets1}$, the
empty place in the list of~$\p$ is uniformly distributed among $\frac{1}{2}S'+2$
places in the list of~$\p_{y\gets1}$. Hence, any place in the list
of $\p_{y\gets1}$ has probability of at most~$\frac{1}{\frac{1}{2}S'+2}=O\left(\frac{1}{S}\right)$
to be the empty place in~$\p$.
\item On the other hand, if $\sigma=0$, %
then the empty place in the list
of~$\p$ contains~$\j$ in the list of~$\p_{y\gets0}$. From the
perspective of someone who only sees~$\p_{y\gets0}$, the number~$\j$
could be any number among the $\frac{1}{2}S'+1$ numbers that appear
exactly twice in the list of~$\p_{y\gets0}$. Therefore, conditioned
on~$\p_{y\gets0}$, the empty place in the list of~$\p$ is uniformly
distributed among $2\cdot\left(\frac{1}{2}S'+1\right)$ places in
the list of~$\p_{y\gets0}$. Hence, any place in the list has probability
of at most~$\frac{1}{2\cdot\left(\frac{1}{2}S'+1\right)}=O\left(\frac{1}{S}\right)$
to be the empty place in~$\p$.
\end{itemize}
We showed that $\p$ is an $O(\frac{1}{S})$-hiding projection, and
therefore by our shrinkage theorem for hiding projections (specifically,
\expref{Corollary}{cor:shrinkage-without-depth-hiding}) we get that
\[
\E\left[L(\surj_{n}|_{\p})\right]\lesssim\frac{1}{S^{2}}\cdot L(\surj).
\]
(We ignored the power of $\frac{1}{\sqrt{\log L(\surj_{n})}}$ from
\expref{Corollary}{cor:shrinkage-without-depth-hiding} in order to simplify the exposition.)
On the other hand, as explained above, we have $\surj|_{\pi}=y$,
and therefore $L(\surj_{n}|_{\p})=1$. It follows that $L(\surj_{n})\gtrsim S^{2}$, as
required.

We note that the above proof is a special case of a more general phenonmenon. Specifically, in \expref{Section}{sec:adversary},
we show that such a hiding projection can be constructed for every function that has a large (unweighted) adversary bound. As \cite{BM12} showed, the surjectivity function has such a large adversary bound, and therefore we can construct a hiding projection for it.

\subsubsection{Lower bound for \texorpdfstring{$f\protect\d\protect\surj_{n}$}{f o surj\_n}}

Let $f\colon\B^{k}\to\B$ be any function, and let $\surj_{n}$ be defined
as in the previous section. The block-composition of $f$ and~$\surj_{n}$,
denoted $f\d\surj_{n}$, is the function that takes $k$~inputs $x_{1},\ldots,x_{k}\in\B^{n}$
for $\surj_{n}$ and outputs
\[
(\fs_{n})(x_{1},\ldots,x_{k})=f\left(\surj(x_{1}),\ldots,\surj(x_{k})\right)
\]
We sketch a proof that $L(f\d\surj_{n})\approx\Omega\left(L(f)\cdot\frac{n^{2}}{\log^{2}n}\right)$.
Denote by $x_{i,1},\ldots,x_{i,n}$ the boolean variables of the input
$x_{i}$ for each $i\in\left[k\right]$. Let $\p$ be the $O(\frac{1}{S})$-hiding
projection from the previous section, and let $\p^{k}$ be random
projection from $x_{1,1},\ldots,x_{k,n}$ to $y_{1},\ldots,y_{k}$
that is obtained by joining $k$~independent copies of~$\p$. In
\expref{Section}{sec:join} we show that the join operation maintains the hiding
property, and therefore $\p^{k}$ is an $O(\frac{1}{S})$-hiding projection.
Now, it is not hard to see that $\fs_{n}|_{\p^{k}}=f$ and therefore
$L(\fs_{n}|_{\p^{k}})=L(f)$. On the other hand, by our shrinkage
theorem for hiding projections (specifically, 
\expref{Corollary}{cor:shrinkage-without-depth-hiding})
we get that
\[
\E\left[L(\fs_{n}|_{\p^{k}})\right]\lesssim\frac{1}{S^{2}}\cdot L(\fs_{n}).
\]
(Again, we ignored the power of $\frac{1}{\sqrt{\log L(\fs_{n})}}$
from \expref{Corollary}{cor:shrinkage-without-depth-hiding} in order to simplify the
exposition.) It follows that
\[
L(f)\lesssim\frac{1}{S^{2}}\cdot L(\fs_{n}),
\]
and therefore 
\[
L(f\d\surj_{n})\approx\Omega\left(L(f)\cdot S^{2}\right)=\tilde{\Omega}\left(L(f)\cdot n^{2}\right),
\]
as required.

\subsubsection{Cubic lower bound for \texorpdfstring{$\protect\ACz$}{AC0}}

Let $n=2^{k}\cdot(k+1)$ and let $h_{n}\colon\B^{n}\to\B$ be the function
that takes as input the truth table of a function~$f\colon\B^{k}\to\B$
and $k$~inputs $x_{1},\ldots,x_{k}$ for~$\surj_{2^{k}}$, and
outputs
\[
F(f,x_{1},\ldots,x_{k})=(f\d\surj)(x_{1},\ldots,.x_{k})
\]
As mentioned above, the function~$h_{n}$ is a variant of the Andreev
function in which the parity function is replaced with~$\surj$.
It is not hard to show that $h_{n}$ is in~$\ACz$ (see \expref{Section}{subsec:AC0-bounds}
for details). We sketch a proof that $L(h_{n})\approx\Omega(n^{3-o(1)})$.

To see that $L(h_{n})\approx\Omega(n^{3-o(1)})$, recall that by a
standard counting argument, there exists a function $f\colon\B^{k}\to\B$
such that $L(f)=\Theta\left(\frac{2^{k}}{\log k}\right)$. We now
hardwire~$f$ as an input to~$h_{n}$, which turns the function~$h_{n}$
to the function~$\fs_{2^{k}}$. By the lower bound we have for $\fs_{2^{k}}$,
it follows that
\[
L(F)\ge L(\fs)\approx\tilde{\Omega}\left(L(f)\cdot2^{2k}\right)=\tilde{\Omega}\left(\frac{2^{k}}{\log k}\cdot2^{2k}\right)=\Omega(n^{3-o(1)}),
\]
as required.
\begin{remark}
Note that in the proof of this result, we did not use the shrinkage
theorem for fixing projections --- indeed, we only used the shrinkage
theorem for hiding projections. However, the proof of the shrinkage
theorem for hiding projections itself relies on the shrinkage theorem for fixing
projections.
\end{remark}

\section{\label{sec:shrinkage-fixing}Proof of shrinkage theorem for fixing
projections}

In this section, we prove our the shrinkage theorem for fixing projections.
Our proof is based on the ideas of~\cite{H98}, but the presentation
is different. We start by recalling the definition of fixing projections
and the theorem's statement.
\begin{repdefinition}{\expref{Definition}{def:fixing-projections}}
Let $0\le q_{0},q_{1}\le1$. We say that a random projection $\p$
is a \emph{$(q_{0},q_{1})$-fixing projection} if for every projection~$\pi$,
every bit $\sigma\in\B$, and every variable~$x_{i}$, we have
\[
\Pr\!\left[\p(x_{i})\notin\B\text{ and }\p_{\var(\p(x_{i}))\gets\sigma}=\pi\right]\le q_{\sigma}\cdot\Pr\!\left[\p=\pi\right].
\]
For shorthand, we say that $\p$ is a $q$-fixing projection, for
$q=\sqrt{q_{0}q_{1}}$. 
\end{repdefinition}

\begin{restated}{\expref{Theorem}{thm:shrinkage-fixing}}[Shrinkage under fixing projections]
Let $\phi$ be a formula of size~$s$ and depth~$d$, and let
$\p$ be a $q$-fixing projection. Then 
\[
\E\left[L(\phi|_{\p})\right]=O\!\left(q^{2}\cdot d^{2}\cdot s+q\cdot\sqrt{s}\right).
\]
\end{restated}

Fix a formula $\phi$ of size~$s$ and depth~$d$, and let $\p$
be a $q$-fixing projection. We would like to
get an upper bound on  %
the expectation
of~$L(\phi|_{\p})$. As in~\cite{H98}, we start by 
giving an upper bound %
the probability that the projection~$\p$ shrinks a formula to size~$1$.
Specifically, we prove the following lemma in~\expref{Section}{subsec:shrinkage-to-single-literal}. 
\begin{lemma}
\label{lemma:shrinkage-to-single-literal}Let $f\colon\B^{n}\to\B$ be a
Boolean function, and let $\p$ be a $q$-fixing projection. Then,
\[
\Pr\!\left[L(f|_{\p})=1\right]\le q\cdot\sqrt{L(f)}.
\]
\end{lemma}

Next, we show that to %
find and upper bound on %
the expectation of~$L(\phi|_{\p})$,
it suffices to %
give an upper bound on %
the probability that the projection~$\p$
shrinks two formulas to size~$1$ simultaneously. In order to state
this claim formally, we introduce some notation. 
\begin{notation}
Let $g$ be a gate of~$\phi$. We denote the depth of~$g$ in~$\phi$
by $\dpt_{\phi}(g)$, and omit~$\phi$ if it is clear from context
(the root has depth~$0$). If $g$ is an internal gate, we denote
the %
subformulas  %
that are rooted in its left and right children by
$\lc(g)$ and $\rc(g)$, respectively. Here, by ``internal gate'',
we refer to a gate which is an internal node of the tree, \ie, either
an AND or an OR gate.
\end{notation}

We prove the following lemma, which says that in order to upper-bound
$\E\left[L(\phi|_{\p})\right]$ it suffices to upper-bound, for every
internal gate~$g$, the probability that $\lc(g)$ and $\rc(g)$
shrink to size~$1$ under~$\p$. 
\begin{lemma}
\label{lemma:reduction-to-size-2}For every projection~$\pi$ it holds
that   %
\[
L(\phi|_{\pi})\le\sum_{\text{\rm internal gate \ensuremath{g} of \ensuremath{\phi}}}(\dpt(g)+2)\cdot1_{\left\{ L(\lc(g)|_{\pi})=1\text{\rm  and }L(\rc(g)|_{\pi})=1\right\} }+1_{L(\phi|_{\pi})=1}.
\]
\end{lemma}

We would like to use \expref{Lemmas}{lemma:shrinkage-to-single-literal}
\expref{and}{lemma:reduction-to-size-2}
to prove the shrinkage theorem. As a warm-up, let us make the simplifying
assumption that for every two functions $f_{1},f_{2}\colon\B^{n}\to\B$,
the events $L(f_{1}|_{\p})=1$ and $L(f_{2}|_{\p})=1$ are independent.
If this were true, we could have
given an upper bound on %
$\E\left[L(\phi|_{\p})\right]$
as follows: 
\begin{align*}
\E\left[L(\phi|_{\p})\right] & \le\sum_{\text{int.\ gate \ensuremath{g} of \ensuremath{\phi}}}(\dpt(g)+2)\cdot\E\left[1_{\left\{ L(\lc(g)|_{\p})=1\text{ and }L(\rc(g)|_{\p})=1\right\} }\right] 
 +\E\left[1_{L(\phi|_{\pi})=1}\right]\\
 & \qquad\qquad\qquad\qquad\qquad\qquad\qquad\qquad\qquad
   \qquad\qquad\qquad\qquad\qquad\qquad\ \
\text{(\expref{Lemma}{lemma:reduction-to-size-2})}\\        %
 & \le(d+2)\cdot\sum_{\text{int.\ gate \ensuremath{g} of \ensuremath{\phi}}}\Pr\!\left[L(\lc(g)|_{\p})=1\text{ and }L(\rc(g)|_{\p})=1\right] 
 +\E\left[1_{L(\phi|_{\pi})=1}\right]\\
 & \qquad\qquad\qquad\qquad\qquad\qquad\qquad\qquad\qquad
   \qquad\qquad\qquad\qquad\qquad\quad\ 
\text{(\ensuremath{\phi} is of depth\,\ensuremath{d})}\\ %
 & =(d+2)\cdot\sum_{\text{int.\ gate \ensuremath{g} of \ensuremath{\phi}}}\Pr\!\left[L(\lc(g)|_{\p})=1\right]\cdot\Pr\!\left[L(\rc(g)|_{\p})=1\right] 
  +\E\left[1_{L(\phi|_{\pi})=1}\right]\\
& \qquad\qquad\qquad\qquad\qquad\qquad\qquad\qquad\qquad
  \qquad\qquad\qquad\quad\ 
\text{(simplifying assumption)}\\          %
  & \le(d+2)\cdot\sum_{\text{int.\ gate \ensuremath{g} of \ensuremath{\phi}}}q^{2}\cdot\sqrt{L(\lc(g))\cdot L(\rc(g))}+q\cdot\sqrt{s}
\qquad\qquad\quad\ \
  \text{(\expref{Lemma}{lemma:shrinkage-to-single-literal})}\\  %
& \le q^{2}\cdot(d+2)\cdot\sum_{\text{int.\ gate \ensuremath{g} of \ensuremath{\phi}}}\left(L(\lc(g))+L(\rc(g))\right)+q\cdot\sqrt{s}
  \quad\ \    \text{(AM--GM inequality)}\\              %
 & \le q^{2}\cdot(d+2)\cdot\sum_{\text{int.\ gate \ensuremath{g} of \ensuremath{\phi}}}\left(\s(\lc(g))+\s(\rc(g))\right)+q\cdot\sqrt{s}\\
 & =q^{2}\cdot(d+2)\cdot\sum_{\text{int.\ gate \ensuremath{g} of \ensuremath{\phi}}}\s(g)+q\cdot\sqrt{s}.
\end{align*}  
The last sum counts every leaf~$\ell$ of~$\phi$ once for each
internal ancestor of~$\ell$, so the last expression is equal to
\begin{align*}
 & q^{2}\cdot(d+2)\cdot\sum_{\text{leaf \ensuremath{\ell} of \ensuremath{\phi}}}\dpt(\ell)+q\cdot\sqrt{s}\\
 & \qquad\le q^{2}\cdot(d+2)\cdot\sum_{\text{leaf \ensuremath{\ell} of \ensuremath{\phi}}}d+q\cdot\sqrt{s}\\
 & \qquad=q^{2}\cdot(d+2)\cdot d\cdot s+q\cdot\sqrt{s}\\
 & \qquad=O\!\left(q^{2}\cdot d^{2}\cdot s+q\cdot\sqrt{s}\right),
\end{align*}
which is the bound we wanted. However, the above calculation only
works under our simplifying assumption, which is false: the events
$L(f_{1}|_{\p})=1$ and $L(f_{2}|_{\p})=1$ will often be dependent.
In particular, in order for the foregoing calculation to work, we
need the following inequality to hold: 
\[
\Pr\!\left[L(f_{2}|_{\p})=1\mid L(f_{1}|_{\p})=1\right]\le q\cdot\sqrt{L(f_{2})}.
\]
This inequality holds under our simplifying assumption by \expref{Lemma}{lemma:shrinkage-to-single-literal},
but may not hold in general. Nevertheless, we prove the following
similar statement in \expref{Section}{subsec:shrinkage-to-two-literals}. 
\begin{lemma}
\label{lemma:shrinkage-with-conditioning}Let $\p$ be a $q$-fixing projection.
Let $f_{1},f_{2}\colon\B^{n}\to\B$, let $\sigma,\tau\in\B$, and
let $y_{j}$ be a variable. Then, 
\[
\Pr\!\left[L(f_{2}|_{\p_{y_{j}\gets\sigma}})=1\;\middle|\;f_{1}|_{\p}=y_{j}^{\tau}\right]\le q\cdot\sqrt{L(f_{2})}.
\]
\end{lemma}

Intuitively, \expref{Lemma}{lemma:shrinkage-with-conditioning} breaks the dependency
between the events $L(f_{1}|_{\p})=1$ and $L(f_{2}|_{\p})=1$ by
fixing in~$f_{2}$ the single literal to which $f_{1}$ has shrunk.
We would now like to use \expref{Lemma}{lemma:shrinkage-with-conditioning} to prove
the theorem. To this end, we prove an appropriate variant of \expref{Lemma}{lemma:reduction-to-size-2},
which allows using the projection $\pi_{y_{j}\gets\sigma}$ rather
than $\pi$ in the second function. This variant is motivated by the
following ``one-variable simplification rules'' of \cite{H98},
which are easy to verify. 
\begin{fact}[one-variable simplification rules]
\label{fact:simplification-rules}Let $h\colon\B^{m}\to\B$ be a function
over the variables $y_{1},\ldots,y_{m}$, and let $\sigma\in\B$.
We denote by $h_{y_{j}\gets\sigma}$ the function obtained from $h$
by setting~$y_{j}$ to the bit~$\sigma$. Then: 
\begin{itemize}
\item The function $y_{j}^{\sigma}\vee h$ is equal to the function $y_{j}^{\sigma}\vee h_{y_{j}\gets\sigma}$. 
\item The function $y_{j}^{\sigma}\wedge h$ is equal to the function $y_{j}^{\sigma}\wedge h_{y_{j}\gets\overline{\sigma}}$. 
\end{itemize}
\end{fact}

In order to use the simplification rules, we define, for every internal
gate~$g$ of~$\phi$ and projection~$\pi$, an event $\cE_{g,\pi}$
as follows: if $g$ is an OR gate, then $\cE_{g,\pi}$ is the event
that there exists some literal~$y_{j}^{\sigma}$ (for $\sigma\in\B$)
such that $\lc(g)|_{\pi}=y_{j}^{\sigma}$ and $L(\rc(g)|_{\pi_{y_{j}\gets\sigma}})=1$.
If $g$ is an AND gate, then $\cE_{g,\pi}$ is defined similarly,
except that we replace $\pi_{y_{j}\gets\sigma}$ with $\pi_{y_{j}\gets\overline{\sigma}}$.
We have the following lemma, which is proved in \expref{Section}{subsec:reduction-to-size-2-with-substitution}. 
\begin{lemma}
\label{lemma:reduction-to-size-2-with-substitution}For every projection~$\pi$
we have 
\[
L(\phi|_{\pi})\le\sum_{\text{internal gate \ensuremath{g} of \ensuremath{\phi}}}(\dpt(g)+2)\cdot1_{\cE_{g,\pi}}+1_{L(\phi|_{\pi})=1}.
\]
\end{lemma}

We can now use the following corollary of \expref{Lemma}{lemma:shrinkage-with-conditioning}
to replace our simplifying assumption. 
\begin{corollary}
\label{cor:shrinkage-to-two-literals}For every internal gate~$g$ of~$\phi$
we have 
\[
\Pr\!\left[\cE_{g,\p}\right]\le q^{2}\cdot\sqrt{L(\lc(g))\cdot L(\rc(g))}.
\]
\end{corollary}

\begin{proof}
Let $g$ be an internal gate of~$\phi$. We prove the corollary for
the case where $g$ is an OR gate, and the proof for the case that
$g$ is an AND gate is similar. We have 
\begin{align*}
\Pr\!\left[\cE_{g,\p}\right] & =\Pr\!\left[\exists\,\text{literal }y_{j}^{\sigma}\colon\lc(g)|_{\p}=y_{j}^{\sigma}\text{ and }L(\rc(g)|_{\p_{y_{j}\gets\sigma}})=1\right]\\
 & =\sum_{\text{literal }y_{j}^{\sigma}}\Pr\!\left[\lc(g)|_{\p}=y_{j}^{\sigma}\text{ and }L(\rc(g)|_{\p_{y_{j}\gets\sigma}})=1\right]\\
 & =\sum_{\text{literal }y_{j}^{\sigma}}\Pr\!\left[\left.L(\rc(g)|_{\p_{y_{j}\gets\sigma}})=1\right|\lc(g)|_{\p}=y_{j}^{\sigma}\right]\cdot\Pr\!\left[\lc(g)|_{\p}=y_{j}^{\sigma}\right]\\
 & \le q\cdot\sqrt{L(\rc(g))}\cdot\sum_{\text{literal }y_{j}^{\sigma}}\Pr\!\left[\lc(g)|_{\p}=y_{j}^{\sigma}\right] & \text{(\expref{Lemma}{lemma:shrinkage-with-conditioning})}\\
 & =q\cdot\sqrt{L(\rc(g))}\cdot\Pr\!\left[L(\lc(g)|_{\p})=1\right]\\
 & \le q^{2}\cdot\sqrt{L(\lc(g))\cdot L(\rc(g))}, & \text{(\expref{Lemma}{lemma:shrinkage-to-single-literal})}
\end{align*}
as required. 
\end{proof}
The shrinkage theorem now follows using the same calculation as above,
replacing \expref{Lemma}{lemma:reduction-to-size-2} with \expref{Lemma}{lemma:reduction-to-size-2-with-substitution}
and the simplifying assumption with \expref{Corollary}{cor:shrinkage-to-two-literals}:
\begin{align*}
\E\left[L(\phi|_{\p})\right] & \le(d+2)\cdot\sum_{\text{internal gate \ensuremath{g} of \ensuremath{\phi}}}\Pr\!\left[\cE_{g,\p}\right]+q\cdot\sqrt{s} & \ \text{(\expref{Lemma}{lemma:reduction-to-size-2-with-substitution})}\\
 & \le q^{2}\cdot(d+2)\cdot\sum_{\text{internal gate \ensuremath{g} of \ensuremath{\phi}}}\sqrt{L(\lc(g))\cdot L(\rc(g))}+q\cdot\sqrt{s} & \ \text{(\expref{Corollary}{cor:shrinkage-to-two-literals})}\\
& \le q^{2}\cdot(d+2)\cdot\sum_{\text{internal gate \ensuremath{g} of \ensuremath{\phi}}}\left(L(\lc(g))+L(\rc(g))\right)+q\cdot\sqrt{s} &
\ \text{(AM--GM ineq.)}\\   %
 & =q^{2}\cdot(d+2)\cdot\sum_{\text{internal gate \ensuremath{g} of \ensuremath{\phi}}}\s(g)+q\cdot\sqrt{s}\\
 & \le O(q^{2}\cdot d^{2}\cdot s+q\cdot\sqrt{s}).
\end{align*}

\noindent In the remainder of this section, we prove
\expref{Lemmas}{lemma:shrinkage-to-single-literal},
\expref{and}{lemma:shrinkage-with-conditioning},
\expref{and}{lemma:reduction-to-size-2-with-substitution}.
\begin{remark}
In this paper, we do not prove \expref{Lemma}{lemma:reduction-to-size-2}, since we
do not actually need it to prove our main results in \expref{Section}{sec:applications}.
However, this lemma can be
established using the proof of \expref{Lemma}{lemma:reduction-to-size-2-with-substitution},
with some minor changes. 
\end{remark}

\subsection{\label{subsec:shrinkage-to-single-literal}Proof of \expref{Lemma}{lemma:shrinkage-to-single-literal}}

We begin with proving \expref{Lemma}{lemma:shrinkage-to-single-literal}, restated
next.
\begin{restated}{\expref{Lemma}{lemma:shrinkage-to-single-literal}}
Let $f\colon\B^{n}\to\B$ be a Boolean function, and let $\p$ be
a $q$-fixing projection. Then, 
\[
\Pr\!\left[L(f|_{\p})=1\right]\le q\cdot\sqrt{L(f)}.
\]
\end{restated}

Let $f\colon\B^{n}\to\B$, and let $\cE$ be the set of projections~$\pi$
such that $L(f|_{\pi})=1$. We prove that the probability that $\p\in\cE$
is at most $q\cdot\sqrt{L(f)}$. Our proof follows closely the proof
of \cite[Lemma 4.1]{H98}.

Let $\Pi$ be a protocol that solves~$\KW_{f}$ and has $L(f)$ leaves
(such a protocol exists by \expref{Theorem}{thm:KW-connection}).
Let $\cA$ and $\cB$
be the sets of projections~$\pi$ for which $f|_{\pi}$ is the constants
$1$ and~$0$, respectively. We extend the protocol~$\Pi$ to take
inputs from~$\cA\times\cB$ as follows: when Alice and Bob are given
as inputs the projections $\pi^{A}\in\cA$ and~$\pi^{B}\in\cB$,
respectively, they construct strings $a,b\in\B^{n}$ from $\pi^{A},\pi^{B}$
by substituting~$0$ for all the variables $y_{1},\ldots,y_{m}$,
and invoke~$\Pi$ on the inputs $a$ and~$b$. Observe that $a$
and~$b$ are indeed legal inputs for~$\Pi$ (since $f(a)=1$ and
$f(b)=0$). Moreover, recall that the protocol~$\Pi$ induces a partition
of~$\cA\times\cB$ to combinatorial rectangles, and that we denote
the rectangle of the leaf~$\ell$ by $\cA_{\ell}\times\cB_{\ell}$
(see \expref{Notation}{notation:partition-to-rectangles}).

Our proof strategy is the following: We associate with every projection~$\pi\in\cE$
a leaf of~$\Pi$, denoted $\lp(\pi)$. We consider the two disjoint
events $\cE^{+},\cE^{-}$ that correspond to the event that $f|_{\pi}$
is a single positive literal or a single negative literal, respectively,
and show that for every leaf~$\ell$ we have 
\begin{align}
\Pr\!\left[\p\in\cE^{+}\text{ and }\lp(\p)=\ell\right] & \le q\cdot\sqrt{\Pr\!\left[\p\in\cA_{\ell}\right]\cdot\Pr\!\left[\p\in\cB_{\ell}\right]}\label{shrinkage-to-single-literal-leaf-bound-positive}\\
\Pr\!\left[\p\in\cE^{-}\text{ and }\lp(\p)=\ell\right] & \le q\cdot\sqrt{\Pr\!\left[\p\in\cA_{\ell}\right]\cdot\Pr\!\left[\p\in\cB_{\ell}\right]}.\label{shrinkage-to-single-literal-leaf-bound-negative}
\end{align}
Together, the two inequalities imply that 
\[
\Pr\!\left[\p\in\cE\text{ and }\lp(\p)=\ell\right]\le2q\cdot\sqrt{\Pr\!\left[\p\in\cA_{\ell}\right]\cdot\Pr\!\left[\p\in\cB_{\ell}\right]}.
\]
The desired bound on $\Pr\!\left[\p\in\cE\right]$ will follow by
summing the latter bound over all the leaves~$\ell$ of~$\Pi$.

We start by explaining how to associate a leaf with every projection~$\pi\in\cE^{+}$.
Let $\pi\in\cE^{+}$. Then, it must be the case that $f|_{\pi}=y_{j}$
for some $j\in\left[m\right]$. We define the projections $\pi^{1}=\pi_{y_{j}\gets1}$
and $\pi^{0}=\pi_{y_{j}\gets0}$, and observe that $\pi^{1}\in\cA$
and $\pi^{0}\in\cB$. We now define $\lp(\pi)$ to be the leaf to
which~$\Pi$ arrives when invoked on inputs $\pi^{1}$ and $\pi^{0}$.
Observe that the output of~$\Pi$ at~$\lp(\pi)$ must be a variable~$x_{i}$
that satisfies $\pi(x_{i})\in\left\{ y_{j},\overline{y}_{j}\right\} $,
and thus $\pi_{\var(\pi(x_{i}))\gets1}=\pi^{1}$.

Next, fix a leaf~$\ell$. We prove that $\Pr\!\left[\p\in\cE^{+}\text{ and }\lp(\p)=\ell\right]\le q_{1}\cdot\Pr\!\left[\p\in\cA_{\ell}\right]$.
Let $x_{i}$ be the output of the protocol~$\Pi$ at~$\ell$. Then,
\begin{align*}
\Pr\!\left[\p\in\cE^{+}\text{ and }\lp(\p)=\ell\right] & \le\Pr\!\left[\p^{1}\in\cA_{\ell}\text{ and }\p(x_{i})\notin\B\text{ and }\p_{\var(\p(x_{i}))\gets1}=\p^{1}\right]\\
 & \le\Pr\!\left[\p(x_{i})\notin\B\text{ and }\p_{\var(\p(x_{i}))\gets1}\in\cA_{\ell}\right]\\
 & =\sum_{\pi\in\cA_{\ell}}\Pr\!\left[\p(x_{i})\notin\B\text{ and }\p_{\var(\p(x_{i}))\gets1}=\pi\right]\\
 & \le q_{1}\cdot\sum_{\pi\in\cA_{\ell}}\Pr\!\left[\p=\pi\right]\qquad\qquad\qquad\text{(since \ensuremath{\p} is \ensuremath{(q_{0},q_{1})}-fixing)}\\
 & =q_{1}\cdot\Pr\!\left[\p\in\cA_{\ell}\right].
\end{align*}
Similarly, it can be proved that $\Pr\!\left[\p\in\cE^{+}\text{ and }\lp(\p)=\ell\right]\le q_{0}\cdot\Pr\!\left[\p\in\cB_{\ell}\right]$.
Together, the two bounds imply that 
\[
\Pr\!\left[\p\in\cE^{+}\text{ and }\lp(\p)=\ell\right]\le\sqrt{q_{1}\Pr\!\left[\p\in\cA_{\ell}\right]\cdot q_{0}\Pr\!\left[\p\in\cB_{\ell}\right]}=q\cdot\sqrt{\Pr\!\left[\p\in\cA_{\ell}\right]\cdot\Pr\!\left[\p\in\cB_{\ell}\right]}
\]
for every leaf~$\ell$ of~$\Pi$. We define $\lp(\pi)$ for projections
$\pi\in\cE^{-}$ in an analogous way, and then a similar argument
shows that 
\[
\Pr\!\left[\p\in\cE^{-}\text{ and }\lp(\p)=\ell\right]\le q\cdot\sqrt{\Pr\!\left[\p\in\cA_{\ell}\right]\cdot\Pr\!\left[\p\in\cB_{\ell}\right]}.
\]
It follows that 
\[
\Pr\!\left[\p\in\cE\text{ and }\lp(\p)=\ell\right]\le2q\cdot\sqrt{\Pr\!\left[\p\in\cA_{\ell}\right]\cdot\Pr\!\left[\p\in\cB_{\ell}\right]}.
\]
Finally, let $\cL$ denote the set of leaves of~$\Pi$. It holds
that 
\begin{align*}
\Pr\!\left[\p\in\cE\right] & =\sum_{\ell\in\cL}\Pr\!\left[\p\in\cE\text{ and }\lp(\p)=\ell\right]\\
 & \le2q\cdot\sum_{\ell\in\cL}\sqrt{\Pr\!\left[\p\in\cA_{\ell}\right]\cdot\Pr\!\left[\p\in\cB_{\ell}\right]}\\
 & \le2q\cdot\sqrt{\left|\cL\right|}\cdot\sqrt{\sum_{\ell\in\cL}\Pr\!\left[\p\in\cA_{\ell}\right]\cdot\Pr\!\left[\p\in\cB_{\ell}\right]} & \text{(Cauchy-Schwarz -- see~\expref{Fact}{fact:special-cauchy-schwarz})}\\
 & =2q\cdot\sqrt{L(f)}\cdot\sqrt{\sum_{\ell\in\cL}\Pr\!\left[\p\in\cA_{\ell}\right]\cdot\Pr\!\left[\p\in\cB_{\ell}\right]}.
\end{align*}
We conclude the proof by showing that $\sum_{\ell\in\cL}\Pr\!\left[\p\in\cA_{\ell}\right]\cdot\Pr\!\left[\p\in\cB_{\ell}\right]\le\frac{1}{4}$.
To this end, let $\p^{A},\p^{B}$ be two independent random variables
that are distributed identically to~$\p$. Then, we have 
\begin{align*}  %
\sum_{\ell\in\cL}\Pr\!\left[\p\in\cA_{\ell}\right]\cdot\Pr\!\left[\p\in\cB_{\ell}\right] & =\sum_{\ell\in\cL}\Pr\!\left[\p^{A}\in\cA_{\ell}\right]\cdot\Pr\!\left[\p^{B}\in\cB_{\ell}\right]\\
 & =\sum_{\ell\in\cL}\Pr\!\left[(\p^{A},\p^{B})\in\cA_{\ell}\times\cB_{\ell}\right] %
\qquad\quad\ \ \ \text{(\ensuremath{\p^{A},\p^{B}} are independent)}\\  %
& =\Pr\!\left[(\p^{A},\p^{B})\in\cA\times\cB\right]
\qquad   %
 \text{(\ensuremath{\cA_{\ell}\times\cB_{\ell}} are a partition of\,\ensuremath{\cA\times\cB})}\\    %
 & =\Pr[\p^{A}\in\cA]\cdot\Pr[\p^{B}\in\cB] %
\qquad\qquad\quad \text{(\ensuremath{\p^{A},\p^{B}} are independent)}\\  %
 & =\Pr[\p\in\cA]\cdot\Pr[\p\in\cB]\\
 & \le\Pr[\p\in\cA]\cdot\left(1-\Pr[\p\in\cA]\right) 
\le 1/4.  %
\qquad\qquad \text{(\ensuremath{\cA},\ensuremath{\cB} are disjoint)}\\   %
\end{align*}

\subsection{\label{subsec:shrinkage-to-two-literals}Proof of \expref{Lemma}{lemma:shrinkage-with-conditioning}}

We turn to proving \expref{Lemma}{lemma:shrinkage-with-conditioning}, 
restated next.
\begin{restated}{\expref{Lemma}{lemma:shrinkage-with-conditioning}}
Let $\p$ be a $q$-fixing projection. Let $f_{1},f_{2}\colon\B^{n}\to\B$,
let $\sigma,\tau\in\B$, and let $y_{j}$ be a variable. Then, 
\begin{equation}
\Pr\!\left[L(f_{2}|_{\p_{y_{j}\gets\sigma}})=1\;\middle|\;f_{1}|_{\p}=y_{j}^{\tau}\right]\le q\cdot\sqrt{L(f_{2})}.\label{eq:shrinkage-with-conditioning}
\end{equation}
\end{restated}

Let $\p$ be a $(q_{0},q_{1})$-fixing random projection, and let
$q=\sqrt{q_{0}q_{1}}$. Let $f_{1},f_{2}\colon\B^{n}\to\B$, let $\sigma,\tau\in\B$,
and let $y_{j}$ be a variable. We prove \expref{Equation}{eq:shrinkage-with-conditioning}.
For simplicity, we focus on the case that $f_{1}|_{\p}=y_{j}$, and
the case that $f_{1}|_{\p}=\overline{y_{j}}$ can be dealt with similarly.
The crux of the proof is to show that the random projection $\p_{y_{j}\gets\sigma}$
is essentially a $(q_{0},q_{1})$-fixing projection even when conditioned
on the event~$f_{1}|_{\p}=y_{j}$, and therefore \expref{Equation}{eq:shrinkage-with-conditioning}
is implied immediately by \expref{Lemma}{lemma:shrinkage-to-single-literal}.

In order to carry out this argument, we first establish some notation.
Let $\bI^{+}=\p^{-1}(y_{j})$ and $\bI^{-}=\p^{-1}(\overline{y_{j}})$,
and denote by $\p'$ the projection obtained from~$\p$ by restricting
its domain to~$\left\{ x_{1},\ldots,x_{n}\right\} \backslash(\bI^{+}\cup\bI^{-})$.
We denote by $f_{2,\bI^{+},\bI^{-}}$ the function over $\left\{ x_{1},\ldots,x_{n}\right\} \backslash(\bI^{+}\cup\bI^{-})$
that is obtained from~$f_{2}$ by hard-wiring $\sigma$ and $\overline{\sigma}$
to the variables in~$\bI^{+}$ and~$\bI^{-}$, respectively. Observe
that $f_{2}|_{\p_{y_{j}\gets\sigma}}=f_{2,\bI^{+},\bI^{-}}|_{\p'}$,
so it suffices to prove that for every two disjoint sets $I^{+},I^{-}\subseteq\left\{ x_{1},\ldots,x_{n}\right\} $
we have 
\begin{equation}
\Pr\!\left[L(f_{2,\bI^{+},\bI^{-}}|_{\p'})=1\;\middle|\;f_{1}|_{\p}=y_{j},\bI^{+}=I^{+},\bI^{-}=I^{-}\right]\le q\cdot\sqrt{L(f_{2})}.\label{eq:shrinkage-to-two-literals-conditioned-on-I}
\end{equation}
Let $I^{+},I^{-}\subseteq\left\{ x_{1},\ldots,x_{n}\right\} $ be
disjoint sets, and let $\cI$ be the event that $\bI^{+}=I^{+}$ and
$\bI^{-}=I^{-}$. For convenience, let $K=\left\{ x_{1},\ldots,x_{n}\right\} \backslash(I^{+}\cup I^{-})$
and $Y=\left\{ y_{1},\ldots,y_{m}\right\} \backslash\left\{ y_{j}\right\} $,
so $\p'$ is a random projection from $K$ to~$Y$ when conditioned
on $\cI$. To prove \expref{Equation}{eq:shrinkage-to-two-literals-conditioned-on-I},
it suffices to prove that $\p'$ is a $(q_{0},q_{1})$-fixing projection
when conditioned on the events~$\cI$ and~$f_{1}|_{\p}=y_{j}$,
and then the inequality will follow from \expref{Lemma}{lemma:shrinkage-to-single-literal}.
We first prove that $\p'$ is a $(q_{0},q_{1})$-fixing projection
when conditioning only on the event~$\cI$ (and not on $f_{1}|_{\p}=y_{j}$). 
\begin{proposition}
Conditioned on the event $\cI$, the projection $\p'$ is a $(q_{0},q_{1})$-fixing
projection. 
\end{proposition}

\begin{proof}
We prove that $\p'$ satisfies the definition of a fixing projection.
Let $\pi'$ be a projection from $K$ to $Y$, and let $x_{i}\in K$.
Let $\sigma\in\B$. We have 
\begin{align*}
 & \Pr\!\left[\p'(x_{i})\notin\B\text{ and }\p'_{\var(\p'(x_{i}))\gets\sigma}=\pi'\;\middle|\;\cI\right]\\
 & =\Pr\!\left[\p(x_{i})\notin\B\text{ and }\p_{\var(\p(x_{i}))\gets\sigma}|_{K}=\pi'\text{ and }\cI\right]/\Pr\!\left[\cI\right]\\
 & =\sum_{\substack{\pi\colon\pi|_{K}=\pi',\pi^{-1}(y_{j})=I^{+},\\
\pi^{-1}(\overline{y_{j}})=I^{-}
}
}\Pr\!\left[\p(x_{i})\notin\B\text{ and }\p|_{\var(\p(x_{i}))\gets\sigma}=\pi\right]/\Pr\!\left[\cI\right]\\
 & \le\sum_{\substack{\pi\colon\pi|_{K}=\pi',\pi^{-1}(y_{j})=I^{+},\\
\pi^{-1}(\overline{y_{j}})=I^{-}
}
}q_{\sigma}\cdot\Pr[\p=\pi]/\Pr\!\left[\cI\right] & \text{(since \ensuremath{\p} is \ensuremath{(q_{0},q_{1})}-fixing)}\\
 & =q_{\sigma}\cdot\Pr\!\left[\p|_{K}=\pi'\text{ and }\cI\right]/\Pr\!\left[\cI\right]\\
 & =q_{\sigma}\cdot\Pr\!\left[\p'=\pi'\mid\cI\right],
\end{align*}
as required. 
\end{proof}
We now prove that $\p'$ remains a $(q_{0},q_{1})$-fixing projection
when conditioning on $f_{1}|_{\p}=y_{j}$ in addition to $\cI$. The
crucial observation is that the event $f_{1}|_{\p}=y_{j}$ is essentially
a \emph{filter}, defined next. 
\begin{definition}
A set of projections $\cE$ from $x_{1},\ldots,x_{n}$ to $y_{1},\ldots,y_{m}$
is a \emph{filter} if it is closed under assignment to variables,
\ie, if for every $\pi\in\cE$, every variable~$y_{j}$, and every
bit~$\tau\in\B$, we have $\pi_{y_{j}\gets\tau}\in\cE$. 
\end{definition}

It turns out that the property of a projection being a $(q_{0},q_{1})$-fixing
projection is preserved when conditioning on filters. Formally: 
\begin{proposition}
\label{prop:conditioning-on-a-filter}Let $\cE$ be a filter and let $\p^{*}$
be a $(q_{0},q_{1})$-fixing projection. Then, $\p^{*}|\cE$ is a
$(q_{0},q_{1})$-fixing projection. 
\end{proposition}

\begin{proof}
Let $\pi^{*}$ be a projection, and let $x_{i}$ be a variable. Let
$\sigma\in\B$. We would like to prove that 
\[
\Pr\!\left[\p^{*}(x_{i})\notin\B\text{ and }\p_{\var(\p^{*}(x_{i}))\gets\sigma}^{*}=\pi^{*}\;\middle|\;\p^{*}\in\cE\right]\le q_{\sigma}\cdot\Pr\!\left[\p^{*}=\pi^{*}\mid\p^{*}\in\cE\right].
\]
If $\pi^{*}\notin\cE$, then both sides of the equation are equal
to zero: this is obvious for the right-hand side, and holds for the
left-hand side since if there is a projection~$\pi^{0}\in\cE$ and
a variable $y_{j}$ such that $\pi_{y_{j}\gets\sigma}^{0}=\pi^{*}$
then it must be the case that $\pi^{*}\in\cE$ by the definition of
a filter. Thus, we may assume that $\pi^{*}\in\cE$. Now, it holds
that 
\begin{align*}
 & \Pr\!\left[\p^{*}(x_{i})\notin\B\text{ and }\p_{\var(\p^{*}(x_{i}))\gets\sigma}^{*}=\pi^{*}\;\middle|\;\p^{*}\in\cE\right]\\
 & \le\Pr\!\left[\p^{*}(x_{i})\notin\B\text{ and }\p_{\var(\p^{*}(x_{i}))\gets\sigma}^{*}=\pi^{*}\right]/\Pr\!\left[\p^{*}\in\cE\right]\\
 & \le q_{\sigma}\cdot\Pr\!\left[\p^{*}=\pi^{*}\right]/\Pr\!\left[\p^{*}\in\cE\right] & \text{(\ensuremath{\p^{*}} is \ensuremath{(q_{0},q_{1})}-fixing)}\\
 & =q_{\sigma}\cdot\Pr\!\left[\left.\p^{*}=\pi^{*}\right|\p^{*}\in\cE\right], & \text{(\ensuremath{\pi^{*}\in\cE})}
\end{align*}
as required. 
\end{proof}
Consider the event $f_{1}|_{\p}=y_{j}$. Viewed as a set of projections
from $x_{1},\ldots,x_{n}$ to $y_{1},\ldots,y_{m}$, this event is
not a filter, since it is not closed under assignments to~$y_{j}$.
However, this event is closed under assignments \emph{to all variables
except~$y_{j}$}: when $f_{1}|_{\p}=y_{j}$, the equality continues
to hold even if the variables in~$Y$ are fixed to constants. Moreover,
observe that conditioned on $\cI$, the event $f_{1}|_{\p}=y_{j}$
depends only on the values that $\p$ assigns to~$K$. Thus, we can
view the event $f_{1}|_{\p}=y_{j}$ as a set of projections from $K$
to $Y$, and taking this view, \emph{this event is a filter}. Since
$\p'$ is a $(q_{0},q_{1})$-fixing projection from $K$ to $\left\{ y_{1},\ldots,y_{m}\right\} \backslash\left\{ y_{j}\right\} $
when conditioned on $\cI$, we conclude that it is a $(q_{0},q_{1})$-fixing
projection when conditioned on both $\cI$ and $f_{1}|_{\p}=y_{j}$.
It follows by \expref{Lemma}{lemma:shrinkage-to-single-literal} that 
\[
\Pr\!\left[L(f_{2,I^{+},I^{-}}|_{\p'})=1\;\middle|\;f_{1}|_{\p}=y_{j},\bI^{+}=I^{+},\bI^{-}=I^{-}\right]\le q\cdot\sqrt{L(f_{2,I^{+},I^{-}})}\le q\cdot\sqrt{L(f_{2})},
\]
as required.

\subsection{\label{subsec:reduction-to-size-2-with-substitution}Proof of \expref{Lemma}{lemma:reduction-to-size-2-with-substitution}}

Finally, we prove \expref{Lemma}{lemma:reduction-to-size-2-with-substitution},
restated next.
\begin{restated}{\expref{Lemma}{lemma:reduction-to-size-2-with-substitution}}
For every projection~$\pi$ we have 
\begin{equation}
L(\phi|_{\pi})\le\sum_{\text{internal gate \ensuremath{g} of \ensuremath{\phi}}}(\dpt(g)+2)\cdot1_{\cE_{g,\pi}}+1_{L(\phi|_{\pi})=1}.\label{eq:reduction-to-size-2-with-substitution}
\end{equation}
\end{restated}

Let $\pi$ be a projection. We prove that $\pi$ satisfies \expref{Equation}{eq:reduction-to-size-2-with-substitution}.
Recall that $\cE_{g,\pi}$ is the event that there exists some literal
$y_{j}^{\sigma}$ such that $\lc(g)|_{\pi}=y_{j}^{\sigma}$ and
\begin{itemize}
\item $L(\rc(g)|_{\pi_{y_{j}\gets\sigma}})=1$ if $g$ is an OR gate, or 
\item $L(\rc(g)|_{\pi_{y_{j}\gets\overline{\sigma}}})=1$ if $g$ is an
AND gate. 
\end{itemize}
We prove \expref{Equation}{eq:reduction-to-size-2-with-substitution} by induction.
If $\phi$ consists of a single leaf, then the upper bound clearly
holds. Otherwise, the root of~$\phi$ is an internal gate. Without
loss of generality, assume that the root is an OR gate. We denote
the %
subformulas  %
rooted at the left and right children of the root
by $\phi_{\ell}$ and $\phi_{r}$, respectively. We consider several
cases: 
\begin{itemize}
\item If $L(\phi|_{\pi})=1$, then the upper bound clearly holds. 
\item Suppose that $L(\phi_{\ell}|_{\pi})\ge2$. By the subadditivity of
formula complexity (\expref{Fact}{fact:subadditivity}), we have 
\[
L(\phi|_{\pi})\le L(\phi_{\ell}|_{\pi})+L(\phi_{r}|_{\pi}).
\]
By the induction hypothesis, we have 
\begin{align*}
L(\phi_{\ell}|_{\pi}) & \le\sum_{\text{internal gate }g\text{ of }\phi_{\ell}}(\dpt_{\phi_{\ell}}(g)+2)\cdot1_{\cE_{g,\pi}}+1_{L(\phi_{\ell}|_{\pi})=1}\\
 & =\sum_{\text{internal gate }g\text{ of }\phi_{\ell}}(\dpt_{\phi_{\ell}}(g)+2)\cdot1_{\cE_{g,\pi}} & (L(\phi_{\ell}|_{\pi})\ge2)\\
 & =\sum_{\text{internal gate }g\text{ of }\phi_{\ell}}(\dpt_{\phi}(g)+1)\cdot1_{\cE_{g,\pi}},
\end{align*}
where the equality holds since $\dpt_{\phi_{\ell}}(g)=\dpt_{\phi}(g)-1$
for every gate~$g$ of $\phi_{\ell}$. Since $L(\phi_{\ell}|_{\pi})\ge2$,
at least one of the terms in the last sum must be non-zero, so it
holds that 
\[
\sum_{\text{internal gate }g\text{ of }\phi_{\ell}}(\dpt_{\phi}(g)+1)\cdot1_{\cE_{g,\pi}}\le\sum_{\text{internal gate }g\text{ of }\phi_{\ell}}(\dpt_{\phi}(g)+2)\cdot1_{\cE_{g,\pi}}-1.
\]
Next, by the induction hypothesis we have 
\begin{align*}
L(\phi_{r}|_{\pi}) & \le\sum_{\text{internal gate }g\text{ of }\phi_{r}}(\dpt_{\phi_{r}}(g)+2)\cdot1_{\cE_{g,\pi}}+1_{L(\phi_{r}|_{\pi})=1}\\
 & \le\sum_{\text{internal gate }g\text{ of }\phi_{r}}(\dpt_{\phi}(g)+2)\cdot1_{\cE_{g,\pi}}+1.
\end{align*}
By combining the two bounds, we get that 
\begin{align*}
L(\phi|_{\pi}) & \le L(\phi_{\ell}|_{\pi})+L(\phi_{r}|_{\pi})\\
 & \le\sum_{\text{internal gate }g\text{ of }\phi_{\ell}}(\dpt_{\phi}(g)+2)\cdot1_{\cE_{g,\pi}}-1\\
 & \qquad+\sum_{\text{internal gate }g\text{ of }\phi_{r}}(\dpt_{\phi}(g)+2)\cdot1_{\cE_{g,\pi}}+1\\
 & \le\sum_{\text{internal gate }g\text{ of }\phi}(\dpt(g)+2)\cdot1_{\cE_{g,\pi}}\\
 & \le\sum_{\text{internal gate }g\text{ of }\phi}(\dpt(g)+2)\cdot1_{\cE_{g,\pi}}+1_{L(\phi|_{\pi})=1},
\end{align*}
as required. 
\item If $L(\phi_{r}|_{\pi})\ge2$, then we use the same argument of the
previous case by exchanging $\phi_{\ell}$ and $\phi_{r}$. 
\item Suppose that $L(\phi|_{\pi})\ge2$, $L(\phi_{\ell}|_{\pi})\le1$ and
$L(\phi_{r}|_{\pi})\le1$. Then, it must be the case that $L(\phi|_{\pi})=2$
and also that $L(\phi_{\ell}|_{\pi})=1$ and $L(\phi_{r}|_{\pi})=1$
(or otherwise $L(\phi|_{\pi})=1$). In particular, $\phi_{\ell}|_{\pi}$~is
equal to some literal $y_{j}^{\sigma}$. It follows that $\phi|_{\pi}=y_{j}^{\sigma}\vee\phi_{r}|_{\pi}$,
and by the one-variable simplification rules 
(\expref{Fact}{fact:simplification-rules}),
this function is equal to 
\[
y_{j}^{\sigma}\vee\left(\phi_{r}|_{\pi}\right)_{y_{j}\gets\sigma}=y_{j}^{\sigma}\vee\phi_{r}|_{\pi_{y_{j}\gets\sigma}}.
\]
Thus, it must be the case that $L(\phi_{r}|_{\pi_{y_{j}\gets\sigma}})=1$
(since $L(\phi|_{\pi})=2$). It follows that if we let $g$ be the
root of~$\phi$, then the event $\cE_{g,\pi}$ occurs and so 
\[
(\dpt(g)+2)\cdot1_{\cE_{g,\pi}}=2=L(\phi|_{\pi}),
\]
so the desired upper bound holds. 
\end{itemize}
We proved that the upper bound holds in each of the possible cases,
so the required result follows.

\section{\label{sec:shrinkage-hiding}Proof of shrinkage theorem for hiding
projections}

In this section we prove our the shrinkage theorem for hiding projections.
We start by recalling the definition of hiding projections and their
shrinkage theorem.
\begin{repdefinition}{\expref{Definition}{def:hiding-projections}}
Let $0\leq q_{0},q_{1}\leq1$. We say that a random projection $\p$
is a \emph{$(q_{0},q_{1})$-hiding projection} if for every projection
$\pi$, every bit $\sigma\in\B$, and all variables $x_{i},y_{j}$,
we have 
\[
\Pr\!\left[\p(x_{i})\in\left\{ y_{j},\overline{y_{j}}\right\} \;\middle|\;\p_{y_{j}\gets\sigma}=\pi\right]\leq q_{\sigma},
\]
whenever the event conditioned on has positive probability. For shorthand,
we say that $\p$ is a $q$-hiding projection, for $q=\sqrt{q_{0}q_{1}}$. 
\end{repdefinition}

\begin{restated}{\expref{Theorem}{thm:shrinkage-hiding}}[Shrinkage under hiding projections]
Let $\phi$ be a formula of size $s$ and depth $d$, and let $\p$
be a $q$-hiding projection. Then 
\[
\E\left[L(\phi|_{\p})\right]=O\!\left(m^{4}\cdot q^{2}\cdot d^{2}\cdot s+m^{2}\cdot q\cdot\sqrt{s}\right).
\]
\end{restated}

The crux of the proof of \expref{Theorem}{thm:shrinkage-hiding} is that hiding projections
can be converted into fixing projections. This is captured by the
following result. 
\begin{lemma}
\label{lemma:hiding-to-fixing}Let $\p$ be a $q$-hiding projection from
$x_{1},\ldots,x_{n}$ to $y_{1},\ldots,y_{m}$. Then, there exists
a $(4m^{2}\cdot q)$-fixing projection $\p'$ from $x_{1},\ldots,x_{n}$
to~$y_{1},\ldots,y_{m}$ and an event~$\cH$ of probability at least~$\frac{1}{2}$
such that $\p'=\p$ conditioned on~$\cH$. Furthermore, the event
$\cH$ is independent of $\p$.
\end{lemma}

We prove \expref{Lemma}{lemma:hiding-to-fixing} in \expref{Section}{subsec:hiding-to-fixing}.
By combining \expref{Lemma}{lemma:hiding-to-fixing} with our shrinkage theorem for
fixing projections, we obtain the following shrinkage theorem for
hiding projections. 
\begin{restated}{\expref{Theorem}{thm:shrinkage-hiding}}
Let $\phi$ be a formula of size $s$ and depth $d$, and let $\p$
be a $q$-hiding projection. Then 
\[
\E\left[L(\phi|_{\p})\right]=O\!\left(m^{4}\cdot q^{2}\cdot d^{2}\cdot s+m^{2}\cdot q\cdot\sqrt{s}\right).
\]
\end{restated}

\begin{proof}
Let $\p'$ and~$\cH$ be the $(4m^{2}\cdot q)$-fixing projection
and event that are obtained from~$\p$ by \expref{Lemma}{lemma:hiding-to-fixing}.
Since $\Pr\left[\cH\right]\geq\frac{1}{2}$, we have 
\[
\E\left[L(\phi|_{\p'})\right]\geq\Pr\!\left[\cH\right]\cdot\E\left[L(\phi|_{\p'})\;\middle|\;\cH\right]\geq\frac{1}{2}\cdot\E\left[L(\phi|_{\p})\;\middle|\;\cH\right]=\frac{1}{2}\cdot\E\left[L(\phi|_{\p})\right],
\]
where the second inequality holds since conditioned on~$\cH$ it
holds that $\p'=\p$, and the last equality holds since the event~$\cH$
is independent of~$\p$. On the other hand, by applying 
\expref{Theorem}{thm:shrinkage-fixing}
to~$\p'$, we obtain that 
\[
\E\left[L(\phi|_{\p'})\right]=O(m^{4}\cdot q^{2}\cdot d^{2}\cdot s+m^{2}\cdot q\cdot\sqrt{s}).
\]
The theorem follows by combining the two bounds. 
\end{proof}

\subsection{\label{subsec:hiding-to-fixing}Proof of \expref{Lemma}{lemma:hiding-to-fixing}}

We use the following straightforward generalization of the property
of hiding projections. 
\begin{claim}
\label{claim:generalized-hiding}Let $\p$ be a $(q_{0},q_{1})$-hiding
projection, and let $\cbE$ be a random set of projections that is
independent of~$\p$. Then, for every $\sigma\in\B$, we have
\[
\Pr\!\left[\p(x_{i})\in\left\{ y_{j},\overline{y_{j}}\right\} \;\middle|\;\p_{y_{j}\gets\sigma}\in\cbE\right]\le q_{\sigma}.
\]
\end{claim}
\begin{proof}
Let $\p$ and $\cbE$ be as in the claim, and let $\sigma\in\B$.
We have 
\begin{align*}
 & \Pr\!\left[\p(x_{i})\in\left\{ y_{j},\overline{y_{j}}\right\} \;\middle|\;\p_{y_{j}\gets\sigma}\in\cbE\right]\\
 & =\sum_{\pi}\Pr\!\left[\p(x_{i})\in\left\{ y_{j},\overline{y_{j}}\right\} \;\middle|\;\p_{y_{j}\gets\sigma}=\pi\text{~and~}\pi\in\cbE\right]\cdot\Pr\!\left[\p_{y_{j}\gets\sigma}=\pi\;\middle|\;\p_{y_{j}\gets\sigma}\in\cbE\right]\\
 & =\sum_{\pi}\Pr\!\left[\p(x_{i})\in\left\{ y_{j},\overline{y_{j}}\right\} \;\middle|\;\p_{y_{j}\gets\sigma}=\pi\right]\cdot\Pr\!\left[\p_{y_{j}\gets\sigma}=\pi\;\middle|\;\p_{y_{j}\gets\sigma}\in\cbE\right]\tag{\ensuremath{\p} and \ensuremath{\cbE} are independent}\\
 & \le\sum_{\pi}q_{\sigma}\cdot\Pr\!\left[\p_{y_{j}\gets\sigma}=\pi\;\middle|\;\p_{y_{j}\gets\sigma}\in\cbE\right]\tag{\ensuremath{\p} is \ensuremath{(q_{0},q_{1})}-hiding}\\
 & =q_{\sigma}.\qedhere
\end{align*}
\end{proof}

We turn to proving \expref{Lemma}{lemma:hiding-to-fixing},
restated next.
\begin{restated}{\expref{Lemma}{lemma:hiding-to-fixing}}
Let $\p$ be a $q$-hiding projection from $x_{1},\ldots,x_{n}$
to $y_{1},\ldots,y_{m}$. Then, there exists a $(4m^{2}\cdot q)$-fixing
projection $\p'$ from $x_{1},\ldots,x_{n}$ to~$y_{1},\ldots,y_{m}$
and an event~$\cH$ of probability at least~$\frac{1}{2}$ such
that $\p'=\p$ conditioned on~$\cH$. Furthermore, the event $\cH$
is independent of $\p$.
\end{restated}

Suppose that $\p$ is a $(q_{0},q_{1})$-hiding projection from $x_{1},\ldots,x_{n}$
to $y_{1},\ldots,y_{m}$. Let $\rs$ be a $\left(1-\frac{1}{2m}\right)$-random
restriction over $y_{1},\ldots,y_{m}$, \ie, $\rs$ is the random
projection from $y_{1},\ldots,y_{m}$ to $y_{1},\ldots,y_{m}$ that
assigns each $y_{j}$ independently as follows: 
\[
\rs(y_{j})=\begin{cases}
y_{j} & \text{with probability \ensuremath{1-\frac{1}{2m}}},\\
0 & \text{with probability \ensuremath{\frac{1}{4m}}},\\
1 & \text{with probability \ensuremath{\frac{1}{4m}}}.
\end{cases}
\]
For convenience, we define $\rs(0)=0$, $\rs(1)=1$, and $\rs(\overline{y_{j}})=\overline{\rs(y_{j})}$
for every $j\in\left[m\right]$. We now choose the random projection~$\p'$
to be the composition~$\rs\circ\p$, and define the event~$\cH$
to be the event that $\rs(y_{j})=y_{j}$ for every $j\in\left[m\right]$.
Observe that the event $\p'=\p$ occurs whenever~$\cH$ occurs, and
moreover 
\[
\Pr\!\left[\cH\right]=\left(1-\frac{1}{2m}\right)^{m}\ge1-\frac{m}{2m}=\frac{1}{2},
\]
as required. Moreover, the event~$\cH$ is independent of $\p$,
since $\rs$ is independent of $\p$. We show that $\p'$ is a $(4\cdot m^{2}\cdot q_{0},4\cdot m^{2}\cdot q_{1})$-fixing
projection. To this end, we should show that for every projection
$\pi'$, every bit $\sigma\in\B$, and every variable $x_{i}$, 
\begin{equation}
\Pr\!\left[\p'(x_{i})\notin\B\text{ and }\p_{\var(\p'(x_{i}))\gets\sigma}'=\pi'\right]\leq4\cdot m^{2}\cdot q_{\sigma}\cdot\Pr\!\left[\p'=\pi'\right].\label{eq:shrinkage-hiding-goal-after-summation}
\end{equation}
This is implied by the following inequality, which we will prove for
all $j\in\left[m\right]$: 
\begin{equation}
\Pr\!\left[\p'(x_{i})\in\left\{ y_{j},\overline{y_{j}}\right\} \text{ and }\p_{y_{j}\gets\sigma}'=\pi'\right]\leq4\cdot m\cdot q_{\sigma}\cdot\Pr\!\left[\p'=\pi'\right].\label{eq:shrinkage-hiding-goal}
\end{equation}
Let $j\in\left[m\right]$ and let $\rs_{-j}$ be the random projection
obtained by restricting~$\rs$ to the domain~$\left\{ y_{1},\ldots,y_{m}\right\} \setminus\left\{ y_{j}\right\} $.
First, observe that if $\rs(y_{j})=\sigma$ then $\p'=\rs\circ\p=\rs_{-j}\circ\p_{y_{j}\gets\sigma}$,
and therefore 
\[
\Pr\!\left[\p'=\pi'\right]\geq\Pr\!\left[\rs(y_{j})=\sigma\text{ and }\rs_{-j}\circ\p_{y_{j}\gets\sigma}=\pi'\right]=\frac{1}{4m}\cdot\Pr\!\left[\rs_{-j}\circ\p_{y_{j}\gets\sigma}=\pi'\right],
\]
where the equality holds since $\rs(y_{j})$ is independent of $\rs_{-j}$
and $\p$. In the rest of this section we will prove the following
inequality, which together with the last inequality will imply \expref{Equation}{eq:shrinkage-hiding-goal}:
\begin{align}
\Pr\!\left[\p'(x_{i})\in\left\{ y_{j},\overline{y_{j}}\right\} \text{ and }\p_{y_{j}\gets\sigma}'=\pi'\right] & \le q_{\sigma}\cdot\Pr\!\left[\rs_{-j}\circ\p_{y_{j}\gets\sigma}=\pi'\right].\label{eq:shrinkage-hiding-auxiliary-goal}
\end{align}
To this end, observe that the event $\p'(x_{i})\in\left\{ y_{j},\overline{y_{j}}\right\} $
happens if and only if $\p(x_{i})\in\left\{ y_{j},\overline{y_{j}}\right\} $
and $\rs(y_{j})=y_{j}$, and therefore 
\begin{align*}
 & \Pr\!\left[\p'(x_{i})\in\left\{ y_{j},\overline{y_{j}}\right\} \text{ and }\p_{y_{j}\gets\sigma}'=\pi'\right]\\
 & =\Pr\!\left[\p(x_{i})\in\left\{ y_{j},\overline{y_{j}}\right\} \text{ and }\rs(y_{j})=y_{j}\text{ and }\p_{y_{j}\gets\sigma}'=\pi'\right].
\end{align*}
Next, observe that if $\rs(y_{j})=y_{j}$ then $\p_{y_{j}\gets\sigma}'=\rs_{-j}\circ\p_{y_{j}\gets\sigma}$.
It follows that the latter expression is equal to 
\begin{align*}
 & \Pr\!\left[\p(x_{i})\in\left\{ y_{j},\overline{y_{j}}\right\} \text{ and }\rs(y_{j})=y_{j}\text{ and }\rs_{-j}\circ\p_{y_{j}\gets\sigma}=\pi'\right]\\
 & \le\Pr\!\left[\p(x_{i})\in\left\{ y_{j},\overline{y_{j}}\right\} \text{ and }\rs_{-j}\circ\p_{y_{j}\gets\sigma}=\pi'\right]\\
 & =\Pr\!\left[\p(x_{i})\in\left\{ y_{j},\overline{y_{j}}\right\} \;\middle|\;\rs_{-j}\circ\p_{y_{j}\gets\sigma}=\pi'\right]\cdot\Pr\!\left[\rs_{-j}\circ\p_{y_{j}\gets\sigma}=\pi'\right]\\
 & \le q_{\sigma}\cdot\Pr\!\left[\rs_{-j}\circ\p_{y_{j}\gets\sigma}=\pi'\right],
\end{align*}
where the last inequality follows by applying \expref{Claim}{claim:generalized-hiding}
with $\cbE$ being the set of projections $\pi$ that satisfy $\rs_{-j}\circ\pi=\pi'$.
This concludes the proof of \expref{Equation}{eq:shrinkage-hiding-auxiliary-goal}.
\begin{remark}
We can improve the bound $m^{2}$ in the statement of \expref{Lemma}{lemma:hiding-to-fixing}
to $mk$, where $k$ is the maximal number of variables $1,\ldots,m$
which could appear in any position of $\p$. The reason is that in
the latter case, the transition from \expref{Equation}{eq:shrinkage-hiding-goal}
to \expref{Equation}{eq:shrinkage-hiding-goal-after-summation} incurs a factor
of~$k$ rather than~$m$. This is useful, for example, since the
random projections $\p$ of \expref{Section}{subsec:KRW} have the feature that
for each $i\in[n]$ there is a unique $j\in[m]$ such that $\p(x_{i})\in\{0,1,y_{j},\overline{y_{j}}\}$,
and so for these projections $k=1$. 
\end{remark}

\section{\label{sec:join}Joining projections}

In this section, we define a \emph{join} operation on fixing and hiding
projections, and show that it preserves the corresponding properties.
This operation provides a convenient tool for constructing random
projections, and will be used in our applications.
\begin{definition}
Let $\alpha$ be a projection from $x_{1},\ldots,x_{n_{a}}$ to $y_{1},\ldots,y_{m_{a}}$,
and let $\beta$ be a projection from $w_{1},\ldots,w_{n_{b}}$ to
$z_{1},\ldots,z_{m_{b}}$. The \emph{join} $\alpha\uplus\beta$ is
the projection from $x_{1},\ldots,x_{n_{a}},w_{1},\ldots,w_{n_{b}}$
to $y_{1},\ldots,y_{m_{a}},z_{1},\ldots,z_{m_{b}}$ obtained by joining
$\alpha$ and~$\beta$ together in the obvious way. 
\end{definition}

\begin{lemma}
\label{lemma:concat-fixing}Let $\ga$ and $\gb$ be independent random
projections. If $\ga$ and~$\gb$ are $(q_{0},q_{1})$-fixing projections,
then so is $\ga\uplus\gb$. 
\end{lemma}

\begin{proof}
Let $\ga$ and~$\gb$ be $(q_{0},q_{1})$-fixing projections, and
let $\ggm=\ga\uplus\gb$. We prove that $\ggm$ is a $(q_{0},q_{1})$-fixing
projection. Let $\alpha$ be a projection from $x_{1},\ldots,x_{n_{a}}$
to $y_{1},\ldots,y_{m_{a}}$, let $\beta$ be a projection from $w_{1},\ldots,w_{n_{b}}$
to $z_{1},\ldots,z_{m_{b}}$, and let $\gamma=\alpha\uplus\beta$.
We should show that for every $\sigma\in\B$ and every input variable
$u$ (either $x_{i}$ or $w_{i}$) we have 
\[
\Pr\!\left[\ggm(u)\notin\B\text{ and }\ggm_{\var(\ggm(u))\gets\sigma}=\gamma\right]\leq q_{\sigma}\cdot\Pr\!\left[\ggm=\gamma\right].
\]
Let $\sigma\in\B$~be a bit, and let $u$~be an input variable.
Assume that $u=x_{i}$ for some $i\in\left[n_{a}\right]$ (if $u=w_{i}$
for some $i\in\left[n_{b}\right]$, the proof is similar). The above
equation is therefore equivalent to 
\[
\Pr\!\left[\ga(x_{i})\notin\B\text{ and }\ga_{\var(\ga(x_{i}))\gets\sigma}=\alpha\text{ and }\gb_{\var(\ga(x_{i}))\gets\sigma}=\beta\right]\leq q_{\sigma}\cdot\Pr\!\left[\ga=\alpha\text{ and }\gb=\beta\right].
\]
Since the ranges of $\ga$ and $\gb$ are disjoint, the variable $\var(\ga(x_{i}))$
does not appear in the range of~$\gb$, and therefore $\gb_{\var(\ga(x_{i}))\gets\sigma}=\gb$.
The independence of $\ga$ and $\gb$ shows that the above inequality
is equivalent to 
\[
\Pr\!\left[\ga(x_{i})\notin\B\text{ and }\ga_{\var(\ga(x_{i}))\gets\sigma}=\alpha\right]\cdot\Pr\!\left[\gb=\beta\right]\leq q_{\sigma}\cdot\Pr\!\left[\ga=\alpha\right]\cdot\Pr\!\left[\gb=\beta\right],
\]
which follows from $\ga$ being $(q_{0},q_{1})$-fixing. 
\end{proof}
\begin{lemma}
\label{lemma:concat-hiding}Let $\ga$ and $\gb$ be independent random
projections. If $\ga$ and $\gb$ are $(q_{0},q_{1})$-hiding projections,
then so is $\ga\uplus\gb$. 
\end{lemma}

\begin{proof}
Let $\ga$ and~$\gb$ be $(q_{0},q_{1})$-hiding projections, and
let $\ggm=\ga\uplus\gb$. We prove that $\ggm$ is a $(q_{0},q_{1})$-hiding
projection. Let $\alpha$ be a projection from $x_{1},\ldots,x_{n_{a}}$
to $y_{1},\ldots,y_{m_{a}}$, let $\beta$ be a projection from $w_{1},\ldots,w_{n_{b}}$
to $z_{1},\ldots,z_{m_{b}}$, and let $\gamma=\alpha\uplus\beta$.
We should show that for every $\sigma\in\B$, every input variable
$u$ (either $x_{i}$ or $w_{i}$), and every output variable~$v$
(either $y_{j}$ or $z_{j}$) we have 
\[
\Pr\!\left[\ggm(u)\in\left\{ v,\overline{v}\right\} \;\middle|\;\ggm_{v\gets\sigma}=\gamma\right]\leq q_{\sigma}.
\]
Let $\sigma\in\B$ be a bit, let $u$~be an input variable, and let
$v$~be an output variable. Assume that $u=x_{i}$ for some $i\in\left[n_{a}\right]$
(if $u=w_{i}$ for some $i\in\left[n_{b}\right]$, the proof is similar).
In this case, we may assume that $v=y_{j}$ for some $j\in\left[m_{a}\right]$,
since otherwise the probability on the left-hand side is~$0$. Thus,
the above equation is equivalent to 
\[
\Pr\!\left[\ga(x_{i})\in\left\{ y_{j},\overline{y_{j}}\right\} \;\middle|\;\ga_{y_{j}\gets\sigma}=\alpha\text{ and }\gb_{y_{j}\gets\sigma}=\beta\right]\leq q_{\sigma}.
\]
Since $\ga$ and $\gb$ are independent, whenever the conditioned
event has positive probability, we have 
\[
\Pr\!\left[\ga(x_{i})\in\left\{ y_{j},\overline{y_{j}}\right\} \;\middle|\;\ga_{y_{j}\gets\sigma}=\alpha\text{ and }\gb_{y_{j}\gets\sigma}=\beta\right]=\Pr\!\left[\ga(x_{i})\in\left\{ y_{j},\overline{y_{j}}\right\} \;\middle|\;\ga_{y_{i}\gets\sigma}=\alpha\right]\leq q_{\sigma},
\]
where the inequality holds since $\ga$ is $(q_{0},q_{1})$-hiding. 
\end{proof}

\section{\label{sec:hiding-from-complexity}Hiding projections from complexity
measures}

In this section we define a generalization of the unweighted quantum
adversary bound due to Ambainis~\cite{A02}, and show how it can
be used for constructing hiding projections. We also derive a relationship
between this measure to the Khrapchenko bound~\cite{K72}, and conclude
that the latter bound can be used for constructing hiding projections
as well. 

\subsection{\label{sec:adversary}The soft-adversary method}

Ambainis~\cite{A02} defined the following complexity measure for
Boolean functions, called the \emph{unweighted quantum adversary bound},
and proved that it is a lower bound on quantum query complexity (for
a definition of quantum query complexity, see~\cite{A02} or~\cite{BW02}).
\begin{definition}
\label{def:ambainis}Let $f\colon\B^{n}\to\B$ be a non-constant function.
Let $R\subseteq f^{-1}(1)\times f^{-1}(0)$, and let $A,B$ be the
projections of~$R$ into the first and second coordinates, respectively.
Let $R(a,B)=\{b\in B:(a,b)\in R\}$, let $R_{i}(a,B)=\{b\in B:(a,b)\in R,a_{i}\neq b_{i}\}$,
and define $R(A,b),R_{i}(A,b)$ analogously. The \emph{unweighted
quantum adversary bound of $f$} is 
\[
\Amb(f)=\max_{R\subseteq f^{-1}(1)\times f^{-1}(0)}\sqrt{\frac{\min_{a\in A}|R(a,B)|\cdot\min_{b\in B}|R(A,b)|}{\max_{\substack{a\in A\\
i\in[n]
}
}|R_{i}(a,B)|\cdot\max_{\substack{b\in B\\
i\in[n]
}
}|R_{i}(A,b)|}.}
\]
\end{definition}

\begin{theorem}[\cite{A02}]
The quantum query complexity of a non-constant function $f\colon\B^{n}\to\B$
is $\Omega(\Amb(f))$. 
\end{theorem}

We define the following generalization of the unweighted quantum adversary
bound.
\begin{definition}
\label{def:soft-ambainis}Let $f\colon\B^{n}\to\B$ be a Boolean function.
We define the \emph{soft-adversary bound of $f$}, which we denote $\Am$,
to be the maximum of the quantity 
\[
\sqrt{\min_{\substack{a\in\supp(\a)\\
i\in[n]
}
}\frac{1}{\Pr[\a_{i}\neq\b_{i}\mid\a=a]}\cdot\min_{\substack{b\in\supp(\b)\\
i\in[n]
}
}\frac{1}{\Pr[\a_{i}\neq\b_{i}\mid\b=b]}}
\]
over all distributions $(\a,\b)$ supported on $f^{-1}(1)\times f^{-1}(0)$. 
\end{definition}

We chose to call this notion ``soft-adversary bound'' since we view
it as a variant of the unweighted adversary bound in which, instead
of choosing for each pair $(a,b)$ whether it belongs to $R$ or not
(a ``hard'' decision), we assign each pair $(a,b)$ a probability
of being in the relation $R$ (a ``soft'' decision). We now show
that this bound indeed generalizes the unweighted quantum adversary
method.
\begin{proposition}
\label{prop:Amb-Am}Let $f\colon\B^{n}\to\B$ be non-constant function.
Then $\Am(f)\geq\Amb(f)$. 
\end{proposition}

\begin{proof}
Let $R\subseteq f^{-1}(1)\times f^{-1}(0)$ be a relation that attains
$\Amb(f)$. Let $(\a,\b)$ be a uniformly distributed element of $R$.
Using $\Pr[\a_{i}\neq\b_{i}\mid\a=a]=|R_{i}(a,B)|/|R(a,B)|$ and $\Pr[\a_{i}\neq\b_{i}\mid\b=b]=|R_{i}(A,b)|/|R(A,b)|$,
it is easy to check that the quantity in the definition of $\Am(f)$
is %
at least $\Amb(f)$.  %
\end{proof}
Following \cite{SS06,LLS06}, it can be shown that the soft-adversary
bound is a lower bound on formula complexity. We have the following
result, which is proved in \expref{Section}{sec:L-Am}. %
\begin{proposition}
\label{prop:L-Am} For any non-constant function $f\colon\B^{n}\to\B$
we have $L(f)\geq\Am^{2}(f)$. 
\end{proposition}

Our motivation for introducing the soft-adversary bound is that it
can be used for constructing hiding projections. We have the following
result.
\begin{lemma}
\label{lemma:ambainis-hiding}Let $f\colon\B^{n}\to\B$ be a non-constant
Boolean function. There is a $q$-hiding projection $\p$ to a single
variable~$y$ such that $q=1/\Am(f)$ and such that $f|_{\p}$~is
a non-constant function with probability~$1$. 
\end{lemma}

\begin{proof}
Let $(\a,\b)$ be a distribution supported on $f^{-1}(1)\times f^{-1}(0)$
which attains $\Am(f)$. We construct a $q$-hiding projection~$\p$
from $x_{1},\ldots,x_{n}$ to a single variable~$y$ for $q=1/\Am(f)$,
such that $\p_{y\gets1}=\a$ and $\p_{y\gets0}=\b$. Note that this
implies in particular that $f|_{\p}$ is a non-constant function,
since $\a\in f^{-1}(1)$ and $\b\in f^{-1}(0)$. Specifically, for
every input variable~$x_{i}$ of~$f$, we define~$\p$ as follows: 
\begin{itemize}
\item If $\a_{i}=\b_{i}$ then $\p(x_{i})=\a_{i}$. 
\item If $\a_{i}=1$ and $\b_{i}=0$ then $\p(x_{i})=y$. 
\item If $\a_{i}=0$ and $\b_{i}=1$ then $\p(x_{i})=\overline{y}$. 
\end{itemize}
It is not hard to see that indeed $\p_{y\gets1}=\a$ and $\p_{y\gets0}=\b$.
We now show that $\p$ is $1/\Am(f)$-hiding. To this end, we show
that $\p$ is $(q_{0},q_{1})$-hiding for 
\[
q_{0}=\max_{\substack{b\in\supp(\b)\\
i\in[n]
}
}\Pr[\a_{i}\neq\b_{i}\mid\b=b],\qquad q_{1}=\max_{\substack{a\in\supp(\a)\\
i\in[n]
}
}\Pr[\a_{i}\neq\b_{i}\mid\a=a].
\]
Note that indeed $\sqrt{q_{0}q_{1}}=1/\Am(f)$. In order to prove
that $\p$ is $(q_{0},q_{1})$-hiding, we need to prove that 
\[
\Pr\!\left[\p(x_{i})\in\left\{ y,\overline{y}\right\} \;\middle|\;\p_{y\gets\sigma}=\pi\right]\le q_{\sigma}
\]
for every $i\in\left[n\right]$ and every $\sigma\in\B$. To this
end, observe that the event $\p(x_{i})\in\left\{ y,\overline{y}\right\} $
occurs if and only if $\a_{i}\neq\b_{i}$, and that $\p_{y\gets1},\p_{y\gets0}$
are equal to the strings $\a,\b$ respectively. It follows that
\begin{align*}
\Pr\!\left[\p(x_{i})\in\left\{ y,\overline{y}\right\} \;\middle|\;\p_{y\gets1}=\pi\right] & =\Pr\!\left[\a_{i}\ne\b_{i}\;\middle|\;\a=\pi\right]\le q_{1}\\
\Pr\!\left[\p(x_{i})\in\left\{ y,\overline{y}\right\} \;\middle|\;\p_{y\gets0}=\pi\right] & =\Pr\!\left[\a_{i}\ne\b_{i}\;\middle|\;\b=\pi\right]\le q_{0}
\end{align*}
as required.
\end{proof}
\begin{remark}
Several generalizations of Ambainis' original quantum adversary bound
have appeared in the literature. {\v{S}}palek and Szegedy~\cite{SS06}
showed that all of these methods are equivalent, and so they are known
collectively as the strong quantum adversary bound. The formulation
of the strong quantum adversary bound in Ambainis~\cite{A06} makes
it clear that it subsumes our soft version. Laplante, Lee and Szegedy~\cite{LLS06}
showed that the strong quantum adversary bound
is a lower bound on %
formula complexity, and
they  %
came up with an even stronger measure, $\mathsf{maxPI}^{2}$,
that %
is still a lower bound on  %
formula complexity, but no longer %
a lower bound on  %
quantum query complexity.

H{\o}yer, Lee and {\v{S}}palek~\cite{HLS07} generalized the strong quantum
adversary bound, coming up with a new measure known as the general
adversary bound, which %
is a lower bound both on %
quantum query complexity
and (after squaring)
on %
formula complexity. Reichardt~\cite{R09,R11,LMRSS11,R14}
showed that the general adversary bound in fact \emph{coincides} with
quantum query complexity (up to constant factors). These results are
described in a recent survey by Li and Shirley~\cite{LS20}.
\end{remark}

\begin{remark}
\label{remark:quadratic-barrier}We note that
$\Am(f) \le n$ (always).   %
In order to prove this, we show that
each factor in \expref{Definition}{def:soft-ambainis} is at most~$n$.  %
In particular, observe that it always
holds that $\a\ne\b$, and hence for every possible value~$a$ of~$\a$,
we have
\[
\sum_{i=1}^{n}\Pr[\a_{i}\neq\b_{i}\mid\a=a]\ge\Pr[\exists i\in[n]\text{ s.t. }\a_{i}\neq\b_{i}\mid\a=a]=1.
\]
By averaging, there must be a coordinate $i\in\left[n\right]$ such
that $\Pr[\a_{i}\neq\b_{i}\mid\a=a]\ge\frac{1}{n}$, and therefore
\[
\min_{\substack{a\in\supp(\a)\\
i\in[n]
}
}\frac{1}{\Pr[\a_{i}\neq\b_{i}\mid\a=a]}\le n \,. %
\]
The upper bound for the other factor in \expref{Definition}{def:soft-ambainis} 
can be proved similarly.
\end{remark}

\subsection{\label{sec:khrapchenko}Khrapchenko bound}

Khrapchenko~\cite{K72} defined the following complexity measure
for Boolean functions, and proved that it is a lower bound on formula
complexity. 
\begin{definition}
\label{def:khrapchenko}Let $f\colon\B^{n}\to\B$ be a non-constant function.
For every two sets $A\subseteq f^{-1}(1)$ and $B\subseteq f^{-1}(0)$,
let $E(A,B)$ be the set of pairs $(a,b)\in A\times B$ such that
$a$ and~$b$ differ on exactly one coordinate. The \emph{Khrapchenko
bound of~$f$} is 
\[
\K(f)=\max_{A\subseteq f^{-1}(1),B\subseteq f^{-1}(0)}\frac{\left|E(A,B)\right|^{2}}{\left|A\right|\cdot\left|B\right|}.
\]
\end{definition}

\begin{theorem}[\cite{K72}]
\label{thm:khrapchenko-formula-size}For every non-constant function
$f\colon\B^{n}\to\B$ we have $L(f)\ge\K(f)$.
\end{theorem}

Laplante, Lee, and Szegedy~\cite{LLS06} observed that the strong
adversary method generalizes the Khrapchenko bound. We show that this
holds even for the unweighted adversary bound, up to a constant factor.
Specifically, we have the following result, which is proved in \expref{Section}{sec:K-Km}.
\begin{proposition}
\label{prop:K-Amb}For any non-constant function $f\colon\B^{n}\to\B$
we have $\Amb(f)\geq\frac{\sqrt{\K(f)}}{4}$.
\end{proposition}

By combining \expref{Propositions}{prop:Amb-Am} \expref{and}{prop:K-Amb}
with \expref{Lemma}{lemma:ambainis-hiding}, it follows
that the Khrapchenko bound too can be used for constructing hiding
projections.
\begin{corollary}
\label{cor:khrapchenko-hiding}Let $f\colon\B^{n}\to\B$ be a non-constant
Boolean function. There is a $q$-hiding projection $\p$ to a single
variable~$y$ such that $q=\sqrt{4/\K(f)}$ and such that $f|_{\p}$~is
a non-constant function with probability~$1$.
\end{corollary}

In \expref{Section}{sec:min-khrapchenko}, we describe a new complexity measure,
\emph{the min-entropy Khrapchenko bound}, which generalizes the Khrapchenko
bound in the same way the soft-adversary bound generalizes the unweighted
adversary bound. As we show there, the min-entropy Khrapchenko bound
provides a simple way to prove lower bounds on the soft-adversary
bound.

\section{\label{sec:applications}Applications}

In this section, we apply our shrinkage theorems to obtain new results
regarding the KRW conjecture and the formula complexity of $\ACz$.
First, in \expref{Section}{subsec:KRW}, we prove the KRW conjecture for inner
functions for which the soft-adversary bound is tight. We use this
version of the KRW conjecture in \expref{Section}{subsec:AC0-bounds} to prove
cubic formula lower bounds for a function in $\ACz$. Finally, we
rederive some closely related known results in \expref{Section}{subsec:known-results}.

\subsection{\label{subsec:KRW}Application to the KRW conjecture}

Given two Boolean functions $f\colon\B^{m}\to\B$ and $g\colon\B^{n}\to\B$,
their (block-)\-composition is the function $f\d g\colon\left(\B^{n}\right)^{m}\to\B$
defined by 
\[
(f\d g)(x_{1},\ldots,x_{m})=f\left(g(x_{1}),\ldots,g(x_{m})\right),
\]
where $x_{1},\ldots,x_{m}\in\B^{n}$. It is easy to see that $L(f\d g)\le L(f)\cdot L(g)$.
The KRW conjecture~\cite{KRW95} asserts that this is roughly optimal,
namely, that $L(f\d g)\gtrsim L(f)\cdot L(g)$. In this section,
we prove the following related result. 
\begin{theorem}
\label{thm:composition-theorem}Let $f\colon\B^{m}\to\B$ and $g\colon\B^{n}\to\B$
be non-constant Boolean functions. Then, 
\[
L\left(f\d g\right)^{1+O\left(\frac{1}{\sqrt{\log L(f\d g)}}\right)}\geq\tfrac{1}{O(m^{4})}\cdot L(f)\cdot\Am^{2}(g).
\]
\end{theorem}

\begin{proof}
Let $\p$ be the $q$-hiding projection constructed for~$g$ in \expref{Lemma}{lemma:ambainis-hiding}
with $q=1/\Am(g)$, and recall that $g|_{\p}$ is non-constant with
probability~$1$. Let $\p_{1},\ldots,\p_{m}$ be independent copies
of~$\p$, and let $\p^{m}$ denote their $m$-fold join. Observe
that $\p^{m}$~is $q$-hiding according to \expref{Lemma}{lemma:concat-hiding}.
Applying \expref{Corollary}{cor:shrinkage-without-depth-hiding} 
and using the estimate
$x+\sqrt{x}=O(x+1)$, we see that $s=L(f\d g)$ satisfies 
\[
\E\left[L\left((f\d g)|_{\p^{m}}\right)\right]=m^{4}\cdot\frac{1}{\Am^{2}(g)}\cdot s^{1+O\bigl(\frac{1}{\sqrt{\log s}}\bigr)}+O(1).
\]
On the other hand, it is not hard to see that 
\[
\left((f\d g)|_{\p^{m}}\right)(x_{1},\ldots,x_{m})=f\left(g|_{\p_{1}}(x_{1}),\ldots,g|_{\p_{m}}(x_{m})\right),
\]
and since $g|_{\p_{1}},\ldots,g|_{\p_{m}}$ are non-constant, the
function~$f$ reduces to $(f\d g)|_{\p^{m}}$. In particular, 
\[
L(f)\le\E\left[L\left((f\d g)|_{\p^{m}}\right)\right]=m^{4}\cdot\frac{1}{\Am^{2}(g)}\cdot s^{1+O\bigl(\frac{1}{\sqrt{\log s}}\bigr)}+O(1).
\]
This means that there exists a constant $C>0$ such that for all non-constant $f$ and $g$
\[
L(f)\le m^{4}\cdot\frac{1}{\Am^{2}(g)}\cdot L(f\d g)^{1+\frac{C}{\sqrt{\log L(f\d g)}}}+C.\]
We consider two cases depending on whether $L(f)\ge 2C$ or not. 
If $L(f) \ge 2C$, then the above inequality implies
\[\frac{L(f)}{2}\le m^{4}\cdot\frac{1}{\Am^{2}(g)}\cdot L(f\d g)^{1+\frac{C}{\sqrt{\log L(f\d g)}}},\]
 and we obtain the theorem by rearranging. 

If $L(f) < 2C$, then we have $L(f\d g) \ge L(g)$ since $g$ or $\neg g$ is a sub-function of $f\circ g$ (as follows from the fact that $f$ is non-constant).
This means that in this case 
\[L(f\d g) \ge L(g) > \frac{L(f) \cdot L(g)}{2C} \ge \frac{L(f) \cdot \Am^2(g)}{2C}
\]
which is stronger than the required conclusion.
\end{proof}
A direct consequence of the theorem is the following corollary, which
is a special case of the KRW conjecture for inner functions~$g$
for which the soft-adversary bound is almost tight. 
\begin{corollary}
\label{cor:KRW-special-case}Let $f\colon\B^{m}\to\B$ and $g\colon\B^{n}\to\B$
be non-constant Boolean functions such that $\Am^{2}(g)\ge L(g)^{1-O\left(\frac{1}{\sqrt{\log L(g)}}\right)}$.
Then, 
\[
L(f\d g)\ge\tfrac{1}{O(m^{4})}\cdot L\left(f\right)^{1-O\left(\frac{1}{\sqrt{\log L(f)}}\right)}\cdot L(g)^{1-O\left(\frac{1}{\sqrt{\log L(g)}}\right)}.
\]
\end{corollary}

\begin{proof}
By substituting the assumption on~$g$ in the bound of \expref{Theorem}{thm:composition-theorem},
we obtain that 
\[
L\left(f\d g\right)^{1+O\left(\frac{1}{\sqrt{\log L(f\d g)}}\right)}\geq\tfrac{1}{O(m^{4})}\cdot L(f) \cdot\left(L(g)\right)^{1-O\left(\frac{1}{\sqrt{\log L(g)}}\right)}.
\]
Moreover, since $L(f\d g)\le L(f)\cdot L(g)$, we have 
\begin{align*}
L\left(f\d g\right)^{1+O\left(\frac{1}{\sqrt{\log L(f\d g)}}\right)} & \le L(f\d g)\cdot\left(L(f)\cdot L(g)\right)^{O\left(\frac{1}{\sqrt{\log L(f\d g)}}\right)}\\
 & \le L(f\d g)\cdot L\left(f\right)^{O\left(\frac{1}{\sqrt{\log L(f\d g)}}\right)}\cdot L\left(g\right)^{O\left(\frac{1}{\sqrt{\log L(f\d g)}}\right)}\\
 & \le L(f\d g)\cdot L\left(f\right)^{O\left(\frac{1}{\sqrt{\log L(f)}}\right)}\cdot L\left(g\right)^{O\left(\frac{1}{\sqrt{\log L(g)}}\right)}.
\end{align*}
The corollary follows by combining the two bounds. 
\end{proof}
Using the fact that $\K(g)\le\Am^{2}(g)$ (\expref{Proposition}{prop:K-Amb}), we obtain
the following immediate corollary. 
\begin{corollary}
\label{cor:KRW-original-khrapchenko}Let $f\colon\B^{m}\to\B$ and $g\colon\B^{n}\to\B$
be non-constant Boolean functions. Then, 
\[
L\left(f\d g\right)^{1+O\left(\frac{1}{\sqrt{\log L(f\d g)}}\right)}\geq\tfrac{1}{O(m^{4})}\cdot L(f) \cdot\K(g).
\]
Moreover, if $\K(g)\ge L(g)^{1-O\left(\frac{1}{\sqrt{\log L(g)}}\right)}$,
then 
\[
L(f\d g)\ge\tfrac{1}{O(m^{4})}\cdot L\left(f\right)^{1-O\left(\frac{1}{\sqrt{\log L(f)}}\right)}\cdot L(g)^{1-O\left(\frac{1}{\sqrt{\log L(g)}}\right)}.
\]
\end{corollary}

\subsection{\label{subsec:AC0-bounds}Formula lower bounds for \texorpdfstring{$\protect\ACz$}{AC0}}

In this section, we derive our second main application: cubic formula
lower bounds for $\ACz$. Formally, we have the following result.
\begin{restated}{\expref{Theorem}{thm:main}}
There exists a family of Boolean functions $h_{n}\colon\B^{n}\to\B$ for
$n\in\N$ such that 
\begin{enumerate}
\item $h_{n}$ can be computed by uniform depth-$4$ unbounded fan-in formulas
of size $O(n^{3})$. 
\item The formula size of $h_{n}$ is at least $n^{3-o(1)}$.
\end{enumerate}
\end{restated}

The function $h_{n}$ is constructed similarly to the Andreev function
\cite{A87}, with the parity function replaced with the surjectivity
function~\cite{BM12}, defined next.
\begin{definition}
\label{def:surjectivity}Let $\Sigma$ be a finite alphabet, and let $r\in\N$
be such that $r\ge\left|\Sigma\right|$. The \emph{surjectivity function}
$\surj_{\Sigma,r}\colon\Sigma^{r}\to\{0,1\}$ interprets its input
as a function from $\left[r\right]$ to $\Sigma$, and outputs whether
the function is surjective. In other words, $\surj(\sigma_{1},\ldots,\sigma_{r})=1$
if and only if every symbol in~$\Sigma$ appears in~$(\sigma_{1},\ldots,\sigma_{r})$.
\end{definition}

In order to prove \expref{Theorem}{thm:main}, we will use \expref{Theorem}{thm:composition-theorem}
with the inner function~$g$ being $\surj_{n}$. To this end, we
use the following result of~\cite{BM12}.
\begin{lemma}[\cite{BM12}]
\label{lemma:surj-amb}For every natural number $s$ such that $s\ge2$,
we have $\Amb(\surj_{\left[s\right],2s-2})=\Omega\bigl(s\bigr)$. 
\end{lemma}

Note that we can view the inputs of~$\surj_{\Sigma,r}$ as binary
strings of length~$n=r\cdot\left\lceil \log\left|\Sigma\right|\right\rceil $
by fixing some arbitrary binary encoding of the symbols in~$\Sigma$.
In what follows, we extend the definition of $\surj$ to every sufficiently
large input length~$n$ as follows: Given $n\in\N$, we choose~$s$
to be the largest number such that $(2s-2)\cdot\left\lceil \log s\right\rceil \le n$,
and define $\surj_{n}$ to be the function that interprets its input
as a string in $\left[s\right]^{2s-2}$ and computes $\surj_{\left[s\right],2s-2}$
on it. \expref{Lemma}{lemma:surj-amb} now implies that for every sufficiently large
$n\in\N$ we have $\Amb(\surj_{n})=\Omega\left(\frac{n}{\log n}\right)$.
In particular, this implies the same bound for the soft-adversary
bound, \ie, $\Am(\surj_{n})=\Omega\left(\frac{n}{\log n}\right)$.
For completeness, we provide an alternative proof for this lower bound
on $\Am(\surj_{n})$ in \expref{Section}{sec:min-khrapchenko}. We turn to prove
\expref{Theorem}{thm:main} using our special case of the KRW conjecture 
(\expref{Corollary}{cor:KRW-special-case}). 
\begin{proof}[Proof of \expref{Theorem}{thm:main}]
We would like to construct, for every sufficiently large~$n\in\N$,
a function~$h_{n}\colon\B^{n}\to\B$ that is computable by a formula
with unbounded fan-in of depth~$4$ and size~$O(\frac{n^{3}}{\log^{3}n})$
such that $L(h_{n})=\Omega(n^{3-o(1)})$. We start by constructing
the function~$h_{n}$ for input lengths~$n$ of a special form,
and then extend the construction to all sufficiently large input lengths.

First, assume that the input length~$n$ is of the form $2^{k}\cdot(k+1)$,
where $k\in\N$ is sufficiently large such that $\surj_{2^{k}}$ is
well-defined. The function~$h_{n}\colon\B^{n}\to\B$ takes two inputs:
the truth table of a function~$f\colon\B^{k}\to\B$, and $k$~strings
$x_{1},\ldots,x_{k}\in\B^{2^{k}}$. On such inputs, the function~$F$
outputs 
\[
h_{n}(f,x_{1},\ldots,x_{k})=(f\d\surj_{2^{k}})(x_{1},\ldots,x_{k})=f\left(\surj_{2^{k}}(x_{1}),\ldots,\surj_{2^{k}}(x_{k})\right).
\]
It is easy to see that the function~$h_{n}$ has input length~$n$.
We show that $F$ can be computed by a formula with unbounded fan-in
of depth~$4$ and size~$O(\frac{n^{3}}{\log^{3}n})$. We start by
constructing a formula for $\surj_{2^{k}}$. Recall that $\surj_{2^{k}}$
takes as input (the binary encoding of) a string $(\sigma_{1},\ldots,\sigma_{r})\in\left[s\right]^{r}$,
where $r,s=O(\frac{2^{k}}{k})$. For simplicity, let us assume that
every number in~$\left[s\right]$ is encoded by exactly one binary
string, and that if the binary input to $\surj_{2^{k}}$ contains
a binary string in~$\B^{\left\lceil \log s\right\rceil }$ that does
not encode any number in~$\left[s\right]$, then we do not care what
the formula outputs. Now, observe that 
\begin{equation}
\surj_{2^{k}}(\sigma_{1},\ldots,\sigma_{r})=\bigwedge_{\gamma\in\left[s\right]}\bigvee_{j=1}^{r}(\sigma_{j}=\gamma).\label{eq:surjectivity-formula}
\end{equation}
It is not hard to see that the expression on the right-hand side can
be implemented by a formula with unbounded fan-in of depth~$3$ and
size $s\cdot r\cdot\left\lceil \log s\right\rceil =O\!\left(\frac{2^{2k}}{k}\right)$.
Next, observe that 
\[
h_{n}(f,x_{1},\ldots,x_{k})=\bigvee_{y\in\B^{k}}\left[\left(f(y)=1\right)\wedge\left(\bigwedge_{i=1}^{k}\surj_{2^{k}}(x_{i})=y_{i}\right)\right].
\]
Using the foregoing formula for surjectivity, it is not hard to see
that the expression on the right-hand side can be implemented by a
formula with unbounded fan-in of depth~$4$ and size 
\[
O\!\left(2^{k}\cdot k\cdot\frac{2^{2k}}{k}\right)=O\!\left(2^{3k}\right)=O\!\left(\frac{n^{3}}{\log^{3}n}\right),
\]
as required.

Finally, we prove that $L(h_{n})=\Omega(n^{3-o(1)})$. To this end,
let us hardwire the input~$f$ to be some function from~$\B^{k}$
to $\B$ such that $L(f)=\Omega(\frac{2^{k}}{\log k})$ (such a function
exists by a well-known counting argument, see \cite[Theorem 1.23]{JuknaBook}).
After hardwiring the input~$f$, the function~$h_{n}$ becomes exactly
the function $f\d\surj_{2^{k}}$. Now, by observing that $\Am^{2}(\surj_{2^{k}})=\tilde{\Omega}\left(L(\surj_{2^{k}})\right)$
and applying \expref{Corollary}{cor:KRW-special-case}, it follows that 
\begin{align*}
L(F) & \ge\tfrac{1}{O(k^{4})}\cdot L\left(f\right)^{1-O\left(\frac{1}{\sqrt{\log L(f)}}\right)}\cdot L(\surj_{2^{k}})^{1-O\left(\frac{1}{\sqrt{\log L(\surj_{2^{k}})}}\right)}\\
 & \ge\tfrac{1}{O(k^{4})}\cdot\left(\frac{2^{k}}{\log k}\right)^{1-O\left(\frac{1}{\sqrt{k}}\right)}\cdot\left(\frac{2^{2k}}{k^{2}}\right)^{1-O\left(\frac{1}{\sqrt{k}}\right)}\\
 & \ge2^{3k-o(1)}\\
 & =n^{3-o(1)},
\end{align*}
as required.

It remains to deal with input lengths~$n$ that are not of the form
$2^{k}\cdot(k+1)$. For such input lengths~$n$, we choose $k$ to
be the largest natural number such that $2^{k}\cdot(k+1)\le n$, and
proceed as before. It can be verified that for this choice of~$k$
we have $2^{k}\cdot(k+1)=\Theta(n)$, and therefore all the
foregoing asymptotic bounds continue to hold. 
\end{proof}
\begin{remark}
We note that our methods cannot prove formula lower bounds that are
better than cubic. As explained in \expref{Remark}{remark:quadratic-barrier}, 
it holds
that %
$\Am^2(f) \le n^{2}$.   %
Thus, one cannot expect
to obtain a lower bound that is better than cubic by combining these
measures with Andreev's argument. 
\end{remark}

\subsection{\label{subsec:known-results}Reproving known formula lower bounds}

The proof of \expref{Theorem}{thm:main} combines a lower bound on $\Am(\surj_{n})$
with an upper bound on $L(\surj_{n})$. More generally, we can prove
the following result along similar lines. 
\begin{theorem}
\label{thm:andreev-general}Let $g\colon\B^{n}\to\B$ be an arbitrary
function, and let $k$ be an integer satisfying $k=\log n+O(1)$.
Let $F\colon\B^{2^{k}+nk}\to\B$ be a function on $N=\Theta(n\log n)$
variables, whose input consists of a function $f\colon\B^{k}\to\B$
and $k$ strings $x_{1},\ldots,x_{k}\in\B^{n}$, given by 
\[
F(f,x_{1},\ldots,x_{k})=(f\d g)(x_{1},\ldots,x_{k}).
\]
We call $F$ the \emph{$g$-based Andreev function.} 
\begin{enumerate}
\item If either $\Am^{2}(g)$ or $\K(g)$ are at least $\Omega(n^{2-o(1)})$
then $L(F)=\Omega(N^{3-o(1)})$. 
\item If $L(g)=O(n^{2+o(1)})$ then $L(F)=O(N^{3+o(1)})$. 
\end{enumerate}
\end{theorem}

The proof of \expref{Theorem}{thm:andreev-general} is very similar to the proof
of \expref{Theorem}{thm:main}, and so we leave it to the reader (the only difference
is that we may use \expref{Corollary}{cor:KRW-original-khrapchenko} instead of \expref{Corollary}{cor:KRW-special-case}).
Using \expref{Theorem}{thm:andreev-general}, we can derive two known special
cases: the original Andreev function (in which $g$ is Parity), and
the variant considered in~\cite{GTN19}. In particular, using the
known facts that $\K(\parity)\ge n^{2}$ and $\K(\majority)\ge\Omega(n^{2})$~\cite{K72},
we obtain the following results. 
\begin{corollary}
\label{cor:parity}The formula complexity of the $\parity_{n}$-based
Andreev function is $\Theta(N^{3\pm o(1)})$. 
\end{corollary}

\begin{corollary}
\label{cor:majority} For $m\geq3$ odd, let $\majority_{n}\colon\B^{m}\to\B$
be the Majority function. The formula complexity of the $\majority_{n}$-based
Andreev function is $\Omega(N^{3-o(1)})$. 
\end{corollary}

The best known upper bound on the formula complexity of $\majority_{n}$
is only $O(n^{3.91})$~\cite{S16}, and so we do not expect 
\expref{Corollary}{cor:majority} to be tight.

\appendix

\section{\label{sec:L-Am}Proof of \expref{Proposition}{prop:L-Am}}

In this appendix, we prove \expref{Proposition}{prop:L-Am}, restated next. 
\begin{restated}{\expref{Proposition}{prop:L-Am}}
For any $f\colon\B^{n}\to\B$ we have $L(f)\geq\Am^{2}(f)$. 
\end{restated}

\subsubsection*{Preliminaries}

While the proof of \expref{Proposition}{prop:L-Am} is fairly simple, it uses some basic
concepts from information theory which we review next. 
\begin{definition}
\label{def:information}Let $\x,\y,\z$ be discrete random variables. 
\begin{itemize}
\item The \emph{entropy} of $\x$ is $H(\x)=\E_{x\gets\x}\left[\log\frac{1}{\Pr\!\left[\x=x\right]}\right]$. 
\item The \emph{conditional entropy} of~$\x$ given~$\y$ is $H(\x\mid\y)=\E_{y\gets\y}\left[H\left(\x\mid\y=y\right)\right]$. 
\item The \emph{mutual information} between $\x$ and $\y$ is $I(\x;\y)=H(\x)-H(\x\mid\y).$ 
\item The \emph{conditional mutual information} between $\x$ and~$\y$
given~$\z$ is 
\[
I(\x;\y\mid\z)=H(\x\mid\z)-H(\x\mid\y,\z).
\]
\end{itemize}
\end{definition}

We use the following basic facts from information theory (see \cite{CT91}
for proofs). 
\begin{fact}
\label{fact:information-properties}Let $\x,\y,\z$ be discrete random
variables with finite supports, and let $X$ denote the support of~$\x$. 
\begin{itemize}
\item It holds that $0\le H(\x)\le\log\left|X\right|$. 
\item It holds that $0\le H(\x\mid\y)\le H(\x)$, where $H(\x\mid\y)=0$
if and only if $\y$ determines~$\x$ (\ie, $\x$ is a function
of~$\y$). 
\item If $\y$ determines $\x$ conditioned on~$\z$ (\ie, $\x$ is a
function of $\y$ and~$\z$) then $H(\y\mid\z)\ge H(\x\mid\z)$. 
\item It holds that $0\le I(\x;\y)\le H(\x)$. Similarly, we have
$0\le I(\x;\y\mid\z)\le H(\x\mid\z)$, where $I(\x;\y\mid\z)=0$ if
and only if $\x$ and~$\y$ are independent conditioned on any value
of~$\z$. 
\item It holds that $I(\x;\y)=I(\y;\x)$. Similarly, $I(\x;\y\mid\z)=I(\y;\x\mid\z)$. 
\item \emph{The chain rule:} It holds that 
\[
I(\x;\y,\z\mid\w)=I(\x;\z\mid\w)+I(\x;\y\mid\z,\w).
\]
\end{itemize}
\end{fact}

Finally, we use the following result from the theory of interactive
information complexity. 
\begin{claim}[{{\cite[Fact 4.15]{BBCR13}}}]
\label{claim:internal-vs-external-information}Let $\Pi$ be a deterministic
protocol. Suppose we invoke~$\Pi$ on random inputs $\x$ and $\y$
for Alice and Bob, respectively, and let $\bl$ denote the random
leaf that $\Pi$ reaches on those inputs. Then, 
\[
I(\x,\y;\bl)\ge I(\x;\bl\mid\y)+I(\y;\bl\mid\x).
\]
\end{claim}

\begin{proof}[Proof sketch.]
Let $\boldsymbol{t}$ denote the transcript of the protocol that
is associated with~$\bl$. We prove that $I(\x,\y;\boldsymbol{t})\ge I(\x;\boldsymbol{t}\mid\y)+I(\y;\boldsymbol{t}\mid\x)$,
as this is equivalent to the claim. Suppose Alice speaks first, and
denote the (random) bit she sends by $\boldsymbol{t}_{1}$. We show
that $I(\x,\y;\boldsymbol{t}_{1})\ge I(\x;\boldsymbol{t}_{1}\mid\y)+I(\y;\boldsymbol{t}_{1}\mid\x)$.
Using the chain rule, we can write 
\[
I(\x,\y;\boldsymbol{t}_{1})=I(\y;\boldsymbol{t}_{1})+I(\x;\boldsymbol{t}_{1}\mid\y)\geq I(\x;\boldsymbol{t}_{1}\mid\y)=I(\x;\boldsymbol{t}_{1}\mid\y)+I(\y;\boldsymbol{t}_{1}\mid\x),
\]
where the last equality follows since $I(\y;\boldsymbol{t}_{1}\mid\x)=0$,
as Alice's message $\boldsymbol{t}_{1}$ is independent of $\y$ given
her input $\x$. Proceeding by induction on the coordinates of $\boldsymbol{t}$
using the chain rule finishes the proof. 
\end{proof}

\subsubsection*{The proof of \expref{Proposition}{prop:L-Am}}

Our proof of \expref{Proposition}{prop:L-Am} generalizes similar arguments in~\cite{KW90,GMWW17}.
Let $\Pi$ be a protocol that solves~$\KW_{f}$. For every leaf~$\ell$
of~$\Pi$, we denote by $i_{\ell}$ the output of the protocol at
leaf~$\ell$, so we have $a_{i_{\ell}}\ne b_{i_{\ell}}$ any
pair of inputs~$(a,b)$ that reach the leaf~$\ell$. We prove that
$L(\Pi)\ge\Am^{2}(f)$. Let $(\a,\b)$ be a distribution for~$f$
that attains~$\Am(f)$, and let 
\[
q_{0}=\max_{\substack{b\in\supp(\b)\\
i\in[n]
}
}\Pr[\a_{i}\neq\b_{i}\mid\b=b],\qquad q_{1}=\max_{\substack{a\in\supp(\a)\\
i\in[n]
}
}\Pr[\a_{i}\neq\b_{i}\mid\a=a].
\]
Let $\bl$~be the leaf that $\Pi$~reaches on input~$(\a,\b)$.
By \expref{Fact}{fact:information-properties}, we have 
\[
I(\bl;\a,\b)\le H(\bl)\le\log\left|L(\Pi)\right|.
\]
On the other hand, \expref{Claim}{claim:internal-vs-external-information} implies
that 
\[
I(\bl;\a,\b)\ge I(\bl;\a\mid\b)+I(\bl;\b\mid\a).
\]
Next, observe that $\a$ and $\b$ together determine~$\bl$, and
that the event $\bl=\ell$ implies that $\a_{i_{\ell}}\ne\b_{i_{\ell}}$.
It follows that 
\begin{align*}
I(\bl;\a\mid\b) & =H(\bl\mid\b)-H(\bl\mid\a,\b) & \text{(by definition)}\\
 & =H(\bl\mid\b) & \text{(\ensuremath{\a} and \ensuremath{\b} determine~\ensuremath{\bl})}\\
 & =\E_{\ell\gets\bl,b\gets\b}\left[\left.\log\frac{1}{\Pr\!\left[\bl=\ell\mid\b=b\right]}\right|\b=b\right] & \text{(\expref{Definition}{def:information})}\\
 & \ge\E_{\ell\gets\bl,b\gets\b}\left[\left.\log\frac{1}{\Pr\!\left[\a_{i_{\ell}}\ne\b_{i_{\ell}}\mid\b=b\right]}\right|\b=b\right]\\
 & \ge\E_{\ell\gets\bl,b\gets\b}\left[\left.\log\frac{1}{q_{0}}\right|\b=b\right] & \text{(Definition of \ensuremath{q_{0}})}\\
 & =\log\frac{1}{q_{0}}
\end{align*}
Similarly, it can be shown that $I(\bl;\b\mid\a)\ge\log\frac{1}{q_{1}}$.
By combining the foregoing equations, it follows that 
\[
\log\left|L(\Pi)\right|\ge I(\bl;\a\mid\b)+I(\bl;\b\mid\a)\ge\log\frac{1}{q_{0}}+\log\frac{1}{q_{1}},
\]
and thus 
\[
\left|L(\Pi)\right|\ge\frac{1}{q_{0}}\cdot\frac{1}{q_{1}}=\Am^{2}(f),
\]
as required.
\begin{remark}
We note that \expref{Proposition}{prop:L-Am} can also be proved by observing that the
soft-adversary bound is a special case of the weighted adversary bound~\cite{A06},
which was shown to %
be a lower bound on %
formula complexity by \cite{LLS06}.
Alternatively, one could prove \expref{Proposition}{prop:L-Am} (up to a constant factor)
by combining \expref{Lemma}{lemma:ambainis-hiding} along with 
\expref{Lemma}{lemma:hiding-to-fixing}
and \expref{Lemma}{lemma:shrinkage-to-single-literal}.
\end{remark}

\section{\label{sec:min-khrapchenko}The min-entropy Khrapchenko bound}

In this appendix, we describe a generalization of the Khrapchenko
bound that provides a simple %
lower bound technique for %
the soft-adversary
bound. We then show how this generalization can be used to give an
alternative proof for the lower bound on the soft-adversary bound
of the surjectivity function. We start by recalling the original Khrapchenko
bound.
\begin{repdefinition}{\expref{Definition}{def:khrapchenko}}
Let $f\colon\B^{n}\to\B$ be a non-constant function. For every two
sets $A\subseteq f^{-1}(1)$ and $B\subseteq f^{-1}(0)$, let $E(A,B)$
be the set of pairs $(a,b)\in A\times B$ such that $a$ and~$b$
differ on exactly one coordinate. The \emph{Khrapchenko bound of~$f$}
is 
\[
\K(f)=\max_{A\subseteq f^{-1}(1),B\subseteq f^{-1}(0)}\frac{\left|E(A,B)\right|^{2}}{\left|A\right|\cdot\left|B\right|}.
\]
\end{repdefinition}

The Khrapchenko measure can be viewed as follows: Consider the subgraph
of the Hamming cube that consists only of the cut between $A$ and~$B$.
Then, the measure $\frac{\left|E(A,B)\right|^{2}}{\left|A\right|\cdot\left|B\right|}$
is the product of the average degree of a vertex in~$A$ (which is
$\frac{\left|E(A,B)\right|}{\left|A\right|}$) and the average degree
of a vertex in~$B$ (which is $\frac{\left|E(A,B)\right|}{\left|B\right|}$).
Note that the average degree of~$A$ can also be described as the
average, over all strings~$a\in A$, of the number of coordinates~$i\in\left[n\right]$
such that if we flip the $i$-th bit of~$a$ we get a string in~$B$.
We generalize the Khrapchenko measure as follows: 
\begin{itemize}
\item Whereas the Khrapchenko bound maximizes over all \emph{cuts} $(A,B)$
of the Hamming cube with $A\subseteq f^{-1}(1)$ and $B\subseteq f^{-1}(0)$,
we are maximizing over all \emph{distributions} over edges $(\a,\b)$
of the Hamming cube, where $\a\in f^{-1}(1)$ and $\b\in f^{-1}(0)$. 
\item Whereas the Khrapchenko bound considers the average number of coordinates~$i$
as described above, we consider the min-entropy of the coordinate~$\i$
on which $\a,\b$ differ.
\item Whereas the Khrapchenko bound considers functions whose inputs are
binary strings, we consider inputs that are strings over arbitrary
finite alphabets. 
\end{itemize}
Our generalization uses the following notion: The \emph{conditional
min-entropy} of a random variable~$\x$ conditioned on a random variable~$\y$
is $\Hm(\x\mid\y)=\min_{x,y}\log\frac{1}{\Pr[\x=x\mid\y=y]}.$ We
can now define our generalization formally: 
\begin{definition}
\label{def:min-entropy-khrapchenko}Let $\Sigma$ be a finite alphabet,
and let $f\colon\Sigma^{n}\to\{0,1\}$ be a non-constant Boolean function.
We say that a distribution $(\a,\b)$ on $f^{-1}(1)\times f^{-1}(0)$
is a \emph{Khrapchenko distribution for~$f$} if $\a$ and $\b$
always differ on a unique coordinate $\i\in\left[n\right]$. We define
the \emph{min-entropy Khrapchenko bound of~$f$,} which we denote~$\Km(f)$,
to be the maximum of the quantity 
\[
2^{\Hm(\i\mid\a)+\Hm(\i\mid\b)}
\]
over all Khrapchenko distributions $(\a,\b)$ for~$f$. 
\end{definition}

As noted above, the min-entropy Khrapchenko bound can be used to lower
bound the the soft-adversary bound. In order to define the adversary
bound of a function $f\colon\Sigma^{n}\to\B$ over a non-boolean alphabet~$\Sigma$,
we view~$f$ as if its input is a binary string in~$\B^{n\cdot\left\lceil \log\left|\Sigma\right|\right\rceil }$
that encodes a string in~$\Sigma^{n}$ via some fixed encoding (the
choice of the encoding does not matter for what follows).
\begin{lemma}
\label{lemma:Khr-Amb}For every non-constant function~$f\colon \Sigma^{n}\to\B$
we have $\Am(f)\geq\sqrt{\Km(f)}$. 
\end{lemma}

\begin{proof}
Let $(\a,\b)$ be a Khrapchenko distribution that attains $\Km(f)$.
Let $\i\in[n]$ be the unique index at which $\a,\b$ differ, and
let $\a_{i,j},\b_{i,j}$ be the $j$-th bits of the binary encoding
of $\a_{i},\b_{i}\in\Sigma$. For any $a\in\supp(\a)$ and $(i,j)\in[n]\times[\lceil\log|\Sigma|\rceil]$,
we have 
\[
\Pr[\a_{i,j}\neq\b_{i,j}\mid\a=a]\leq\Pr[\a_{i}\neq\b_{i}\mid\a=a]=\Pr[\i=i\mid\a=a]\leq2^{-\Hm(\i\mid\a)}.
\]
It follows that the first factor in the definition of $\Am$ is at
least $2^{\Hm(\i\mid\a)}$. Similarly, the second factor is at least
$2^{\Hm(\i\mid\b)}$, and so the entire expression is at least $\sqrt{\Km(f)}$. 
\end{proof}
Similarly to the original Khrapchenko bound, the min-entropy Khrapcehnko
bound too gives a lower bound on the formula complexity~$L(f)$,
which can be proved by combining \expref{Proposition}{prop:L-Am} 
and \expref{Lemma}{lemma:Khr-Amb}. Again,
in order to define the formula complexity of a function $f\colon\Sigma^{n}\to\B$
over a non-boolean alphabet, we view~$f$ as if its input is a binary
string in~$\B^{n\cdot\left\lceil \log\left|\Sigma\right|\right\rceil }$. 
\begin{corollary}
\label{cor:min-entropy-khrapchenko-lower-bound}Let $\Sigma$ be a finite
alphabet, and let $f\colon\Sigma^{n}\to\{0,1\}$ be a non-constant
Boolean function. Then $L(f)\geq\Km(f)$. 
\end{corollary}

\noindent By combining \expref{Lemmas}{lemma:ambainis-hiding}
\expref{and}{lemma:Khr-Amb}, it also follows
that the min-entropy Khrapchenko bound too can be used for constructing
hiding projections.
\begin{corollary}
\label{cor:min-khrapchenko-hiding}Let $\Sigma$ be a finite alphabet,
let $f\colon\Sigma^{n}\to\B$ be a non-constant Boolean function.
There is a $q$-hiding projection $\p$ to a single variable~$y$
such that $q=\sqrt{1/\Km(f)}$ and such that $f|_{\p}$~is a non-constant
function with probability~$1$.
\end{corollary}

\subsection{Lower bound on the surjectivity function}  %

In order to demonstrate the usefulness of the min-entropy Khrapchenko
bound, we use it to prove a lower bound on the surjectivity function
from \expref{Section}{subsec:AC0-bounds}. This implies a similar bound on the
soft-adversary bound of the surjectivity function, and provides an
alternative way for proving our main theorem.
\begin{repdefinition}{\expref{Definition}{def:surjectivity}}
Let $\Sigma$ be a finite alphabet, and let $r\in\N$ be such that
$r\ge\left|\Sigma\right|$. The \emph{surjectivity function} $\surj_{\Sigma,r}\colon\Sigma^{r}\to\{0,1\}$
interprets its input as a function from $\left[r\right]$ to $\Sigma$,
and outputs whether the function is surjective. In other words, $\surj(\sigma_{1},\ldots,\sigma_{r})=1$
if and only if every symbol in~$\Sigma$ appears in~$(\sigma_{1},\ldots,\sigma_{r})$.
\end{repdefinition}

\begin{lemma}
\label{lemma:surj-khrap}For every $s\in\N$, we have $\Km(\surj_{\left[2s+1\right],3s+1})=\Omega\bigl(s^{2}\bigr)$. 
\end{lemma}

\begin{proof}
Let $s\in\N$, let $\Sigma=\left[2s+1\right]$, and let $\surj=\surj_{\Sigma,3s+1}$.
We define a Khrapchenko distribution~$(\a,\b)$ for~$\surj$ as
follows. The input~$\b\in\surj^{-1}(0)$ is a uniformly distributed
string in~$\Sigma^{3s+1}$ in which $s+1$~of the symbols in~$\Sigma$
appear exactly twice, $s-1$~of the symbols in~$\Sigma$ appear
exactly once, and the remaining symbol in~$\Sigma$ does not appear
at all. The string~$\a$ is sampled by choosing uniformly at random
a coordinate~$\i$ such that $\b_{\i}$ is one of the symbols that
appear twice in~$\b$, and replacing $\b_{\i}$ with the unique symbol
in~$\Sigma$ that does not appear in~$\b$.

We turn to bound $\Hm(\i\mid\b)$ and $\Hm(\i\mid\a)$. First, observe
that conditioned on any choice of~$\b$, the coordinate~$\i$ is
uniformly distributed among $2s+2$ coordinates, and therefore 
\[
\Hm(\i\mid\b)=\log(2s+2).
\]
In order to bound $\Hm(\i\mid\a)$, observe that $\a$ is distributed
like a uniformly distributed string in~$\Sigma^{3s+1}$ in which
$s+1$~of the symbols in~$\Sigma$ appear exactly once, and the
remaining $s$~symbols in~$\Sigma$ appear exactly twice. Conditioned
on any choice of~$\a$, the coordinate~$\i$ is uniformly distributed
over the $s+1$ coordinates of symbols that appear exactly once. It
follows that 
\[
\Hm(\i\mid\a)\ge\log(s+1).
\]
We conclude that
\[
\Km(\surj_{n})\ge2^{\Hm(\i\mid\a)+\Hm(\i\mid\b)}\ge(s+1)\cdot(2s+2)=\Omega\!\left(s^{2}\right).\qedhere
\]
\end{proof}
\expref{Lemma}{lemma:surj-khrap} implies, via 
\expref{Corollary}{cor:min-entropy-khrapchenko-lower-bound}
a quadratic lower bound on the formula complexity of~$\surj$. The
same result also follows from the lower bound on the adversary bound
of~$\surj$ due to~\cite{BM12} (\expref{Lemma}{lemma:surj-amb}).

\subsubsection{Comparison to the Khrapchenko bound}  %

We now further compare the min-entropy 
Khrap\-chen\-ko %
bound to the original
Khrapchenko bound. Observe that when $\Sigma=\B$, the min-entropy
$\Hm(\i\mid\a)$ is exactly the min-entropy of a random neighbor of~$\a$.
In particular, when $(\a,\b)$ is the uniform distribution over a
set of edges, the min-entropy $\Hm(\i\mid\a)$ is the logarithm of
the minimal degree of a vertex~$a$. Moreover, if the latter set
of edges also induces a regular graph, then the measure $2^{\Hm(\i\mid\a)+\Hm(\i\mid\b)}$
coincides exactly with the original measure~$\frac{\left|E(A,B)\right|^{2}}{\left|A\right|\cdot\left|B\right|}$
of Khrapchenko. More generally, when $\Sigma=\B$, the bound $\Km(f)$
is within a constant factor of the original Khrapchenko bound, as
we show in \expref{Section}{sec:K-Km}. 
\begin{proposition}
\label{prop:K-Km}For any non-constant function $f\colon\B^{n}\to\B$ it
holds that 
\[
\tfrac{\K(f)}{4}\le\Km(f)\le\K(f).
\]
 
\end{proposition}

Unfortunately, when $\Sigma$ is a larger alphabet, the connection
between $\Km(f)$ and the measure $\frac{\left|E(A,B)\right|^{2}}{\left|A\right|\left|B\right|}$
is not so clean. Specifically, the min-entropy $\Hm(\i\mid\a)$ has
no clear connection to the degree of~$\a$, since the vertex~$\a$
may have multiple neighbors that correspond to the same coordinate~$\i$.

\subsubsection{Previous work on the Khrapchenko bound}  %

Several versions of the Khrapchenko bound appeared in the literature:
Zwick~\cite{Z91} generalized the Khrapchenko bound such that different
input coordinates can be given different weights, and Koutsoupias
\cite{K93} gave a spectral generalization of the bound. The paper
of H{\aa}stad~\cite{H98} observed that his analogue of 
\expref{Lemma}{lemma:shrinkage-to-single-literal}
can be viewed as a generalization of the Khrapchenko bound. Ganor,
Komargodski, and Raz~\cite{GKR12} considered a variant of the Khrapchenko
bound in which the edges of the Boolean hypercube are replaced with
random walks on the noisy hypercube. Of particular relevance is a
paper of Laplante, Lee, and Szegedy~\cite{LLS06} that defined a
complexity measure that is very similar to our min-entropy Khrapchenko
bound, except that the entropy is replaced with Kolmogorov complexity.

\subsubsection{An entropy Khrapchenko bound}  %

It is possible to generalize the complexity measure $\Km$ by replacing
the min-entropy in \expref{Definition}{def:min-entropy-khrapchenko} 
with Shannon entropy.
Such a measure would still %
give a lower bound on %
formula complexity --- specifically,
the proof of \expref{Corollary}{cor:min-entropy-khrapchenko-lower-bound} 
would go through
without a change. However, we do not know how to use such a measure
for constructing hiding projections as in 
\expref{Corollary}{cor:min-khrapchenko-hiding}.
We note that it is easy to prove that such a measure is an upper bound
on~$\Km$.

\section{\label{sec:K-Km}Proof of \expref{Propositions}{prop:K-Km}
  \expref{and}{prop:K-Amb}}

In this appendix we prove \expref{Propositions}{prop:K-Km}
  \expref{and}{prop:K-Amb}, restated next. 
\begin{restated}{\expref{Proposition}{prop:K-Km}}
For any $f\colon\B^{n}\to\B$ we have $\tfrac{\K(f)}{4}\le\Km(f)\le\K(f)$. 
\end{restated}

\begin{restated}{\expref{Proposition}{prop:K-Amb}}
For any $f\colon\B^{n}\to\B$ we have $\sqrt{\tfrac{\K(f)}{4}}\le\Amb(f)$. 
\end{restated}

We relate $\K(f)$ to $\Km(f)$ using an auxiliary measure $\Kmin(f)$:
\[
\Kmin(f)=\max_{A\subseteq f^{-1}(1),B\subseteq f^{-1}(0)}\left(\min_{a\in A}{|E(a,B)|}\cdot\min_{b\in B}{|E(A,b)|}\right),
\]
where $E(a,B)=E(\{a\},B)$ and similarly $E(A,b)=E(A,\{b\})$.

We first show that $\Kmin(f)$ and $\K(f)$ are equal up to constants.
\begin{claim}
\label{claim:K-Kmin} For any $f\colon\Sigma^{n}\to\B$ we have $\tfrac{\K(f)}{4}\le\Kmin(f)\le\K(f)$. 
\end{claim}

\begin{proof}
The inequality $\Kmin(f)\le\K(f)$ is simple since $\min_{a\in A}{|E(a,B)|}\le|E(A,B)|/|A|$
and similarly $\min_{b\in B}{|E(A,b)|}\le|E(A,B)|/|B|$.

The other direction is more subtle. For ease of notation, for any
sets $A\subseteq f^{-1}(1)$ and $B\subseteq f^{-1}(0)$ we 
write %
$\K(A,B)=\frac{|E(A,B)|^{2}}{|A||B|}$. Thus, 
\begin{equation}
\K(f)=\max_{A\subseteq f^{-1}(1),B\subseteq f^{-1}(0)}\K(A,B).\label{eq:K(AB)}
\end{equation}
Assume that $A$ and $B$ are sets that maximize $\K(A,B)$ in \expref{Equation}{eq:K(AB)}.
We show that 
\begin{align}
\forall{a\in A}\colon & \quad|E(a,B)|\ge\tfrac{|E(A,B)|}{2|A|},\label{eq:min deg A}\\
\forall{b\in B}\colon & \quad|E(A,b)|\ge\tfrac{|E(A,B)|}{2|B|}.\label{eq:min deg B}
\end{align}
In words, the min-degree is at least half the average degree. Before
showing why \expref{Equations}{eq:min deg A}
\expref{and}{eq:min deg B} hold, we show that they
imply the statement of the claim. Indeed, 
\[
\Kmin(f)\ge\left(\min_{a\in A}{|E(a,B)|}\cdot\min_{b\in B}{|E(A,b)|}\right)\ge\tfrac{|E(A,B)|}{2|A|}\cdot\tfrac{|E(A,B)|}{2|B|}=\tfrac{\K(A,B)}{4}=\tfrac{\K(f)}{4}.
\]

It remains to show that \expref{Equations}{eq:min deg A}
\expref{and}{eq:min deg B} hold. We
focus on \expref{Equation}{eq:min deg A} due to symmetry. Assume by contradiction
that there exists an $a\in A$ for which 
\[
|E(a,B)|<\tfrac{|E(A,B)|}{2|A|}.
\]
It must be the case that $|A|>1$, as otherwise $A$ contains only
one element and $|E(a,B)|=|E(A,B)|/|A|$ for this element. Consider
now the set $A'=A\setminus\{a\}$ --- which is non-empty by the above
discussion. We claim that $\K(A',B)>\K(A,B)$, contradicting the choice
of $A,B$. Indeed, 
\begin{align*}
\K(A',B) & =\frac{|E(A',B)|^{2}}{|A'||B|}\\
 & =\frac{(|E(A,B)|-|E(a,B)|)^{2}}{(|A|-1)|B|}\\
 & >\frac{(|E(A,B)|-\frac{|E(A,B)|}{2|A|})^{2}}{(|A|-1)|B|} & \text{(By assumption on \ensuremath{a})}\\
 & =\frac{|E(A,B)|^{2}\cdot(1-\frac{1}{2|A|})^{2}}{|A||B|\cdot(1-\frac{1}{|A|})}\\
 & =\K(A,B)\cdot\frac{(1-\frac{1}{2|A|})^{2}}{(1-\frac{1}{|A|})}\\
 & >\K(A,B). & \qedhere
\end{align*}
\end{proof}
It turns out that over the binary alphabet, $\Kmin(f)$ and $\Km(f)$
coincide. 
\begin{claim}
\label{claim:Km-Kmin} For any $f\colon\B^{n}\to\B$ we have $\Km(f)=\Kmin(f)$. 
\end{claim}

\begin{proof}
We first show that $\Km(f)\ge\Kmin(f)$. Let $A,B$ be sets that maximize
the expression $\left(\min_{a\in A}{|E(a,B)|}\cdot\min_{b\in B}{|E(A,b)|}\right)$.
We take $(\a,\b)$ to be a uniformly distributed pair in $E(A,B)$.
It is not hard to see that 
\[
2^{\Hm(\i\mid\a)}=2^{\Hm(\b\mid\a)}=\min_{a\in A}{|E(a,B)|},
\]
and similarly $2^{\Hm(\i\mid\b)}=\min_{b\in B}{|E(A,b)|}$. We thus
get $\Km(f)\ge\Kmin(f)$.

Next, we show the other direction, $\Kmin(f)\ge\Km(f)$. Let $(\a,\b)$
be a random variable distributed according to a Khrapchenko distribution
for $f$ that attains the maximum of $2^{\Hm(\a\mid\b)+\Hm(\b\mid\a)}$
over all such distributions. Let $A:=\supp(\a)$ and $B:=\supp(\b)$
be the supports of $\a$ and $\b$, respectively. By definition of
$\Hm$ we have $2^{\Hm(\a\mid\b)}=\frac{1}{\max_{a,b}\Pr[\a=a\mid\b=b]}$.
Rearranging, we get that for any $b$ in the support of $\b$ it holds
that 
\[
\Pr[\a=a\mid\b=b]\le1/2^{\Hm(\a|\b)}.
\]
In particular, since theses probabilities sum to $1$, it must be
the case that there are at least $2^{\Hm(\a\mid\b)}$ neighbors of
$b$ in $A$ --- \ie, $|E(A,b)|\ge2^{\Hm(\a\mid\b)}$ for all $b\in B$.
Similarly, $|E(a,B)|\ge2^{\Hm(\b\mid\a)}$ for all $a\in A$. The
two sets $A$ and $B$ show that $\Kmin(f)\ge2^{\Hm(\a\mid\b)}\cdot2^{\Hm(\b\mid\a)}=\Km(f)$. 
\end{proof}
\expref{Proposition}{prop:K-Km} follows by combining these two claims. 
\expref{Proposition}{prop:K-Amb} follows
from \expref{Claim}{claim:Km-Kmin} using the following simple observation.
\begin{claim}
\label{claim:Amb-Kmin} For any $f\colon\B^{n}\to\B$ we have $\Amb(f)\geq\sqrt{\Kmin(f)}$. 
\end{claim}

\begin{proof}
Let $A,B$ be sets that attain $\Kmin(f)$. Define $R$ to be the
set of pairs $(a,b)\in A\times B$ that differ at a single coordinate.
Thus $|R(a,B)|=|E(a,B)|$ and $|R(A,b)|=|E(A,b)|$. Moreover, $|R_{i}(a,B)|\leq1$
since there is a unique $b$ that differs from $a$ only on the $i$-th
coordinate. Similarly, $|R_{i}(A,b)|\leq1$. The claim now immediately
follows by comparing the definitions of $\Kmin(f)$ and $\Amb(f)$. 
\end{proof}

\subsection*{Acknowledgements}

A.\,T.\ would like to thank Igor Carboni Oliveira for bringing the question
of proving formula size lower bounds for $\ACz$ to his attention.
We are also grateful to Robin Kothari for posing this open question
on ``Theoretical Computer Science Stack Exchange'' \cite{K11},
and to Kaveh Ghasemloo and Stasys Jukna for their feedback on this
question. We would like to thank Anna G{\'{a}}l for very helpful discussions,
and Gregory Rosenthal for helpful comments on the manuscript. Finally,
we are grateful to anonymous referees for comments that improved the
presentation of this work.

This work was partly carried out while the authors were visiting the
Simons Institute for the Theory of Computing in association with the
DIMACS/Simons Collaboration on Lower Bounds in Computational Complexity,
which is conducted with support from the National Science Foundation.

\nocite{conf-version}

\bibliographystyle{tocplain}
\bibliography{bibstrings,v019a007,bibtail}

\newcommand{\prelim}{Preliminary version in\ }\newcommand{\prelims}{Preliminary
  versions in\ }
\providecommand{\bibhead}[1]{}
\expandafter\ifx\csname pdfbookmark\endcsname\relax%
  \providecommand{\tocrefpdfbookmark}{}
\else\providecommand{\tocrefpdfbookmark}{%
   \hypertarget{tocreferences}{}%
   \pdfbookmark[1]{References}{tocreferences}}%
\fi

\tocrefpdfbookmark
\begin{thebibliography}{10}

\bibitem{A83}\bibhead{A83}
{\sc Mikl\'{o}s Ajtai}: {$\Sigma_1^1$}-formulae on finite structures.
\newblock {\em Ann. Pure Appl. Logic}, 24(1):1--48, 1983.
\newblock [\epfmtdoi{10.1016/0168-0072(83)90038-6}]

\bibitem{A02}\bibhead{A02}
{\sc Andris Ambainis}: Quantum lower bounds by quantum arguments.
\newblock {\em J. Comput. System Sci.}, 64(4):750--767, 2002.
\newblock Special issue on STOC 2000 (Portland, OR).
\newblock [\epfmtdoi{10.1006/jcss.2002.1826}]

\bibitem{A06}\bibhead{A06}
{\sc Andris Ambainis}: Polynomial degree vs. quantum query complexity.
\newblock {\em J. Comput. System Sci.}, 72(2):220--238, 2006.
\newblock [\epfmtdoi{10.1016/j.jcss.2005.06.006}]

\bibitem{A87}\bibhead{A87}
{\sc Alexander~E. Andreev}: On a method for obtaining more than quadratic
  effective lower bounds for the complexity of $\pi$-schemes.
\newblock {\em Moscow University Mathematics Bulletin}, 42(1):24--29, 1987.
\newblock \href{https://mi.mathnet.ru/vmumm3034}{Vestnik Moskov. Univ. Ser.~1:
  Mat. Mekh. 1987(1) 70--73 (Russian original on Math-net.ru)}
  \href{https://mathscinet.ams.org/mathscinet/article?mr=0883632}{MR 0883632}.

\bibitem{BBCR13}\bibhead{BBCR13}
{\sc Boaz Barak, Mark Braverman, Xi~Chen, and Anup Rao}: How to compress
  interactive communication.
\newblock {\em SIAM J. Comput.}, 42(3):1327--1363, 2013.
\newblock [\epfmtdoi{10.1137/100811969}]

\bibitem{BM12}\bibhead{BM12}
{\sc Paul Beame and Widad Machmouchi}: The quantum query complexity of
  {AC}\({}^{\mbox{0}}\).
\newblock {\em Quantum Inf. Comput.}, 12(7--8):670--676, 2012.
\newblock [\epfmtdoi{10.26421/QIC12.7-8-11}]

\bibitem{BB94}\bibhead{BB94}
{\sc Maria~Luisa Bonet and Samuel~R. Buss}: Size-depth tradeoffs for {B}oolean
  formulae.
\newblock {\em Inform. Process. Lett.}, 49(3):151--155, 1994.
\newblock [\epfmtdoi{10.1016/0020-0190(94)90093-0}]

\bibitem{B74}\bibhead{B74}
{\sc Richard~P. Brent}: The parallel evaluation of general arithmetic
  expressions.
\newblock {\em J. ACM}, 21(2):201--206, 1974.
\newblock [\epfmtdoi{10.1145/321812.321815}]

\bibitem{BW02}\bibhead{BW02}
{\sc Harry Buhrman and Ronald de~Wolf}: Complexity measures and decision tree
  complexity: a survey.
\newblock {\em Theoret. Comput. Sci.}, 288(1):21--43, 2002.
\newblock [\epfmtdoi{10.1016/S0304-3975(01)00144-X}]

\bibitem{CKK12}\bibhead{CKK12}
{\sc Andrew~M. Childs, Shelby Kimmel, and Robin Kothari}: The quantum query
  complexity of read-many formulas.
\newblock In {\em Proc. 20th Eur. Symp. Algorithms (ESA'12)}, pp. 337--348.
  Springer, 2012.
\newblock [\epfmtdoi{10.1007/978-3-642-33090-2\_30}]

\bibitem{CT91}\bibhead{CT91}
{\sc Thomas~M. Cover and Joy~A. Thomas}: {\em Elements of Information Theory}.
\newblock Wiley-Interscience, 1991.
\newblock [\epfmtdoi{10.1002/047174882X}]

\bibitem{RMNPR20}\bibhead{RMNPR20}
{\sc Susanna~F. de~Rezende, Or~Meir, Jakob Nordstr{\"{o}}m, Toniann Pitassi,
  and Robert Robere}: {KRW} composition theorems via lifting.
\newblock In {\em Proc. 61st FOCS}, pp. 43--49. IEEE Comp. Soc., 2020.
\newblock [\epfmtdoi{10.1109/FOCS46700.2020.00013}, \epfmt{eccc}{TR20-099}]

\bibitem{DM18}\bibhead{DM18}
{\sc Irit Dinur and Or~Meir}: Toward the {KRW} composition conjecture: Cubic
  formula lower bounds via communication complexity.
\newblock {\em Comput. Complexity}, 27(3):375--462, 2018.
\newblock [\epfmtdoi{10.1007/s00037-017-0159-x}]

\bibitem{EIRS01}\bibhead{EIRS01}
{\sc Jeff Edmonds, Russell Impagliazzo, Steven Rudich, and Jir{\'{\i}} Sgall}:
  Communication complexity towards lower bounds on circuit depth.
\newblock {\em Comput. Complexity}, 10(3):210--246, 2001.
\newblock [\epfmtdoi{10.1007/s00037-001-8195-x}]

\bibitem{conf-version}\bibhead{conf-version}
{\sc Yuval Filmus, Or~Meir, and Avishay Tal}: Shrinkage under random
  projections, and cubic formula lower bounds for {AC0} (extended abstract).
\newblock In {\em Proc. 12th Innovations in Theoret. Comp. Sci. Conf.
  (ITCS'21)}, pp. 89:1--7. Schloss Dagstuhl--Leibniz-Zentrum fuer Informatik,
  2021.
\newblock [\epfmtdoi{10.4230/LIPIcs.ITCS.2021.89}]

\bibitem{FSS84}\bibhead{FSS84}
{\sc Merrick~L. Furst, James~B. Saxe, and Michael Sipser}: Parity, circuits,
  and the polynomial-time hierarchy.
\newblock {\em Math. Sys. Theory}, 17(1):13--27, 1984.
\newblock [\epfmtdoi{10.1007/BF01744431}]

\bibitem{GTN19}\bibhead{GTN19}
{\sc Anna G{\'{a}}l, Avishay Tal, and Adrian~Trejo Nu{\~{n}}ez}: Cubic formula
  size lower bounds based on compositions with majority.
\newblock In {\em Proc. 10th Innovations in Theoret. Comp. Sci. Conf.
  (ITCS'19)}, pp. 35:1--35:13. Schloss Dagstuhl--Leibniz-Zentrum fuer
  Informatik, 2019.
\newblock [\epfmtdoi{10.4230/LIPIcs.ITCS.2019.35}]

\bibitem{GKR12}\bibhead{GKR12}
{\sc Anat Ganor, Ilan Komargodski, and Ran Raz}: The spectrum of small {D}e
  {M}organ formulas.
\newblock {\em Electron. Colloq. Comput. Complexity}, TR12-174, 2012.
\newblock [\epfmt{ecccprimary}{TR12-174}]

\bibitem{GMWW17}\bibhead{GMWW17}
{\sc Dmitry Gavinsky, Or~Meir, Omri Weinstein, and Avi Wigderson}: Toward
  better formula lower bounds: The composition of a function and a universal
  relation.
\newblock {\em SIAM J. Comput.}, 46(1):114--131, 2017.
\newblock [\epfmtdoi{10.1137/15M1018319}]

\bibitem{GH92}\bibhead{GH92}
{\sc Mikael Goldmann and Johan H{\aa}stad}: A simple lower bound for monotone
  clique using a communication game.
\newblock {\em Inform. Process. Lett.}, 41(4):221--226, 1992.
\newblock [\epfmtdoi{10.1016/0020-0190(92)90184-W}]

\bibitem{GP18}\bibhead{GP18}
{\sc Mika G{\"{o}}{\"{o}}s and Toniann Pitassi}: Communication lower bounds via
  critical block sensitivity.
\newblock {\em SIAM J. Comput.}, 47(5):1778--1806, 2018.
\newblock [\epfmtdoi{10.1137/16M1082007}]

\bibitem{H86}\bibhead{H86}
{\sc Johan H{\aa}stad}: Almost optimal lower bounds for small depth circuits.
\newblock In {\sc S.~Micali}, editor, {\em Randomness and Computation},
  volume~5 of {\em Adv. Computing Research}, pp. 143--170. JAI Press, 1989.
\newblock \href{https://hdl.handle.net/1721.1/151177}{DSpace@MIT}. \prelim
  \href{https://doi.org/10.1145/12130.12132}{STOC'86}.

\bibitem{H98}\bibhead{H98}
{\sc Johan H{\aa}stad}: The shrinkage exponent of {D}e {M}organ formulas is 2.
\newblock {\em SIAM J. Comput.}, 27(1):48--64, 1998.
\newblock \prelim \href{https://doi.org/10.1109/SFCS.1993.366876}{FOCS'93}.
\newblock [\epfmtdoi{10.1137/S0097539794261556}]

\bibitem{HJP95}\bibhead{HJP95}
{\sc Johan H{\aa}stad, Stasys Jukna, and Pavel Pudl{\'a}k}: Top-down lower
  bounds for depth-three circuits.
\newblock {\em Comput. Complexity}, 5(2):99--112, 1995.
\newblock [\epfmtdoi{10.1007/BF01268140}]

\bibitem{HRST17}\bibhead{HRST17}
{\sc Johan H{\aa}stad, Benjamin Rossman, Rocco~A. Servedio, and Li{-}Yang Tan}:
  An average-case depth hierarchy theorem for {B}oolean circuits.
\newblock {\em J. ACM}, 64(5):35:1--27, 2017.
\newblock \prelim \href{https://doi.org/10.1109/FOCS.2015.67}{FOCS'15}.
\newblock [\epfmtdoi{10.1145/3095799}]

\bibitem{HW93}\bibhead{HW93}
{\sc Johan H{\aa}stad and Avi Wigderson}: Composition of the universal
  relation.
\newblock In {\sc Jin yi~Cai}, editor, {\em Advances In Computational
  Complexity Theory}, volume~13 of {\em DIMACS Ser. in Discr. Math.}, pp.
  119--134. Amer. Math. Soc., 1993.
\newblock [\epfmtdoi{10.1090/dimacs/013/07}]

\bibitem{H87}\bibhead{H87}
{\sc Johan~Torkel H{\aa}stad}: {\em Computational limitations of small-depth
  circuits}.
\newblock Ph.\,D.\ thesis, MIT, 1986.
\newblock \href{https://www.csc.kth.se/~johanh/thesis_with_figs.pdf}{Available
  on author's website}, published as a book by
  \href{https://mitpress.mit.edu/9780262081672/computational-limitations-for-small-depth-circuits/}{MIT
  Press, 1987}.

\bibitem{HLS07}\bibhead{HLS07}
{\sc Peter H{\o}yer, Troy Lee, and Robert {\v{S}}palek}: Negative weights make
  adversaries stronger.
\newblock In {\em Proc. 39th STOC}, pp. 526--535. ACM Press, 2007.
\newblock [\epfmtdoi{10.1145/1250790.1250867}]

\bibitem{IN93}\bibhead{IN93}
{\sc Russell Impagliazzo and Noam Nisan}: The effect of random restrictions on
  formula size.
\newblock {\em Random Struct. Algor.}, 4(2):121--134, 1993.
\newblock [\epfmtdoi{10.1002/rsa.3240040202}]

\bibitem{JuknaBook}\bibhead{JuknaBook}
{\sc Stasys Jukna}: {\em {B}oolean Function Complexity -- Advances and
  Frontiers}.
\newblock Volume~27 of {\em Algorithms and Combinatorics}.
\newblock Springer, 2012.
\newblock [\epfmtdoi{10.1007/978-3-642-24508-4}]

\bibitem{KRW95}\bibhead{KRW95}
{\sc Mauricio Karchmer, Ran Raz, and Avi Wigderson}: Super-logarithmic depth
  lower bounds via the direct sum in communication complexity.
\newblock {\em Comput. Complexity}, 5(3--4):191--204, 1995.
\newblock [\epfmtdoi{10.1007/BF01206317}]

\bibitem{KW90}\bibhead{KW90}
{\sc Mauricio Karchmer and Avi Wigderson}: Monotone circuits for connectivity
  require super-logarithmic depth.
\newblock {\em SIAM J. Discr. Math.}, 3(2):255--265, 1990.
\newblock [\epfmtdoi{10.1137/0403021}]

\bibitem{K72}\bibhead{K72}
{\sc Valerii~Mikhailovich Khrapchenko}: A method of obtaining lower bounds for
  the complexity of {$\pi$}-schemes.
\newblock {\em Math. Notes. Acad. Sci. USSR (English translation)},
  10:474--479, 1971.
\newblock \href{https://mi.mathnet.ru/mzm7071}{Matem. Zametki 10(1) 83-92, 1971
  (Russian original on Math-net.ru)}.
\newblock [\epfmtdoi{10.1007/BF01747074}]

\bibitem{KM18}\bibhead{KM18}
{\sc Sajin Koroth and Or~Meir}: Improved composition theorems for functions and
  relations.
\newblock In {\em Proc. 22nd Internat. Conf. on Randomization and Computation
  (RANDOM'18)}, pp. 48:1--18. Schloss Dagstuhl--Leibniz-Zentrum fuer
  Informatik, 2018.
\newblock [\epfmtdoi{10.4230/LIPIcs.APPROX-RANDOM.2018.48}]

\bibitem{K11}\bibhead{K11}
{\sc Robin Kothari}: Formula size lower bounds for {AC0} functions, 2011.
\newblock
  \href{https://cstheory.stackexchange.com/questions/7156/formula-size-lower-bounds-for-ac0-functions}{Theoretical
  Computer Science Stack Exchange}.

\bibitem{K93}\bibhead{K93}
{\sc Elias Koutsoupias}: Improvements on {Khrapchenko}'s theorem.
\newblock {\em Theoret. Comput. Sci.}, 116(2):399--403, 1993.
\newblock [\epfmtdoi{10.1016/0304-3975(93)90330-V}]

\bibitem{KN_book}\bibhead{KN_book}
{\sc Eyal Kushilevitz and Noam Nisan}: {\em Communication Complexity}.
\newblock Cambridge Univ. Press, 1997.
\newblock [\epfmtdoi{10.1017/9781108671644}]

\bibitem{LLS06}\bibhead{LLS06}
{\sc Sophie Laplante, Troy Lee, and Mario Szegedy}: The quantum adversary
  method and classical formula size lower bounds.
\newblock {\em Comput. Complexity}, 15(2):163--196, 2006.
\newblock [\epfmtdoi{10.1007/s00037-006-0212-7}]

\bibitem{LMRSS11}\bibhead{LMRSS11}
{\sc Troy Lee, Rajat Mittal, Ben~W. Reichardt, Robert {\v{S}}palek, and Mario
  Szegedy}: Quantum query complexity of state conversion.
\newblock In {\em Proc. 52nd FOCS}, pp. 344--353. IEEE Comp. Soc., 2011.
\newblock [\epfmtdoi{10.1109/FOCS.2011.75}]

\bibitem{LS20}\bibhead{LS20}
{\sc Lily Li and Morgan Shirley}: The general adversary bound: A survey, 2021.
\newblock [\epfmt{arxiv}{2104.06380}]

\bibitem{M20}\bibhead{M20}
{\sc Or~Meir}: Toward better depth lower bounds: Two results on the multiplexor
  relation.
\newblock {\em Comput. Complexity}, 29(1):4, 2020.
\newblock [\epfmtdoi{10.1007/s00037-020-00194-8}]

\bibitem{MS21}\bibhead{MS21}
{\sc Ivan Mihajlin and Alexander Smal}: Toward better depth lower bounds: The
  {XOR-KRW} conjecture.
\newblock In {\em Proc. 36th Comput. Complexity Conf. (CCC'21)}, pp. 38:1--24.
  Schloss Dagstuhl--Leibniz-Zentrum fuer Informatik, 2021.
\newblock [\epfmtdoi{10.4230/LIPIcs.CCC.2021.38}]

\bibitem{N66}\bibhead{N66}
{\sc {\`E}duard~Ivanovi{\v c} Ne{\v c}iporuk}: On a {B}oolean function.
\newblock {\em Soviet Math. Doklady}, 7(4):999--1000, 1966.
\newblock \href{http://mi.mathnet.ru/dan32449}{Doklady Akad. Nauk SSSR 169(4),
  1966, 765--766 (Russian original on Math-net.ru)}.

\bibitem{PZ93}\bibhead{PZ93}
{\sc Mike Paterson and Uri Zwick}: Shrinkage of {D}e {M}organ formulae under
  restriction.
\newblock {\em Random Struct. Algor.}, 4(2):135--150, 1993.
\newblock [\epfmtdoi{10.1002/rsa.3240040203}]

\bibitem{PR17}\bibhead{PR17}
{\sc Toniann Pitassi and Robert Robere}: Strongly exponential lower bounds for
  monotone computation.
\newblock In {\em Proc. 49th STOC}, pp. 1246--1255. ACM Press, 2017.
\newblock [\epfmtdoi{10.1145/3055399.3055478}]

\bibitem{RM99}\bibhead{RM99}
{\sc Ran Raz and Pierre McKenzie}: Separation of the monotone {NC} hierarchy.
\newblock {\em Combinatorica}, 19(3):403--435, 1999.
\newblock [\epfmtdoi{10.1007/s004930050062}]

\bibitem{RW92}\bibhead{RW92}
{\sc Ran Raz and Avi Wigderson}: Monotone circuits for matching require linear
  depth.
\newblock {\em J. ACM}, 39(3):736--744, 1992.
\newblock [\epfmtdoi{10.1145/146637.146684}]

\bibitem{R90}\bibhead{R90}
{\sc Alexander~A. Razborov}: Applications of matrix methods to the theory of
  lower bounds in computational complexity.
\newblock {\em Combinatorica}, 10(1):81--93, 1990.
\newblock [\epfmtdoi{10.1007/BF02122698}]

\bibitem{R09}\bibhead{R09}
{\sc Ben~W. Reichardt}: Span programs and quantum query complexity: the general
  adversary bound is nearly tight for every {B}oolean function.
\newblock In {\em Proc. 50th FOCS}, pp. 544--551. IEEE Comp. Soc., 2009.
\newblock [\epfmtdoi{10.1109/FOCS.2009.55}]

\bibitem{R11}\bibhead{R11}
{\sc Ben~W. Reichardt}: Reflections for quantum query algorithms.
\newblock In {\em Proc. 22nd Ann. ACM--SIAM Symp. on Discrete Algorithms
  (SODA'11)}, pp. 560--569. SIAM, 2011.
\newblock [\epfmtdoi{10.1137/1.9781611973082.44}]

\bibitem{R14}\bibhead{R14}
{\sc Ben~W. Reichardt}: Span programs are equivalent to quantum query
  algorithms.
\newblock {\em SIAM J. Comput.}, 43(3):1206--1219, 2014.
\newblock [\epfmtdoi{10.1137/100792640}]

\bibitem{S16}\bibhead{S16}
{\sc Igor~Sergeevich Sergeev}: Complexity and depth of formulas for symmetric
  {B}oolean functions.
\newblock {\em Moscow University Mathematics Bulletin}, 71(3):127--130, 2016.
\newblock \href{https://mi.mathnet.ru/vmumm155}{Vestnik Moskov. Univ. Ser.~1:
  Mat. Mekh. 2016(3) 53--57 (Russian original on Math-net.ru)}.
\newblock [\epfmtdoi{10.3103/S0027132216030098}]

\bibitem{SS06}\bibhead{SS06}
{\sc Robert {\v{S}}palek and Mario Szegedy}: All quantum adversary methods are
  equivalent.
\newblock {\em Theory of Computing}, 2(1):1--18, 2006.
\newblock [\epfmtdoi{10.4086/toc.2006.v002a001}]

\bibitem{S71}\bibhead{S71}
{\sc Philip~M. Spira}: On time-hardware complexity tradeoffs for {B}oolean
  functions.
\newblock In {\em Proc. 4th Hawaii International Symposium on System Sciences},
  pp. 525--527. Western Periodicals, 1971.
\newblock
  \href{https://books.google.com/books/about/Proceedings_Fourth_Hawaii_International.html?id=Y0zVjwEACAAJ}{Google
  books}.

\bibitem{S61}\bibhead{S61}
{\sc Bella~Abramovna Subbotovskaya}: Realizations of linear functions by
  formulas using $\vee$, $\&$, $^-$.
\newblock {\em Soviet Math. Doklady}, 2:110--112, 1961.
\newblock \href{https://mi.mathnet.ru/dan24539}{Doklady Akad. Nauk SSSR 136(3),
  553--555, 1961 (Russian original on Math-net.ru)}.

\bibitem{T14}\bibhead{T14}
{\sc Avishay Tal}: Shrinkage of {D}e {M}organ formulae by spectral techniques.
\newblock In {\em Proc. 55th FOCS}, pp. 551--560. IEEE Comp. Soc., 2014.
\newblock [\epfmtdoi{10.1109/FOCS.2014.65}]

\bibitem{Y85}\bibhead{Y85}
{\sc Andrew~Chi{-}Chih Yao}: Separating the polynomial-time hierarchy by
  oracles (preliminary version).
\newblock In {\em Proc. 26th FOCS}, pp. 1--10. IEEE Comp. Soc., 1985.
\newblock [\epfmtdoi{10.1109/SFCS.1985.49}]

\bibitem{Z91}\bibhead{Z91}
{\sc Uri Zwick}: An extension of {Khrapchenko}'s theorem.
\newblock {\em Inform. Process. Lett.}, 37(4):215--217, 1991.
\newblock [\epfmtdoi{10.1016/0020-0190(91)90191-J}]

\end{thebibliography}

\begin{tocauthors}
\begin{tocinfo}[filmus]
 Yuval Filmus\\
 Associate Professor\\
 Technion --- Israel Institute of Technology\\
 The Henry and Marilyn Taub Faculty of Computer Science\\
 Faculty of Mathematics\\
 Haifa 3200003, Israel\\
 yuvalfi\tocat{}cs\tocdot{}technion\tocdot{}ac\tocdot{}il \\   %
 \url{https://yuvalfilmus.cs.technion.ac.il/}      %
\end{tocinfo}

\begin{tocinfo}[meir]
 Or Meir\\
 Associate Professor\\
 Department of Computer Science\\
 University of Haifa\\
 Haifa 3303220, Israel\\
 ormeir\tocat{}cs\tocdot{}haifa\tocdot{}ac\tocdot{}il \\   %
 \url{https://cs.haifa.ac.il/~ormeir/}      %
\end{tocinfo}

\begin{tocinfo}[tal]
 Avishay Tal\\
 Assistant Professor\\
 Department of Electrical Engineering and Computer Sciences\\
 University of California, Berkeley\\
 Berkeley, CA 94720, United States\\
 atal\tocat{}berkeley\tocdot{}edu \\ 
 \url{https://www.avishaytal.org/}
\end{tocinfo}

\end{tocauthors}

\begin{tocaboutauthors}
\begin{tocabout}[filmus] 
\textsc{Yuval Filmus} received his \phd\ from the University of Toronto under the supervision of Toni Pitassi. Subsequently, he held postdoctoral positions at the Simons Institute for the Theory of Computing and the Institute for Advanced Study.
He is interested in Boolean function analysis, complexity theory, and combinatorics.
\end{tocabout}
\begin{tocabout}[meir]
\textsc{Or Meir} received his \phd\ from the Weizmann Institute 
of  %
Science under the supervision of Oded Goldreich. Subsequently, he held postdoctoral positions at Stanford University, the Institute %
for  %
Advanced Study, and the Weizmann Institute %
of  %
Science.
He is interested in complexity theory, and in particular in circuit complexity, communication complexity, coding theory, probabilistic proof systems, and the theory of pseudorandomness.
\end{tocabout}
\begin{tocabout}[tal]
\textsc{Avishay Tal} received his M.\,Sc. from the Technion under the supervision of Amir Shpilka, and his \phd\ from the Weizmann Institute of Science under the supervision of Ran Raz. Subsequently, he held postdoctoral positions at the Institute for Advanced Study, the Simons Institute for the Theory of Computing, and Stanford University.
He is interested in Boolean function analysis, complexity theory and combinatorics, and in particular in circuit complexity, query complexity, computational learning theory, quantum complexity, and the theory of pseudorandomness.
\end{tocabout}

\end{tocaboutauthors}

\end{document}